\documentclass[10pt,a4paper]{article}
\usepackage{amsmath,amssymb,amsfonts,amsthm}
\usepackage{centernot}
\usepackage{bbm} 
\usepackage{caption}
\usepackage{graphicx}
\usepackage{float}
\usepackage{natbib}
\usepackage{mathrsfs}
\usepackage{epstopdf}
\usepackage{longtable}
\usepackage{adjustbox}
\usepackage{multirow}
\usepackage{soul}
\usepackage{color}
\usepackage[pdfpagemode=UseNone,bookmarksopen=false,colorlinks=true]{hyperref}
\hypersetup{linkcolor=blue, citecolor=blue,
  filecolor=blue, urlcolor=red}

\allowdisplaybreaks

\newtheorem{assumption}{Assumption}

\newtheorem{corollary}{Corollary}[section]
\newtheorem{lemma}{Lemma}[section]
\newtheorem{result}{Result}[section]
\newtheorem{remark}{Remark}

\newtheorem{theorem}{Theorem}[section]
\theoremstyle{remark}
\newtheorem{definition}{Definition}[section]

\begin{document}
\title{\bf Linear Regression: Inference Based on Cluster Estimates } \vspace{1cm}			
		
\author{Subhodeep Dey\thanks{ Theoretical Statistics and Mathematics Unit,  Indian Statistical Institute,  Kolkata, India} \; Gopal K. Basak\thanks{ Theoretical Statistics and Mathematics Unit,  Indian Statistical Institute,  Kolkata, India} \;   Samarjit Das\thanks{ Economic Research Unit, Indian Statistical Institute, Kolkata, India} \\ Indian Statistical Institute }  
\date{}
\maketitle
			
\begin{abstract}
This article proposes a novel estimator for regression coefficients in clustered data that explicitly accounts for within-cluster dependence. We study the asymptotic properties of the proposed estimator under both finite and infinite cluster sizes. The analysis is then extended to a standard random coefficient model, where we derive asymptotic results for the average (common) parameters and develop a Wald-type test for general linear hypotheses. We also investigate the performance of the conventional pooled ordinary least squares (POLS) estimator within the random coefficients framework and show that it can be unreliable across a wide range of empirically relevant settings. Furthermore, we introduce a new test for parameter stability at a higher (superblock; Tier 2, Tier 3,...) level, assuming that parameters are stable across clusters within that level. Extensive simulation studies demonstrate the effectiveness of the proposed tests, and an empirical application illustrates their practical relevance.


\end{abstract}
\vspace{0.2cm}
\noindent {\bf Key Words:}  Least Squares, Cross-section data, cluster dependence.
\vspace{0.2cm}

\noindent {\bf JEL Classification: C01, C13 }
\newpage


\section{Introduction}

Applied researchers in economics often use data consisting of independent clusters of dependent random variables. \cite{moulton1986random} and \cite{moulton1990illustration} provide several canonical examples illustrating how clustering naturally arises in empirical research. It often occurs when the sampling design involves selecting random groups, such as schools, households, or towns, and surveying all units within each selected group. Clustering may also result from cost-reduction strategies in data collection, such as interviewing multiple households located on the same block of the city. Furthermore, even under a Simple Random Sampling scheme, dependence can emerge due to unobserved effects shared across units, for instance, a latent factor common to all households within a given state.  The fundamental nature of cluster data is that individuals within a cluster are generally interdependent, and individuals between clusters are expected to be independent.

A plethora of literature has been developed to facilitate robust inference through appropriate correction of standard errors (see \cite{white1984asymptotic}, \cite{arellano1987computing}, \cite{conley1999gmm}, \cite{wooldridge2003cluster}, \cite{cameron2005microeconometrics}, \cite{wooldridge2006cluster}, \cite{cameron2008bootstrap}, \cite{cameron2015practitioner}, \cite{mackinnon2019and}, \cite{abadie2020sampling}, \cite{canay2021wild}, \cite{hansen2022econometrics}, \cite{abadie2023should},  \cite{chiang2023genuinely}, \cite{mackinnon2023cluster}). The literature mostly assumes that the cluster size ($N,$  say) is finite and the number of clusters ($G,$ say) diverges to infinity. Such a framework limits the scope of applicability of the associated robust standard error adjustments. In practice, three different asymptotic regimes may arise: (1) $G$ is finite, and $N \rightarrow \infty $. (2)  $N \rightarrow \infty, $ and  $ G \rightarrow \infty,$ (3) $N$ is finite, but $ G \rightarrow \infty.$
Scenario (1) typically leads to inconsistency of the LSE and invalidity of the Central Limit Theorem.\footnote{\cite{andrews2005cross} demonstrates that the LSE is inconsistent when the sample size grows with a fixed number of clusters.} \footnote{\cite{bester2011inference} study asymptotic properties under fixed $G$ and large $N$, invoking a “limited memory” assumption under which dependence vanishes with distance, resulting in sparse error covariance matrices.} \footnote{\cite{ibragimov2016inference} proposes inference strategies for small numbers of clusters, effectively assuming Gaussianity of model errors within each cluster.}
The same inferential challenges arise when both $N\to\infty$ and $G\to\infty$, without additional restrictions on their relative magnitudes. Scenario (3), where the cluster size is fixed and the number of clusters increases, has been extensively studied in the above references.

However, the literature remains limited for scenario (2), with two notable exceptions: \cite{hansen2019asymptotic} and \cite{djogbenou2019asymptotic}. Both works assume some homogeneity in cluster sizes, excluding settings where a few clusters are disproportionately large while the majority are relatively small. In such contexts, the POLS (Pooled Ordinary Least Squares) estimator used in these studies may perform poorly. For instance, \cite{djogbenou2019asymptotic} found some undesirable properties when a cluster is dominant, having half of the observations. These cases are not merely theoretical. For example, studies on state laws and corporate governance in the United States often encounter clustering at the state level, with Delaware alone housing nearly half of all corporations (\cite{spamann2019inference}). To illustrate further, we consider examples from India’s Annual Survey of Industries (ASI) for the years 2017–18 and 2019–20. The ASI data for regions such as  Tripura, Maharashtra, and West Bengal are particularly relevant, where 4-digit NIC codes (NIC4s) may serve as natural cluster identifiers.

\textit{Example 1}: In the state of Tripura, one extremely large cluster exists among a total of 50 clusters. For the year 2017–18, the largest cluster has a size of 323, whereas the average size of the remaining clusters is just 6. Tripura is considered one of the least industrialized states in India.

\textit{Example 2}: In Maharashtra, 17 of 162 clusters each contain more than 500 units (based on 2017–18 data, where 4-digit NICs are treated as clusters). The average size of the remaining clusters is 95, and 56 clusters comprise only 1 to 25 units. Maharashtra is recognized as one of the most industrialized states in India.

\textit{Example 3}: In West Bengal, 3 out of 148 clusters each contain more than 500 units (based on 2017–18 data, with 4-digit NICs considered as clusters). The average size of the remaining clusters is 55, and 70 clusters have sizes ranging from 1 to 25 units. West Bengal is regarded as a moderately industrialized state.

These examples illustrate cases in which a finite number of clusters are disproportionately large compared to the others.

In this paper, we demonstrate that the existence of a few large clusters (dominating clusters) can cause the POLS to be inconsistent. We propose a new estimator and demonstrate that it is consistent even when a few large clusters exist. We then extend the basic cluster model to accommodate random variations of the underlying parameters. Random coefficients (RC) models have been widely employed to accommodate unobserved heterogeneity in regression relationships across observational units. \cite{swamy1970efficient} developed estimation methods for linear random coefficients models.  As he rightly pointed out, ``....  it is unlikely that interindividual differences observed in a cross-sectional sample can be explained by a simple regression equation with a few independent variables. In such situations, the coefficient vector of a regression model can be treated as random to account for interindividual heterogeneity".  We show that the proposed estimator is also consistent under RC model as well. We study the asymptotic properties of this proposed estimator and demonstrate the inferential methods. 

Swamy's work extended to a panel allowing regression coefficients to vary across individuals and firms while maintaining a common mean structure. One valuable review on the parameter stability test is available in  \cite{pesaran2008testing}. This paper itself develops a test for slope homogeneity by standardizing  Swamy’s dispersion statistic and is valid under large \(N\) and \(T\) panels for testing the null of homogeneous coefficients across cross-sectional units. \cite{breitung2016lagrange} develop Lagrange multiplier tests for slope homogeneity in panel data that generalize the Breusch--Pagan framework while accommodating heteroskedasticity, serial correlation, and non-normality. 

In the context of pure cross-sectional data, it appears that there is no test for parameter stability. The existing literature does not allow underlying parameters to vary across clusters. Here, we develop a new test to examine whether the parameter vector is stable across clusters or not, at the superblock level. Here, it may be mentioned that the new test for parameter stability is developed for the higher (superblock) level, assuming that parameters are stable across clusters within the superblock. From a practitioner's point of view, this may not be a restrictive assumption. Consider the following examples.

\textit{Example 4}: Suppose we want to estimate the Engel curve for India. India is an extremely heterogeneous country, divided into various, somewhat homogeneous, states based on common language, culture, value system, and food habits.  Again, each state is divided into several districts primarily based on some common features (homogeneity). Again, districts are divided into subdivisions.  Consider the Household Consumption Expenditure Survey (HCES) of India for the year 2022--23, with households as individual observations, and First Stage Units (FSUs) as clusters. FSU contains around 5-18 households. For rural India, FSUs are from villages. For the urban area,  FSUs are collected from small communes.  We can consider the districts or the states as superblocks. It is natural to assume that the parameter vector of the Engel curve would vary across states; parameters may not vary across FSUs within a subdivision or even within districts. 

If one wants to test for parameter stability at the state level, we assume that all the clusters (FSUs) in each state have stable parameters. On the other hand, if one wants to test for parameter stability at the district level, we assume that all the clusters (FSUs) in each district have stable parameters. 

\textit{Example 5}: In India, the school (Age 16-18) dropout rate or the proportion of dropouts is extremely high, particularly in recent years. The dropout decision (Binary) is expected to be determined by variables like family income (Consumption), education level of the head of the household, Household size, religious status, etc.  The education system is mainly run by the state governments. It is observed that there is a significant amount of variation in dropout rate across states or regions.  Schools may be treated as basic clusters.  Districts or states may be treated as superblocks.

In summary, in this paper, we propose a new estimator that is consistent under both the basic linear model and the random coefficients model. The proposed estimator remains valid as \( G \to \infty \), regardless of the magnitude or imbalance of cluster sizes.  We show that the proposed estimator does have a few advantages over the standard estimator as studied by  \cite{moulton1986random}, \cite{moulton1990illustration}, \cite{cameron2008bootstrap}, \cite{cameron2015practitioner}, \cite{wooldridge2003cluster}, \cite{wooldridge2006cluster}, \cite{cameron2005microeconometrics}, and \cite{hansen2019asymptotic}. In particular, we establish that the traditional pooled ordinary least squares (POLS) estimator is inconsistent for the random coefficients model under various forms of cross-sectional dependence when a small number of large clusters dominate the sample. By contrast, the proposed estimator remains consistent in such settings. We further develop a Wald-type test statistic based on the proposed estimator for testing general linear hypotheses under both the basic and random coefficients models. In addition, we propose a valid test for assessing parameter constancy across superblocks, which represent higher-level clustering structures. Finally, we present simulation results and empirical applications that support our theoretical findings.

The remainder of this paper is organized as follows. Section \ref{section:model}  presents the basic model and its associated assumptions under the parameter constancy framework. This section introduces the proposed estimator and establishes its asymptotic properties under various forms dependence. Section \ref{section:model2} extends the analysis by allowing the regression parameters to vary across clusters and provides the corresponding asymptotic results for this generalized model. We also present results for the conventional pooled ordinary least squares (POLS) estimator under the varying-parameter specification. Section \ref{section:analysis} reports the simulation study and empirical analysis using real-world data. Section \ref{section:conclusion} concludes the paper. All formal proofs are provided in the Appendix.

\section{Model with Constant Parameters}\label{section:model}

Let $Y_{gi}$ be the observation on the $i^{th}$ individual of the $g^{th}$ cluster, and suppose that it is generated according to the linear model for clustered data as 
\begin{equation} \label{eq: model}
Y_{gi}=X_{gi}^{\prime} \beta +\epsilon_{gi}, \; \; i=1,2,\ldots,N_g; \; g=1,2,\ldots,G, 
\end{equation} 
where $X_{gi}$ is a $k\times1$ vector of observed regressors (including categorical variables) on the  $i^{th}$  individual of the $g^{th}$ cluster and $\epsilon_{gi}$  is the corresponding random error.
We consider one-way clustering only. Stacking observations within a cluster yields the following model.  
\begin{equation} \label{eq: model1}
 Y_g=X_g\beta +\epsilon_g ,\;\;  g=1,2,\ldots ,G,
\end{equation}
 where $Y_g$  is a $N_g\times 1$  vector, $X_g$  is a  $N_g \times k$  matrix and $\beta$  is a $k\times 1$  vector.  Note that for every $g=1,2,\ldots,G,$ the variance-covariance matrix of ${\epsilon}_g$ is defined as $\Omega_g,$  which is an $N_g\times N_g$ positive definite matrix with uniformly bounded diagonal elements and is independent of $X$.\footnote{ Conditional heteroskedasticity may be allowed with some assumptions on $X.$}

 Before presenting the main results, we establish some fundamental assumptions and make key observations. 
 
\begin{assumption}\label{assump:exogeneity}
    We assume that $X=\left(X_1^\prime,X_2^\prime,\ldots,X_G^\prime\right)^\prime$ are strictly exogenous with respect to $\epsilon_{g}$, $\forall g,$ i.e., 
$$ \mathbb{\mathbb{E}}\left[\epsilon_{g}|X \right] = 0, \;\;  \forall g.$$
\end{assumption}

\begin{definition} \label{def:exact order and order uniformly}
    (i)     Let $\left\{z_n\right\}$ be a sequence of real numbers and $\left\{a_n\right\}$ be a sequence of positive real numbers.
    
    (a) By $z_n=O_e\left(a_n\right),$ we mean that for some $0<c_1\leq c_2<\infty,$
    \[ c_1 \leq \left| \frac{z_n}{a_n} \right| \leq c_2, \; \text{ for all } n. \]
    
    (b)  By $z_n=O\left(a_n\right),$ we mean that for some $0<c<\infty,$
    \[
     \left| \frac{z_n}{a_n} \right| \leq c, \; \text{ for all } n.
    \]  
    
    (ii) Let $\left\{z_{N_g}\right\}$ be a doubly indexed sequence of real numbers and $\left\{a_{N_g}\right\}$ be a doubly indexed sequence of positive real numbers, where $N_g$ is the cluster size for the $g^{th}$ cluster, which may increase to $+\infty.$
    
    (a) By $z_{N_g}=O_e\left(a_{N_g}\right)$ uniformly in $g,$ we mean that for some $0<c_1\leq c_2<\infty,$
    \[ c_1 \leq \inf_{g,N_g}\;\left| \frac{z_{N_g}}{a_{N_g}} \right| \leq \sup_{g,N_g}\;\left| \frac{z_{N_g}}{a_{N_g}} \right| \leq c_2. \]

     (b) By $z_{N_g}=O\left(a_{N_g}\right)$ uniformly in $g,$ we mean that for some $0<c<\infty,$
    \[ \sup_{g,N_g}\;\left| \frac{z_{N_g}}{a_{N_g}} \right| \leq c. \]
\end{definition}

\begin{definition} \label{def:stochastically bounded uniformly}
    (i) Let $\left\{Z_n\right\}$ be a sequence of random variables and $\left\{a_n\right\}$ be a sequence of positive real numbers.
    By $Z_n=O_P\left(a_n\right),$ we mean that for every $\varepsilon > 0$, there exist $M_\varepsilon > 0$ and ${L_\varepsilon} \in \mathbb{N}$ such that 
     \[ \mathbb{P}\left(\left|\frac{Z_n}{a_n}\right| > M_\varepsilon\right) < \varepsilon, \; \text{ for all } n \geq {L_\varepsilon}. \]

    (ii)    
    Let $\left\{Z_{N_g}\right\}$ be a doubly indexed sequence of random variables and $\left\{a_{N_g}\right\}$ be a doubly indexed sequence of positive real numbers, where $N_g$ is the cluster size for the $g^{th}$ cluster, which may increase to $+\infty.$
     By $Z_{N_g}=O_P\left(a_{N_g}\right)$ uniformly in $g,$ we mean that for every $\varepsilon > 0$, there exists $M_\varepsilon > 0$ such that 
     \[ \sup_{g,N_g}\;\mathbb{P}\left(\left|\frac{Z_{N_g}}{a_{N_g}}\right| > M_\varepsilon\right) < \varepsilon. \] 
\end{definition}
Definition \ref{def:stochastically bounded uniformly}(i) essentially means that for every $\varepsilon > 0$, there exists $M_\varepsilon > 0$ such that $\underset{n}{\operatorname{sup}}\;\mathbb{P}\left(\left|\frac{Z_{n}}{a_{n}}\right| > M_\varepsilon\right) < \varepsilon.$

We now characterize the dependence structure of the error terms, which is fundamental to our analysis.

\begin{definition} \label{def:dependency rates}
We define three possible kinds of cross-sectional dependence. Let us denote $\Omega_g=\mathbb{\mathbb{E}}\left[\epsilon_g \epsilon_g^\prime \big | X\right]$ for a fixed cluster $g,$ and the $g^{th}$ cluster contains $N_g$ observations.
\begin{enumerate}
    \item \textbf{ Strong Dependence:} Dependence across the $N_g$ individuals is said to be strong, when $\lambda_{\max}(\Omega_{g})=O_e(N_{g}).$ 
    
    \item \textbf{ Semi-strong Dependence:} Dependence across the $N_g$ individuals is said to be semi-strong or moderate, when the following condition is met: $\lambda_{\max}(\Omega_{g})=O_e(h(N_{g})),$ where $h(N_{g})\uparrow\infty,$
    as $N_{g}\uparrow\infty,$ but $\frac{h(N_{g})}{N_{g}}\rightarrow0,$
    as $N_{g}\rightarrow\infty.$

    \item \textbf{ Weak Dependence:} Dependence across the $N_g$ individuals is said
    to be weak, when $\lambda_{\max}(\Omega_{g})=O_e(1).$
\end{enumerate}
\end{definition}

Strong dependence implies that (almost) all individuals are correlated/interconnected. Semi-strong dependence implies that the number of dependent pairs increases with sample size. Weak dependence implies that all the eigenvalues of $\Omega_{g}$ are finite. Independence is regarded as weak dependence. Weak dependence may also hold for dependence that decays sufficiently fast as observations become more distant according to some measure. Further details and examples of these types of cross-sectional dependence are described in \cite{basak2018understanding}.
Consider a scenario where we have clustered data with $G$ clusters, where the $g^{th}$ cluster contains $N_g$ observations. If the $g^{th}$ cluster exhibits strong or semi-strong dependence, then $N_g$ must diverge to $+\infty,$ along with $G\to\infty.$ Otherwise, if $N_g$ is finite, and $G\to\infty,$ then one will always be left with weakly dependent observations.

\begin{remark}\label{rem:explanation of strong, semi-strong, weak dependence}
    Throughout this paper, by strong dependence, we refer to a scenario in which $O(G)$ clusters exhibit strong dependency, while other clusters may be semi-strongly or weakly dependent. Note that the number of strongly dependent clusters may be, say, $\frac{G}{2}$ or $\frac{G}{3},$ which are also included in the above case. Semi-strong dependence means that $O(G)$ clusters are semi-strongly dependent, with no strongly dependent clusters present, although weakly dependent clusters may still exist. Weak dependence, on the other hand, describes a situation where all clusters are weakly dependent, with no strong or semi-strongly dependent clusters involved.
\end{remark}

\begin{remark}
    For simplicity, we have assumed that the rate of semi-strong dependence is uniform across clusters, i.e., $h_g(N_g)\simeq h(N_g), \; \forall g,$ although it can be generalized.
\end{remark}

\begin{assumption}\label{assump:error independence}
    The independence of error between clusters is assumed so that for each $g$, 
\begin{align*}
    \mathbb{\mathbb{E}}\left[\epsilon_{g} \epsilon^\prime_{g_1} \big|X \right]  & = 0,\; \forall g_1\neq g \quad \text{ and } \\
    \mathbb{\mathbb{E}}\left[\epsilon_{g} \epsilon^\prime_{g} \big|X \right]  &  = \Omega_g,
\end{align*} 
 This structure allows for arbitrary within-cluster heteroskedasticity and correlation among individuals belonging to the same cluster. As per the Definition \ref{def:dependency rates}, further assume that 
\[
\lambda_{\max}\left(\Omega_{g}\right)=\begin{cases}
O_{e}\left(N_{g}\right), & \text{for strong dependence}\\
O_{e}\left(h(N_{g})\right), & \text{for semi-strong dependence}\\
O_{e}\left(1\right), & \text{for weak dependence}
\end{cases}
\]
uniformly in $g,$ as in Definition \ref{def:exact order and order uniformly}.
\end{assumption}

Our use of precise definitions and order notation is designed to systematically track the numerous fixed constants that arise in the assumptions, theorems, and their corresponding proofs.

\begin{definition} \label{def:norms}
     We use two norms, viz., the maximum eigenvalue norm and the trace norm. The maximum eigenvalue norm for any matrix $A$ is defined as $ \Vert A\Vert_e= \substack{{\text{ max }}\\{z\neq 0}}\frac{\parallel Az \parallel}{\parallel z \parallel},$ where for a $n\times 1$ vector $y,$ $\Vert y \Vert=\left(y^\prime y\right)^{1/2}.$ In particular, for any nonnegative definite matrix $A,$ the maximum eigenvalue norm becomes $ \Vert A\Vert_e= \substack{{\text{ max }}\\{z:z^\prime z=1}}z^\prime Az.  $ The trace norm is defined as  $\Vert A\Vert=\left[\operatorname{tr}(A^{\prime}A)\right]^{\frac{1}{2}}.$ This is also called the Frobenius norm with $ F=\frac{1}{2}$.
\end{definition}

\begin{definition}
\textbf{(a) Nearly balanced clusters:}  $\min_g\{N_g\}/\max_g\{N_g\}$  is bounded away from zero.

\textbf{(b) Unbalanced clusters}:  $\min_g\{N_g\}/\max_g\{N_g\}=o(1).$
\end{definition}

\subsection{The Proposed Estimator}

We estimate $\beta$ separately for each cluster and then combine all these $G$ estimates by taking the simple average. The estimate for $g^{th}$ cluster is
$\hat{\beta}_g=\left(X_g^\prime X_g\right)^{-1} X_g^\prime Y_g, \; g=1, 2\ldots ,G.$
Then, our proposed estimator is the average of all these estimates, and it is defined as 
\begin{equation} \label{eq: proposed estimator}
    \hat{\bar{\beta}}=\frac{1}{G} \sum_{g=1}^{G}\hat{\beta}_g.
\end{equation}

The variance-covariance matrix of $\hat{\bar\beta}$ under model \eqref{eq: model1} is given by
\begin{equation} \label{eq: V}
        V=\frac{1}{G^{2}}\sum_{g=1}^G \mathbb{E}\left[\left(X_{g}^{\prime}X_{g}\right)^{-1}X_{g}^{\prime}\Omega_{g}X_{g}\left(X_{g}^{\prime}X_{g}\right)^{-1}\right].
\end{equation}

\begin{assumption}\label{assump:Xg'Xg/Ng inverse is bounded}
For some fixed $0<M<\infty,$ assume that for some $p>1,$ 
\[
\underset{g}{sup}\;\mathbb{E}\left[\left(\operatorname{tr}\left(\left(\frac{X_{g}^{\prime}X_{g}}{N_{g}}\right)^{-1}\right)\right)^{p}\right]\leq M.
\]
\end{assumption}

Assumption \ref{assump:Xg'Xg/Ng inverse is bounded} imposes a uniform moment bound on the inverse of the normalized design matrix across clusters. Basically, it requires that the trace of $\left(X_g'X_g/N_g\right)^{-1}$ admits a finite $p$-th moment, uniformly in $g$. The assumption implies that the eigenvalues of $X_g'X_g/N_g$ are uniformly bounded away from zero in expectation.

\begin{theorem}
\label{thm:consistency}
Under Assumptions \ref{assump:exogeneity}, \ref{assump:error independence} and \ref{assump:Xg'Xg/Ng inverse is bounded}, for strong, semi-strong or weak dependence, and for any cluster sizes,  $\hat{\bar\beta}$ is consistent for $\beta$.
\end{theorem}

The proof of Theorem \ref{thm:consistency} is given in the Appendix. We now provide a lemma that will clarify the subsequent results.

\begin{lemma}\label{lemma:Exact order summability}
    Let $\left\{z_{N_g}\right\}$ and $\left\{a_{N_g}\right\}$ be doubly indexed sequences of positive real numbers such that $z_{N_g}=O_e\left(a_{N_g}\right)$ uniformly in $g,$ where $N_g$ is the cluster size for the $g^{th}$ cluster. Then for any $G\geq 1,$  $\sum_{g=1}^G z_{N_g}=O_e\left(\sum_{g=1}^G a_{N_g}\right).$ 
\end{lemma}
The proof of this lemma is given in the Appendix. Analogous results hold if we replace $O_e(\cdot)$ by $O(\cdot).$

Note that $\hat{\bar{\beta}}-\beta=\frac{1}{G}\sum_{g=1}^{G}\left(\frac{X_{g}^{\prime}X_{g}}{N_{g}}\right)^{-1}\frac{X_{g}^{\prime}\epsilon_{g}}{N_{g}}.$ Clearly, $\hat{\bar{\beta}}$ is an unbiased estimator of $\beta.$ From the proof of Theorem \ref{thm:consistency} given in the Appendix, we have 
\[
\mathbb{E}\left\Vert \hat{\bar{\beta}}-\beta \right\Vert^2 \leq \frac{1}{G^{2}}\sum_{g=1}^{G}\mathbb{E}\left[\operatorname{tr}\left(\left(\frac{X_{g}^{\prime}X_{g}}{N_{g}}\right)^{-1}\right)\right]\frac{\lambda_{max}\left(\Omega_{g}\right)}{N_{g}}.
\]
Then, by Assumption \ref{assump:error independence}, Assumption \ref{assump:Xg'Xg/Ng inverse is bounded}, and Lemma \ref{lemma:Exact order summability}, we obtain
\begin{eqnarray}
    \mathbb{E}\left\Vert \hat{\bar{\beta}}-\beta \right\Vert^2 = \begin{cases}
O\left(\frac{1}{G}\right), & \text{for strong dependence}\\
O\left(\frac{1}{G^{2}}\sum_{g=1}^{G}\frac{h(N_g)}{N_{g}}\right), & \text{for semi-strong dependence}\\
O\left(\frac{1}{G^{2}}\sum_{g=1}^{G}\frac{1}{N_{g}}\right), & \text{for weak dependence}
\end{cases}
\end{eqnarray}
\begin{corollary}
     The above discussion yields that $\hat{\bar\beta}$ is $\sqrt{G}$-consistent for strong dependence. $\hat{\bar\beta}$ is $\sqrt{G}\cdot \bar{h}^{-\frac{1}{2}}_1$-consistent for semi-strong dependence, and $\sqrt{G}\cdot \bar{m}^{-\frac{1}{2}}$-consistent for weak dependence, where $\bar{h}_{1}=\frac{1}{G}\sum_{g=1}^G \frac{h(N_g)}{N_g}$ and $\bar{m}=\frac{1}{G}\sum_{g=1}^G \frac{1}{N_g}.$  
\end{corollary}

\begin{assumption}\label{assump:existence of moment}
 Assume that for some $0<c<\infty,$ and for some $p>1,$ 
\[
\mathbb{E}\left[ |  \epsilon_{gi}  | ^{2p} \Big|  X \right]\leq c,\quad \text{ for all } i, g \text{ and uniformly in } X.
\]
\end{assumption}

The above assumption essentially ensures the existence of moments beyond the second order, thereby controlling the tail behavior of the error terms $\epsilon_{gi}.$ It also establishes the necessary regularity conditions for deriving the limiting distribution of the estimator.

\begin{assumption}\label{assump:lambda max V_G=lambda min V_G}
Assume that 
\[
V_{G}=\frac{1}{G}\sum_{g=1}^{G}\mathbb{E}\left[\left(X_{g}^{\prime}X_{g}\right)^{-1}X_{g}^{\prime}\Omega_{g}X_{g}\left(X_{g}^{\prime}X_{g}\right)^{-1}\right]
\]
 is finite and positive definite. Also, assume that 
\[
O_{e}\left(\lambda_{min}\left(V_{G}\right)\right)=O_{e}\left(\lambda_{max}\left(V_{G}\right)\right)=\begin{cases}
O_{e}\left(1\right), & \text{for strong dependence}\\
O_{e}\left(\frac{1}{G}\sum_{g=1}^{G}\frac{h\left(N_{g}\right)}{N_{g}}\right), & \text{for semi-strong dependence}\\
O_{e}\left(\frac{1}{G}\sum_{g=1}^{G}\frac{1}{N_{g}}\right), & \text{for weak dependence}
\end{cases}
\]
as in Definition \ref{def:exact order and order uniformly}.  
\end{assumption}

For weak dependence, the above follows immediately because both the smallest and largest eigenvalues satisfy $\lambda_{min}\left(\Omega_{g}\right)=O_{e}(1)=\lambda_{max}\left(\Omega_{g}\right)$
uniformly in $g.$ Also, by applying Lemma \ref{lemma:Exact order summability}, we can rigorously characterize the orders of $\lambda_{\min}\left(V_{G}\right)$ and $\lambda_{\max}\left(V_{G}\right),$ for the three kinds of dependence.

\begin{theorem}\label{thm:normal}
Under Assumptions  \ref{assump:exogeneity}, \ref{assump:error independence}, \ref{assump:Xg'Xg/Ng inverse is bounded}, \ref{assump:existence of moment} and \ref{assump:lambda max V_G=lambda min V_G}, for strong dependence
\[
\sqrt{G}V_{G}^{-1/2}\left(\hat{\bar{\beta}}-\beta\right)\xrightarrow{d}N_{k}\left(0,I_{k}\right), \; \text{ as } G\to\infty.
\]
In fact, for semi-strong dependence, the asymptotic normality holds if for some $p>1,$ $\frac{G}{\left[\sum_{g=1}^{G}\frac{h\left(N_{g}\right)}{N_{g}}\right]^{p}}\to0,$
as $G\to\infty.$ For weak dependence, the asymptotic normality holds
if for some $p>1,$ $\frac{G}{\left[\sum_{g=1}^{G}\frac{1}{N_{g}}\right]^{p}}\to0,$ as $G\to\infty.$
\end{theorem}

The proof of Theorem \ref{thm:normal} is given in the Appendix.

\begin{remark} \label{rem:conditions for normality of beta hat}
    Note that for nearly balanced clusters case with $N_{g}\simeq N,$ for strong dependence, the Liapounov condition holds trivially, since it reduces to $O\left(G^{1-p}\right),$ which goes to $0,$ as $G\to\infty,$ for some $p>1.$ For semi-strong dependence with $h(N)\to\infty$ and $N\to\infty,$ the condition reduces to $O\left(\left(\frac{N}{h(N)}\right)^{p}G^{1-p}\right)\to0,$ as $G\to\infty,$ for some $p>1.$ For weak dependence with $N\to\infty,$ the condition reduces to $O\left(N^{p}G^{1-p}\right)\to0,$ as $G\to\infty,$ for some $p>1.$ These are somewhat similar conditions to those used for the asymptotic normality of the POLS estimator.
\end{remark}

\begin{remark}
    If we assume that for some $p>1,$ 
\[
\frac{G}{\left[\lambda_{min}\left(\sum_{g=1}^{G}\mathbb{E}\left[\left(X_{g}^{\prime}X_{g}\right)^{-1}X_{g}^{\prime}\Omega_{g}X_{g}\left(X_{g}^{\prime}X_{g}\right)^{-1}\right]\right)\right]^{p}}\to0,\text{ as }G\to\infty.
\]
then, the asymptotic normality of $\hat{\bar{\beta}}$ holds without Assumption \ref{assump:lambda max V_G=lambda min V_G}. Although this assumption is relatively weak, more specific conditions are imposed in Assumption \ref{assump:lambda max V_G=lambda min V_G}.
\end{remark}

\begin{assumption}\label{assump:Order of second moment of Xg'Omega_g Xg}
Assume that 
\[
\mathbb{E}\left\Vert X_{g}^{\prime}\epsilon_{g}\epsilon_{g}^{\prime}X_{g}-X_{g}^{\prime}\Omega_{g}X_{g}\right\Vert ^{2}=\begin{cases}
O\left(N_{g}^{4}\right),  & \; \text{for strong dependence}\\
O\left(N_{g}^{2}h^{2}\left(N_{g}\right)\right), & \; \text{for semi-strong dependence}\\
O\left(N_{g}^{2}\right), & \; \text{for weak dependence}
\end{cases}
\]
uniformly in $g.$
\end{assumption}

This assumption controls the second-order variability of the cluster-level quadratic form $X_g^{\prime}\epsilon_g\epsilon_g^{\prime}X_g$ around its mean $X_g^{\prime}\Omega_g X_g$ as the cluster size grows. 
The deviation is allowed to increase with \(N_g\), but at different rates depending on whether within-cluster dependence is strong, semi-strong, or weak.

\begin{lemma}\label{lemma: sum h^2/N^2 to 0}
 Let $\{N_g\}_{g=1}^{\infty}$ be a sequence of positive integers and let $h$ be a function of $N_g$ such that $h(N_g)\uparrow \infty,$ as $N_g\uparrow \infty,$ but $h(N_g) = o(N_g).$ Then,
    \[
    \frac{\sum_{g=1}^{G} \frac{h^{2}(N_{g})}{N_{g}^{2}}}{\left( \sum_{g=1}^{G} \frac{h(N_{g})}{N_{g}} \right)^{2}} \to 0 \; \text{ and } \; \frac{\sum_{g=1}^{G} \frac{1}{N_{g}^{2}}}{\left( \sum_{g=1}^{G} \frac{1}{N_{g}} \right)^{2}} \to 0, \; \text{ as } G \to \infty
    \]
 holds for (i) the nearly balanced case and (ii) the unbalanced case with $O(G)$ clusters being nearly balanced.
\end{lemma}

\begin{theorem}\label{thm: consistency of V hat}
Under Assumptions  \ref{assump:exogeneity}, \ref{assump:error independence}, \ref{assump:Xg'Xg/Ng inverse is bounded}, \ref{assump:existence of moment}, \ref{assump:lambda max V_G=lambda min V_G}, \ref{assump:Order of second moment of Xg'Omega_g Xg}, and Lemma \ref{lemma: sum h^2/N^2 to 0}, for nearly balanced clusters, under strong, semi-strong, or weak dependence, $\hat{V}_G $ is consistent for $V_G$, where 
\[
\hat{V}_G = \frac{1}{G}\sum_{g=1}^{G}\left(X_{g}^{\prime}X_{g}\right)^{-1}X_{g}^{\prime}e_{g}e_{g}^{\prime}X_{g}\left(X_{g}^{\prime}X_{g}\right)^{-1}
\]
with $e_{g}=Y_{g}-X_{g}\hat{\bar{\beta}}.$ 
For the unbalanced case, the consistency holds for: (i) strong dependence generally, (ii) weak and semi-strong dependence, if $O(G)$ clusters are nearly balanced.
\end{theorem}
The proof of Theorem \ref{thm: consistency of V hat} is given in the Appendix.

\begin{remark}
    The consistency of $\hat{V}_G$ holds for weak or semi-strong dependence for unbalanced clusters, even if $O(G)$ clusters (say, $G/2$ clusters) are extremely large and $O(G)$ clusters (say, $G/2$ clusters) are nearly balanced.
\end{remark}

\begin{remark}\label{Rem:asympnorm and consistency of V for unbalanced}
\begin{enumerate}
    \item   The consistency of $\hat{V}_G$ is preserved in the unbalanced case under strong dependence.
    \item   For the unbalanced case, if all the clusters are semi-strongly dependent then it requires $\frac{\sum_{g=1}^G  \frac{h^2(N_g)}{N_g^2}}{\left(\sum_{g=1}^G  \frac{h(N_g)}{N_g}\right)^2}\rightarrow 0$ for the consistency of $\hat{V}_G.$
    \item   For the unbalanced case, if all the clusters are weakly dependent then it requires $\frac{\sum_{g=1}^G  \frac{1}{N_g^2}}{\left(\sum_{g=1}^G  \frac{1}{N_g}\right)^2}\rightarrow 0$ for the consistency of $\hat{V}_G.$
\end{enumerate}
\end{remark}

\begin{remark}
Note that under Assumption \ref{assump:Order of second moment of Xg'Omega_g Xg}, the Liapounov condition for the asymptotic normality of $\hat{\bar{\beta}}$ with $p=2$ under semi-strong dependence boils down to 
    \[
     \frac{\sum_{g=1}^{G} \frac{h^{2}(N_{g})}{N_{g}^{2}}}{\left( \sum_{g=1}^{G} \frac{h(N_{g})}{N_{g}} \right)^{2}} \to 0, \quad \text{ as } G\to\infty.
    \]
    Also, if Assumption \ref{assump:Order of second moment of Xg'Omega_g Xg} is changed to the following: For some $p>1,$
    \[  
    \mathbb{E}\left\Vert X_{g}^{\prime}\epsilon_{g}\epsilon_{g}^{\prime}X_{g}-X_{g}^{\prime}\Omega_{g}X_{g}\right\Vert ^{p}=\begin{cases}
    O\left(N_{g}^{2p}\right),  & \; \text{for strong dependence}\\
    O\left(N_{g}^{p}h^{p}\left(N_{g}\right)\right), & \; \text{for semi-strong dependence}\\
    O\left(N_{g}^{p}\right), & \; \text{for weak dependence}
    \end{cases}
    \]
    uniformly in $g,$ then the Liapounov condition reduces to 
    \[
        \frac{\sum_{g=1}^{G} \frac{h^{p}(N_{g})}{N_{g}^{p}}}{\left( \sum_{g=1}^{G} \frac{h(N_{g})}{N_{g}} \right)^{p}} \to 0, \quad \text{ as } G\to\infty.
    \]
    By Lemma \ref{lemma: sum h^2/N^2 to 0}, the above conditions are satisfied when $O(G)$ clusters are nearly balanced.
    
    Similar conclusions apply under weak cross-sectional dependence, in which  the Liapounov condition for some $p>1$ simplifies to
    \[
        \frac{\sum_{g=1}^{G} \frac{1}{N_{g}^{p}}}{\left( \sum_{g=1}^{G} \frac{1}{N_{g}} \right)^{p}} \to 0, \quad \text{ as } G\to\infty.
    \]
\end{remark}

Theorem \ref{thm:normal} and Theorem \ref{thm: consistency of V hat} provide a framework for developing a Wald-type test for the general linear hypothesis problem using our proposed estimator. Let us consider the problem of testing the general linear hypothesis $H_{0}:R\beta=r$ against $H_{1}:R\beta\ne r$, where $R$ is a known $q\times k$ matrix and $r$ is a known $q\times 1$ vector, with rank$(R)=q\leq k.$ The following theorem states our proposed test. 
 
\begin{theorem}\label{thm:testing Rbeta=r}
Let $\Pi_{G}^{q\times q}=RV_{G}R^{\prime}$
and $\hat{\Pi}_{G}^{q\times q}=R\hat{V}_{G}R^{\prime}$. Then,
under $H_{0}:R\beta=r,$ and Assumptions \ref{assump:exogeneity}, \ref{assump:error independence}, \ref{assump:Xg'Xg/Ng inverse is bounded}, \ref{assump:existence of moment}, \ref{assump:lambda max V_G=lambda min V_G}, \ref{assump:Order of second moment of Xg'Omega_g Xg}, and Lemma \ref{lemma: sum h^2/N^2 to 0}, for nearly balanced clusters, 
\[
G\left(R\hat{\bar{\beta}}-r\right)^{\prime}\hat{\Pi}_{G}^{-1}\left(R\hat{\bar{\beta}}-r\right)\xrightarrow{d}\chi_{q}^{2},\mbox{ as \ensuremath{G}\ensuremath{\rightarrow\infty}.}
\]
For the unbalanced case, the above holds for: (i) strong dependence generally, (ii) weak and semi-strong dependence, if $O(G)$ clusters are nearly balanced.
\end{theorem}

The proof of Theorem \ref{thm:testing Rbeta=r} is given in the Appendix.
The test described in this theorem can be used to determine whether explanatory variables in a model are significant, for cross-sectional data with cluster dependence. 

In the next section, we discuss the clustered model with varying parameters and provide some asymptotic results under the three kinds of dependence.

\section{Linear Model with Varying Parameters}\label{section:model2}

There is no compelling justification for assuming that the coefficient vector $\beta$ is homogeneous across clusters. It is therefore worthwhile to allow $\beta$ to vary across clusters or higher-level units (e.g., states or districts). In this section, we demonstrate that the proposed estimator remains consistent when $\beta$ is permitted to vary across clusters, within the framework of \cite{swamy1970efficient}. We further establish a central limit theorem for the estimator and introduce a novel testing procedure based on the concept of superblocks. 

\subsection{Model with Cluster-Specific Random Coefficients}

We consider a linear regression model for clustered data in which the regression coefficients are allowed to vary across clusters. Specifically, for each cluster
$g = 1,2,\ldots,G$, the outcome variable satisfies
\begin{equation} \label{eq: model2}
Y_g = X_g \beta_g + \epsilon_g,
\end{equation}
where $\beta_g$ is a $k \times 1$ vector of cluster-specific regression
coefficients, and $X_g$, $Y_g$, and $\epsilon_g$ are defined analogously to those in model \eqref{eq: model1}. To allow for systematic heterogeneity across clusters, we adopt a random coefficients specification. In particular, we assume that $\beta_g = \beta + u_g,$ where $\beta$ is a common $k \times 1$ parameter vector and $u_g$ represents cluster-specific deviations from the common mean. Substituting this into \eqref{eq: model2} yields
\begin{equation} \label{eq: model2 for u_g}
    Y_g = X_g \beta + X_g u_g + \epsilon_g, \quad g=1,2,\ldots,G.
\end{equation}

\begin{assumption}\label{assump: epsilon_g, X_g and u_g independent}
For each cluster $g = 1,2,\ldots,G$, the random coefficient vector $u_g$ is
independent of the idiosyncratic error term $\epsilon_g$ and the regressor matrix
$X_g$. Moreover, the collection $\{u_g\}_{g=1}^G$ is independent across clusters.
\end{assumption}

Assumption \ref{assump: epsilon_g, X_g and u_g independent} ensures exogeneity of the random coefficients by requiring that $u_g$ is independent of the regressors and idiosyncratic errors within each cluster, and independent across clusters.

\begin{assumption}\label{assump: distribution of u_g}
The random vectors $u_g$ are independently distributed across clusters with
\[
\mathbb{E}[u_g]=0 \quad \text{and} \quad \mathrm{Var}(u_g)=\Delta,
\]
where $\Delta$ is a finite, positive definite $k \times k$ matrix.
\end{assumption}
Assumption \ref{assump: distribution of u_g} formalizes the random coefficients
structure by allowing for unrestricted covariance across the elements of $u_g$,
while maintaining a common mean equal to zero.

For the model \eqref{eq: model2}, the proposed estimator as described in \eqref{eq: proposed estimator} is
\begin{equation} \label{eq: beta_hat expression with beta bar}
     \hat{\bar{\beta}}=\bar{\beta}+\frac{1}{G}\sum_{g=1}^G \left(X_g^\prime X_g\right)^{-1}X_g^\prime \epsilon_g.
\end{equation}

Note that for the random coefficients model \eqref{eq: model2 for u_g}, 
\begin{equation} \label{eq: beta_hat for model2}
    \hat{\bar{\beta}}= \beta + \frac{1}{G}\sum_{g=1}^G u_g + \frac{1}{G}\sum_{g=1}^G \left(X_g^\prime X_g\right)^{-1}X_g^\prime \epsilon_g.
\end{equation}

Clearly, $\hat{\bar{\beta}}$ is an unbiased estimator of $\bar{\beta}$ for model \eqref{eq: model2}. Under the random coefficients specification in model \eqref{eq: model2 for u_g}, the estimator remains unbiased for the common parameter $\beta$. The variance-covariance matrix of $\hat{\bar\beta}$ under model \eqref{eq: model2 for u_g} is defined as 
\begin{equation} \label{eq: V^*}
     V^*=V+\frac{1}{G}\Delta\,,
\end{equation}
where $V$ is the variance-covariance matrix of $\frac{1}{G}\sum_{g=1}^G \left(X_g^\prime X_g\right)^{-1}X_g^\prime \epsilon_g$ as defined in \eqref{eq: V}.

\begin{theorem}\label{thm:consistency under varying parameters}
For the model \eqref{eq: model2 for u_g}, under Assumptions \ref{assump:exogeneity}, \ref{assump:error independence}, \ref{assump:Xg'Xg/Ng inverse is bounded}, \ref{assump: epsilon_g, X_g and u_g independent}, and \ref{assump: distribution of u_g}, for strong, semi-strong, or weak dependence, and for any cluster sizes, $\hat{\bar\beta}$ is consistent for $\beta$.
\end{theorem}

The proof of Theorem \ref{thm:consistency under varying parameters} is given in the Appendix.

Note that from the proof of Theorem \ref{thm:consistency under varying parameters} given in the Appendix, we obtain 
\[
\mathbb{E}\left\Vert \hat{\bar{\beta}}-\beta\right\Vert ^{2} = \operatorname{tr}(V)+\frac{1}{G}\operatorname{tr}\left(\Delta\right) = O\left(\frac{1}{G}\right),
\]
irrespective of the dependence structure of $\epsilon_g$ and for arbitrary cluster sizes. Therefore, for the model \eqref{eq: model2 for u_g}, $\hat{\bar{\beta}}$ is $\sqrt{G}$-consistent for $\beta,$ under any kind of cross-sectional dependence.

\begin{remark}
Since $\hat{\bar{\beta}}$ admits the representation given in \eqref{eq: beta_hat expression with beta bar}, it provides a consistent estimator of the average regression parameter across clusters. Consequently, when the parameter of interest is $\bar{\beta}$, the proposed estimator consistently estimates this quantity.
\end{remark}

\begin{assumption}\label{assump: existence of moment for ug}
For some $p > 1$, the $2p$-th absolute moment of each component $u_{gj}$ exists and
is finite for all $j = 1,\ldots,k$ and $g = 1,2,\ldots,G$.
\end{assumption}
Assumption \ref{assump: existence of moment for ug} provides the moment conditions for the application of central limit theorems in the presence of random coefficients.

\begin{theorem}\label{thm:normal under varying parameters}
Under Assumptions \ref{assump:exogeneity}, \ref{assump:error independence}, \ref{assump:Xg'Xg/Ng inverse is bounded}, \ref{assump:existence of moment}, \ref{assump:lambda max V_G=lambda min V_G}, \ref{assump: epsilon_g, X_g and u_g independent}, \ref{assump: distribution of u_g}, and \ref{assump: existence of moment for ug},  for strong, semi-strong ot weak dependence, and for arbitrary cluster sizes, $V^{*-1/2}(\hat{\bar\beta}-\beta) \xrightarrow{d} N_k(0, I_k)$.
\end{theorem}

The proof of Theorem \ref{thm:normal under varying parameters} is given in the Appendix.

 \begin{remark}\label{rem: no conditions needed for normality in varying clusters}
For the random coefficients model \eqref{eq: model2 for u_g} with $\Delta\neq 0,$ the central limit theorem for the proposed estimator does not require the restrictions stated in Theorem \ref{thm:normal} in case of weak and semi-strong dependence.
\end{remark}

Let us define $V^*_G = G V^*.$ To estimate this variance for the random coefficients model, we use the variance-covariance matrix estimator as in Theorem \ref{thm: consistency of V hat}:
\[
\hat{V}^*_G = \frac{1}{G}\sum_{g=1}^{G}\left(X_{g}^{\prime}X_{g}\right)^{-1}X_{g}^{\prime}e_{g}e_{g}^{\prime}X_{g}\left(X_{g}^{\prime}X_{g}\right)^{-1}.
\]
The residuals for this model are $e_{g}=Y_{g}-X_{g}\hat{\bar{\beta}} = X_gu_g+\epsilon_g-X_g(\hat{\bar{\beta}}-\beta).$

\begin{theorem} \label{thm: V hat consistency for model 2}
    Under Assumptions  \ref{assump:exogeneity}, \ref{assump:error independence}, \ref{assump:Xg'Xg/Ng inverse is bounded}, \ref{assump:existence of moment}, \ref{assump:lambda max V_G=lambda min V_G}, \ref{assump: epsilon_g, X_g and u_g independent}, \ref{assump: distribution of u_g} and \ref{assump: existence of moment for ug}, for strong, semi-strong or weak dependence, and for arbitrary cluster sizes, $\hat{V}^{*}_G $ is consistent for $V_G^*$ under model \eqref{eq: model2 for u_g}.
\end{theorem}
The proof of Theorem \ref{thm: V hat consistency for model 2} is given in the Appendix.

Let us consider the problem of testing the general linear hypothesis $H_{0}:R\beta=r$ against $H_{1}:R\beta\ne r$  for model \eqref{eq: model2 for u_g}, where $R$ is a known $q\times k$ matrix and $r$ is a known $q\times 1$ vector, with rank$(R)=q\leq k.$ The following theorem presents the proposed test statistic and its asymptotic properties. 

\begin{theorem} \label{thm: Rb = r testing for model2}
Let $\Gamma_{G}=RV^*_{G}R^{\prime}$ and $\hat{\Gamma}_{G}=R\hat{V}^{*}_{G}R^{\prime}$. Under the null hypothesis $H_{0}:R\beta=r,$ and Assumptions \ref{assump:exogeneity}, \ref{assump:error independence}, \ref{assump:Xg'Xg/Ng inverse is bounded}, \ref{assump:existence of moment}, \ref{assump:lambda max V_G=lambda min V_G}, \ref{assump: epsilon_g, X_g and u_g independent}, \ref{assump: distribution of u_g}, and \ref{assump: existence of moment for ug},  for strong, semi-strong or weak dependence, and for arbitrary cluster sizes, 
\[
G\left(R\hat{\bar{\beta}}-r\right)^{\prime}\hat{\Gamma}_{G}^{-1}\left(R\hat{\bar{\beta}}-r\right)\xrightarrow{d}\chi_{q}^{2}, \quad \text{ as } G\to\infty.
\]
\end{theorem}
The proof of Theorem \ref{thm: Rb = r testing for model2} is given in the Appendix.
In the next subsection, we propose a novel method for testing whether the parameters are constant across clusters.

\subsection{Testing Parameter Constancy Using Superblocks}

Suppose we have clustered data organized in a superblock structure, where the data are grouped into $D$ superblocks. The $l^{th}$ superblock contains $P_{l}$ clusters such that 
\[
\sum_{l=1}^{D}P_{l}=G.
\]
Let us assume that for the $l^{th}$ superblock $\mathcal{S}_{l}$, $\beta_g = \beta + u_l,$ for $g\in\mathcal{S}_{l}$, where $\beta$ is a common $k \times 1$ parameter vector and $u_l$ has the properties similar to that mentioned in Assumption \ref{assump: epsilon_g, X_g and u_g independent} and Assumption \ref{assump: distribution of u_g}, i.e., $u_l$ has mean $0,$ and variance matrix $\Delta.$ Substituting this representation into \eqref{eq: model2} yields
\begin{equation} \label{eq: model2 for u_l}
    Y_g = X_g \beta + X_g u_l + \epsilon_g; \quad g\in \mathcal{S}_{l},\; l=1,2,\ldots,D.
\end{equation}
This specification allows the slope coefficients to be common within each superblock while permitting systematic heterogeneity across superblocks through the random effects $u_l$. Consequently, the composite error term captures both idiosyncratic shocks and superblock-level parameter variation. This assumption may not be a restriction on the model from a practical point of view.  It is probably safe, for most of the applications, to assume that at the lower tier (village, school...), parameters are not expected to change, but the underlying parameter vector may change at the district, state, or country level, i.e., at the higher level (Tier 2, Tier 3,...).

For the $l^{th}$ superblock $\mathcal{S}_{l}$, define $\Tilde{\beta}_{l}=\frac{1}{P_{l}}\sum_{g\in \mathcal{S}_{l}}\hat{\beta}_{g},$ constructed by taking the average of the clusters belonging to the $l^{th}$ superblock. Then, the proposed estimator mentioned in \eqref{eq: proposed estimator} becomes 
\[
\hat{\bar{\beta}} = \frac{1}{G}\sum_{l=1}^D P_l\Tilde{\beta}_l.
\]

Following the proof of Theorem \ref{thm:consistency}, it is easy to note that for every $l,$ $\Tilde{\beta}_l$ is consistent for $\beta,$ as $P_l\to\infty.$ 
For the $l$-th superblock $\mathcal{S}_{l},$ the variance-covariance
matrix of $\Tilde{\beta}_{l}$ is given by
\[
V_{G,l}=\frac{1}{P_{l}^{2}}\sum_{g\in \mathcal{S}_{l}}\mathbb{E}\left[\left(X_{g}^{\prime}X_{g}\right)^{-1}X_{g}^{\prime}\Omega_{g}X_{g}\left(X_{g}^{\prime}X_{g}\right)^{-1}\right].
\]

\begin{assumption} \label{assump: lambda_max Vl = lambda_min Vl}
    Assume that $V_{G,l}$ is finite and positive definite. Also, following from Assumption \ref{assump:lambda max V_G=lambda min V_G} 
\[
O_{e}\left(\lambda_{min}\left(V_{G,l}\right)\right)=O_{e}\left(\lambda_{max}\left(V_{G,l}\right)\right) = \begin{cases}
O_e\left(\frac{1}{P_{l}}\right), & \text{for strong dependence}\\
O_e\left(\frac{1}{P_{l}^{2}}\sum_{g\in \mathcal{S}_{l}}\frac{h\left(N_{g}\right)}{N_{g}}\right), & \text{for semi-strong dependence}\\
O_e\left(\frac{1}{P_{l}^{2}}\sum_{g\in \mathcal{S}_{l}}\frac{1}{N_{g}}\right), & \text{for weak dependence}
\end{cases}
\]
uniformly in $l$.
\end{assumption}

For weak dependence, the above follows immediately because both the smallest and largest eigenvalues satisfy $\lambda_{min}\left(\Omega_{g}\right)=O_{e}(1)=\lambda_{max}\left(\Omega_{g}\right)$
uniformly in $g.$ Also, by applying Lemma \ref{lemma:Exact order summability}, we can rigorously characterize the orders of $\lambda_{\min}\left(V_{G,l}\right)$ and $\lambda_{\max}\left(V_{G,l}\right),$ for the three kinds of dependence.

\begin{lemma} \label{lemma: V_l hat consistency}
    For the model \eqref{eq: model2 for u_l}, define $\Tilde{V}_{l}=\frac{1}{P_{l}^{2}}\sum_{g\in \mathcal{S}_{l}}\left(X_{g}^{\prime}X_{g}\right)^{-1}X_{g}^{\prime}e_{g}^{*}e_{g}^{*\prime}X_{g}\left(X_{g}^{\prime}X_{g}\right)^{-1},$
where $e_{g}^{*}=Y_{g}-X_{g}\Tilde{\beta}_{l}.$ Then, under Assumptions \ref{assump:exogeneity}, \ref{assump:error independence}, \ref{assump:Xg'Xg/Ng inverse is bounded}, \ref{assump:Order of second moment of Xg'Omega_g Xg} and \ref{assump: lambda_max Vl = lambda_min Vl}, for every $l=1,2,\ldots,D,$ $\Tilde{V}_l$ is consistent for $V_{G,l},$ as $P_l\to\infty.$
\end{lemma}
The proof of Lemma \ref{lemma: V_l hat consistency} is given in the Appendix.

\begin{assumption}
\label{assump: condition for T_SB normality}
Assume that as $P_l \to \infty$ and $D \to \infty$, the following convergence holds:
\[
\frac{1}{\sqrt{D}}\sum_{l=1}^D \frac{P^2_l}{\sum_{g \in \mathcal{S}_l} \frac{h(N_g)}{N_g}} \cdot \frac{1}{G^2} \sum_{l=1}^D \sum_{g \in \mathcal{S}_l} \frac{h(N_g)}{N_g} \to 0.
\]
\end{assumption}

\begin{theorem} \label{thm: testing Delta=0}
 Let $T_{SB}=\sum_{l=1}^{D}\left(\Tilde{\beta}_{l}-\hat{\bar{\beta}}\right)^{\prime}\Tilde{V}_{l}^{-1}\left(\Tilde{\beta}_{l}-\hat{\bar{\beta}}\right).$
Then, under Assumptions \ref{assump:exogeneity}, \ref{assump:error independence}, \ref{assump:Xg'Xg/Ng inverse is bounded}, \ref{assump:Order of second moment of Xg'Omega_g Xg}, \ref{assump: lambda_max Vl = lambda_min Vl}, \ref{assump: condition for T_SB normality}, and Lemma \ref{lemma: V_l hat consistency}, for nearly balanced clusters, under any kind of dependence, 
\[
Z=\frac{T_{SB}-kD}{\sqrt{2kD}}\xrightarrow[H_{0}]{d}N\left(0,1\right),\;\text{ as }\min_{l}P_{l}\to\infty,\;D\to\infty,\text{ and }\frac{D}{\min_{l}P_{l}}\to0.
\]
\end{theorem}

The proof of Theorem \ref{thm: testing Delta=0} is given in the Appendix. 

The test statistic developed in this theorem is designed to test the parameter constancy across superblocks. The assumption that \( D / \min_{l} P_{l} \to 0 \) as \( D \to \infty \) and \( \min_{l} P_{l} \to \infty \) is economically and statistically meaningful, as it ensures that the number of superblocks grows at a slower rate than the minimum number of clusters within each superblock.

In the following subsection, we examine the properties of the traditional pooled ordinary least squares (POLS) estimator under the random coefficients model, considering alternative forms of cross-sectional error dependence for arbitrary cluster sizes.

\subsection{Pooled Ordinary Least Squares}

The most commonly used estimator in the clustered framework is the pooled ordinary least squares (POLS) estimator. 
Stacking the observations across clusters from the model in \eqref{eq: model1} yields
\[
Y = X \beta + \epsilon,
\]
where $Y$ and $\epsilon$ are vectors of dimension $\left(\sum_{g=1}^G N_g\right) \times 1$, and $X$ is a $\left(\sum_{g=1}^G N_g\right) \times k$ matrix of regressors.

For the above model, the POLS estimator of $\beta$ is defined as
\begin{equation} \label{eq: ols estimator}
\hat{\beta}_{POLS} = \left(X^{\prime}X\right)^{-1} X^{\prime} Y,
\end{equation}
which coincides with the ordinary least squares estimator when the data are pooled across clusters. Under Assumptions \ref{assump:exogeneity} and \ref{assump:error independence}, the variance-covariance matrix of $\hat{\beta}_{POLS}$ can be expressed as
\begin{equation} \label{eq: ols variance}
\Sigma = \mathbb{E}\left[ \left(X^{\prime}X\right)^{-1} X^{\prime} \Omega X \left(X^{\prime}X\right)^{-1} \right],
\end{equation}
where $\Omega = \mathrm{Var}(\epsilon)$ represents the variance-covariance matrix of the stacked error vector $\epsilon$. In this setup, the conventional cluster-robust variance estimator (CRVE) takes the form:
\begin{equation} \label{eq: crve}
\hat{\Sigma} = \left(X^{\prime}X\right)^{-1} 
\Bigg( \sum_{g=1}^{G} X_g^{\prime} \hat{\epsilon}_g \hat{\epsilon}_g^{\prime} X_g \Bigg) 
\left(X^{\prime}X\right)^{-1},
\end{equation}
where $\hat{\epsilon}_g = Y_g - X_g \hat{\beta}_{POLS}$ denotes the vector of residuals for the $g$-th cluster, for $g = 1,2,\ldots,G$. This estimator is consistent for the true variance of $\hat{\beta}_{POLS}$ in the presence of arbitrary within-cluster correlation, provided that the number of clusters $G$ is sufficiently large and cluster sizes satisfy a mild homogeneity condition.

There are several settings in which the proposed estimator can be shown to be asymptotically more efficient than the pooled ordinary least squares (POLS) estimator. In this context, we consider the following results.

\begin{result}\label{result: V<VOLS, for weak dependence}
    If $\Omega_g=a_g I_{N_g}$ and $X_g^\prime X_g=a_g,$ then $V<\Sigma.$
\end{result} 
The proof of Result \ref{result: V<VOLS, for weak dependence} is given in the Appendix. The above result clearly suggests that the proposed estimator is more efficient, even under weak cross-sectional dependence.

\begin{result}\label{result:efficiency}
    Let us assume that $\Omega_g= (a-b) I_{N_g} + b \mathbbm{1}\mathbbm{1}^{\prime}$ and $X_g = \mathbbm{1},$ $\forall g,$ where $0<b<a<\infty.$  Also, assume that the first cluster is asymptotically dominant, with \( N_{1} \to \infty \), while the remaining clusters have finite sizes. Then, the proposed estimator $\hat{\bar{\beta}}$ is asymptotically more efficient than $\hat{\beta}_{POLS}$.
\end{result}

The proof of the above result is provided in the Appendix. Note that \(\Omega_g\) exhibits strong cross-sectional dependence in this setting. It follows from the result that the proposed estimator is asymptotically more efficient in the presence of strong error dependence and unbalanced cluster sizes.

We rewrite the model \eqref{eq: model2 for u_g} as
\begin{equation}\label{eq: model 2 for eta_g}
    Y_g = X_g\beta+\eta_g,\quad \text{ where } \eta_g=X_g u_g+\epsilon_g.  
\end{equation}
Here, the composite error term $\eta_g$ incorporates both idiosyncratic shocks and unobserved heterogeneity arising from the random coefficient variation. The POLS estimate of $\beta$ for model \eqref{eq: model 2 for eta_g} is given by
\begin{equation}\label{eq: POLS for model 2}
    \hat{\beta}_{POLS}=(X^{\prime}X)^{-1}X^{\prime}Y=\beta+(X^{\prime}X)^{-1}\sum_{g=1}^G X_g^\prime\eta_g.
\end{equation}
Clearly, $\mathbb{E}[\hat{\beta}_{POLS}]=0,$ since $\mathbb{E}[\eta_g]=0,$ by Assumption \ref{assump:exogeneity} and \ref{assump: distribution of u_g}. Also, by Assumptions \ref{assump: epsilon_g, X_g and u_g independent}, \ref{assump: distribution of u_g} and equation \eqref{eq: ols variance}, the variance-covariance matrix of $\hat{\beta}_{POLS}$ under model \eqref{eq: model 2 for eta_g} is
\begin{equation} \label{eq: ols variance for model 2}
    {\Sigma}^* = \Sigma + \mathbb{E}\left[(X^{\prime}X)^{-1}\sum_{g=1}^G X_g^\prime X_g\Delta X_g^\prime X_g(X^{\prime}X)^{-1}\right].
\end{equation}
Now we establish the asymptotic theory for the POLS for $\beta$ varying with clusters.

\begin{assumption}\label{assump: Xg'Xg=O(Ng)}
(i) Assume that $C_{G}=\mathbb{E}\left[\frac{X^{\prime}X}{\sum_{g=1}^{G}N_{g}}\right]$
is finite and uniformly positive definite over $G.$ Also, assume
that 
\[
\frac{X^{\prime}X}{\sum_{g=1}^{G}N_{g}} - C_{G}=o_{P}(1).
\]

(ii) Assume that for some $p>1,$
\[
\mathbb{E}\left[\left(\operatorname{tr}\left(X_{g}^{\prime}X_{g}\right)\right)^{p}\right]=O_e\left(N_{g}^{p}\right) \; \text{ uniformly in } g,
\]
as in Definition \ref{def:exact order and order uniformly}.
\end{assumption}
By Assumption \ref{assump: Xg'Xg=O(Ng)}(i), we mean that $\left|\left(\frac{X^{\prime}X}{\sum_{g=1}^{G}N_{g}}\right)_{j,j_1}-(D_G)_{j,j_1}\right|=o_P(1),$ for all $j,j_1=1,2,\ldots,k.$ This implies element-wise convergence in probability of the sample matrix to its expectation.
Assumption \ref{assump: Xg'Xg=O(Ng)}(ii) controls the growth of the regressor matrix within each cluster by requiring that the average squared magnitude of the regressors remains bounded away from zero as $N_g$ increases, uniformly over $g$. It rules out designs in which $X_g'X_g$ grows faster than order $N_g$.

\begin{theorem}
\label{thm: ols consistency under varying parameters}
Under  Assumptions \ref{assump:exogeneity}, \ref{assump:error independence}, \ref{assump: epsilon_g, X_g and u_g independent}, \ref{assump: distribution of u_g}, and \ref{assump: Xg'Xg=O(Ng)}, in the presence of strong, semi-strong, or weak cross-sectional dependence, $\hat{\beta}_{POLS}$ is consistent for $\beta$ for model \eqref{eq: model 2 for eta_g} provided (i) the clusters are nearly balanced, or (ii) the clusters are unbalanced with the restriction $\frac{\max_{g} N_g }{\sum_{g=1}^G N_g} \to 0,$ as $G\to\infty.$
\end{theorem}

The proof of Theorem \ref{thm: ols consistency under varying parameters} is given in the Appendix.

\begin{theorem}  \label{thm: ols inconsistency under varying clusters}
    Under Assumptions \ref{assump:exogeneity}, \ref{assump:error independence}, \ref{assump: epsilon_g, X_g and u_g independent}, \ref{assump: distribution of u_g}, and \ref{assump: Xg'Xg=O(Ng)}, the POLS estimator $\hat{\beta}_{POLS}$ is inconsistent for unbalanced clusters, unless $\frac{\sum_{g=1}^G N_g^2}{\left( \sum_{g=1}^G N_g\right)^2} \to 0.$
\end{theorem}

The proof of Theorem \ref{thm: ols inconsistency under varying clusters} is given in the Appendix.

If a finite number of clusters are extremely large and the remaining clusters are bounded or unbounded at a lower rate than the clusters, $\hat{\beta}_{POLS}$ becomes inconsistent. However, it is needless to mention that even in such situations $\hat{\bar{\beta}}$ is consistent. Heterogeneity in cluster sizes is a common feature of firm-level and industry-level datasets.  For instance, firm-level datasets such as the Annual Survey of Industries (ASI) in India, Compustat in the United States, and ORBIS across multiple countries reveal substantial variation in the number of firms across industries. In these data, a small number of industries account for a disproportionately large share of firms, while many others consist of relatively few firms. Such unbalanced cluster sizes arise naturally in practice and are routinely observed in large-scale economic datasets.

The following section presents the application of both our proposed tests and the POLS-based test to simulated data and empirical datasets.

\section{Simulation and Empirical Analysis} 
\label{section:analysis}

This section investigates the finite-sample performance of the proposed estimator
$\hat{\bar{\beta}}$ under a variety of data-generating processes. We examine the behavior of all test procedures and test statistics developed in the preceding sections, with specific emphasis on their empirical size and power. We further compare the performance of the proposed testing procedure with the conventional test based on the Pooled Ordinary Least Squares (POLS) for testing general linear hypotheses. The superblock estimation approach for testing parameter constancy is also assessed, both through Monte Carlo simulation experiments and an empirical application using data from the Household Consumption Expenditure Survey of India.

\subsection{Simulation Study}

To examine the finite-sample properties of our methods, we conducted
Monte Carlo simulations for a linear model with an intercept and a single
regressor.

\subsubsection{General Linear Hypothesis Testing} 

We are interested in testing $H_{0}:\beta=\beta_{0}$ against $H_{1}:\beta\neq\beta_{0}$ under model \eqref{eq: model1}. We generate 10000 replications, where each replication yields a new
draw of data from the dgp, and that leads to rejection or non-rejection of $H_{0}$. In each replication, there are $G$ clusters, with $N_{g}$ individuals in each cluster. The simulation procedure is described below. The simulated data are generated as
\[
Y_{gi}  =  X_{gi}^{\prime}\beta+\epsilon_{gi}.
\]

The error terms $\epsilon_{g}$ are drawn from a multivariate normal distribution $N_{N_g}(0, \Omega_g)$, where $\Omega_{g}$ are constructed to ensure strong dependence by setting:
\begin{equation} \label{eq: Omega_g = M_g M_g'}
    \Omega_{g} = M_{g}M_{g}^{\prime},
\end{equation}
where $M_{g}$ is an $N_{g} \times N_{g}$ matrix with entries $m_{ij}$ drawn independently from  $\text{Unif}\,[-5, 10]$. By construction, this specification results in the maximum eigenvalue of $\Omega_g$ being of order $O(N_g)$ for each group $g$, thereby generating the desired strong dependence structure in the data. In the simulation design, the first cluster is taken to be a large cluster. For $g = 2,3,\ldots,G$, the cluster sizes $N_g$ are independently drawn from the discrete uniform distribution on $\{25,\ldots,50\}$. The regressors $X_{gi}$ are two-dimensional, consisting of an intercept and a slope. For the first cluster, let $p_1^{(\max)}$ denote the eigenvector corresponding to the largest eigenvalue $\lambda_{\max}(\Omega_1)$. Define $X_1^{*} = \mathrm{sgn}\!\left(p_1^{(\max)}\right)$ and construct the regressor matrix
$X_1 = [\,\mathbbm{1}_{N_1},\, c_1 X_1^{*}\,]$, where the scaling coefficients $c_{1j} \sim \mathrm{Unif}(2,10)$ independently for $j = 1,2,\ldots,N_1$. For clusters $g = 2,3,\ldots,G$, the regressor matrices $X_g$ are generated with $\mathbbm{1}_{N_g}$ as the first column and the second column consisting of $N_g$ independent draws from a normal distribution $N(\mu_g,\omega_g^{2})$, where $\mu_g \sim \mathrm{Unif}(10,100)$ and $\omega_g^{2} \sim \mathrm{Unif}(200,300)$. 

After generating $\{X_{gi}, Y_{gi}\}$, the proposed estimator is computed as in \eqref{eq: proposed estimator}. Keeping the regressors and the model parameters fixed, we replicate the entire data-generating process $R^{*}=10{,}000$ times. For each replication, we compute $\hat{\bar{\beta}},$ and its variance estimate
\[
\hat{V}_G
=
\frac{1}{G}\sum_{g=1}^{G}
\left(X_{g}^{\prime}X_{g}\right)^{-1}
X_{g}^{\prime} e_{g} e_{g}^{\prime} X_{g}
\left(X_{g}^{\prime}X_{g}\right)^{-1},
\]
where $e_{g}=Y_{g}-X_{g}\hat{\bar{\beta}}$, as defined in
Theorem~\ref{thm: consistency of V hat}. The corresponding test statistic is
\[
T
=
G\left(R\hat{\bar{\beta}}-r\right)^{\prime}
\hat{\Delta}_{G}^{-1}
\left(R\hat{\bar{\beta}}-r\right),
\qquad
\hat{\Delta}_{G}=R\hat{V}_{G}R^{\prime},
\]
which is evaluated under both the null hypothesis $H_{0}$ and the alternative hypothesis $H_{1}$.

Across replications, the error terms are generated from the same distribution,
while the covariance matrices $\{\Omega_g\}_{g=1}^{G}$ are held fixed. The
empirical size and empirical power of the test at the $5\%$ significance level are then computed. The empirical size is defined as the proportion of test statistics $T$ that fall in the rejection region under $H_{0}$, while the empirical power is defined analogously under $H_{1}$. For meaningful power comparisons, a size correction is implemented.  Specifically, the size-corrected critical value is taken to be the $95^{\text{th}}$ quantile of the empirical distribution of $T$ obtained from all replications under $H_{0}$. This corrected critical value is  subsequently used to compute size-corrected empirical powers.

For comparison, we also consider the usual test statistic based on pooled
ordinary least squares (POLS), defined as
\[
T_{\text{POLS}}
=
\left(R\hat{\beta}_{\text{POLS}}-r\right)^{\prime}
\left(R\hat{V}_{\text{POLS}}R^{\prime}\right)^{-1}
\left(R\hat{\beta}_{\text{POLS}}-r\right),
\]
where
\[
\hat{V}_{\text{POLS}}
=
\left(X^{\prime}X\right)^{-1}
\left(\sum_{g=1}^{G}
X_{g}^{\prime}\hat{\epsilon}_{g}\hat{\epsilon}_{g}^{\prime}X_{g}\right)
\left(X^{\prime}X\right)^{-1},
\]
with $\hat{\epsilon}_{g}=Y_{g}-X_{g}\hat{\beta}_{\text{POLS}}.$ In the following table, we compare the test based on the proposed estimator with the POLS-based test for unbalanced cluster sizes under strong error dependence and report the empirical size and power of the tests
across different values of $G$ and $N_1$.

\begin{table}[H]
\centering
\caption{Comparison  of the proposed test and the POLS-based test under strong dependence for one large cluster}
\begin{tabular}{|c|c|c|c|c|c|c|c|c|}
\hline 
\multirow{2}{*}{$(G,N_{1})$} & \multicolumn{2}{c|}{Size} & \multicolumn{2}{c|}{Critical Value} & \multicolumn{2}{c|}{Power} & \multicolumn{2}{c|}{Size-corrected Power}\tabularnewline
\cline{2-9} \cline{3-9} \cline{4-9} \cline{5-9} \cline{6-9} \cline{7-9} \cline{8-9} \cline{9-9} 
 & $\hat{\bar{\beta}}$ & $\hat{\beta}_{POLS}$ & $\hat{\bar{\beta}}$ & $\hat{\beta}_{POLS}$ & $\hat{\bar{\beta}}$ & $\hat{\beta}_{POLS}$ & $\hat{\bar{\beta}}$ & $\hat{\beta}_{POLS}$\tabularnewline
\hline 
(25, 100) & 0.064 & 0.138 & 4.27 & 6.75 & 0.996 & 0.995 & 0.995 & 0.993\tabularnewline
\hline 
(25, 500) & 0.066 & 0.835 & 4.42 & 30.68 & 0.991 & 0.866 & 0.986 & 0.741\tabularnewline
\hline 
(50, 100) & 0.058 & 0.137 & 4.21 & 6.52 & 1 & 0.999 & 1 & 0.998\tabularnewline
\hline 
(50, 500) & 0.057 & 0.765 & 4.16 & 15.44 & 1 & 0.9164 & 1 & 0.814\tabularnewline
\hline 
(100, 100) & 0.051 & 0.091 & 3.87 & 5.08 & 1 & 1 & 1 & 1\tabularnewline
\hline 
(100, 500) & 0.052 & 0.219 & 3.89 & 6.41 & 1 & 0.928 & 1 & 0.898\tabularnewline
\hline 
\end{tabular}
\begin{minipage}{13.3cm}
\vspace{.1cm}	
The number of replications is 10,000. The empirical size is calculated at $\beta_0=(1,0.5)$ and the empirical power is calculated at $\beta_1=(1,1.6).$
\end{minipage}
\end{table}

The simulation results indicate that the proposed testing procedure exhibits satisfactory size control and closely matches the nominal significance level. In contrast, the POLS-based test suffers from severe size distortion. To address this discrepancy and facilitate a valid comparison, we evaluate the power of both procedures using size-corrected critical values. The results demonstrate that the proposed test consistently yields superior power compared to the POLS-based test. Additionally, the performance of the proposed method relative to POLS improves as the number of clusters increases.

\subsubsection{Testing Parameter Constancy}

We are interested in testing $H_{0}:\Delta=0$ against $H_{1}:\Delta\neq 0$ under model \eqref{eq: model2 for u_l}. We generate 10000 replications, where each replication yields a new draw of data from the dgp, and that leads to rejection or non-rejection of $H_{0}$. In each replication, there are $G$ clusters, with $N_{g}$ individuals in each cluster. In addition, we assume a superblock structure such that there are $D$ superblocks, and in the $l$-th  superblock, we have $P_l$ clusters. For simplicity, we assume a balanced configuration with $P_l = P$, for all $l=1,2,\ldots, D$.  The simulated data are generated from model \eqref{eq: model2 for u_l}.

The cluster sizes $N_g$ are independently drawn from the discrete uniform distribution on $\{25,\ldots,50\}$. The error terms $\epsilon_{g}$ are drawn from a multivariate normal distribution $N_{N_g}(0, \Omega_g)$, where $\Omega_{g}$ are constructed as in \eqref{eq: Omega_g = M_g M_g'} for each group $g$, to get the desired strong dependence structure in the data. In the simulation design, the regressors $X_{gi}$ are two-dimensional, consisting of an intercept and a slope. The regressor matrices $X_g$ are generated with $\mathbbm{1}_{N_g}$ as the first column and the second column consisting of $N_g$ independent draws from a normal distribution $N(\mu_g,\omega_g^{2})$, where $\mu_g \sim \mathrm{Unif}(10,100)$ and $\omega_g^{2} \sim \mathrm{Unif}(200,300)$. The superblock-level random coefficient deviations $\{u_l\}$ are generated from uniform and normal distributions over different parameter ranges, as reported in the table.

Keeping the regressors and the model parameters fixed, we replicate the entire data-generating process $R^{*}=10{,}000$ times. For each replication, we compute 
\[
T_{SB}
= \sum_{l=1}^{D}\left(\Tilde{\beta}_{l}-\hat{\bar{\beta}}\right)^{\prime}\Tilde{V}_{l}^{-1}\left(\Tilde{\beta}_{l}-\hat{\bar{\beta}}\right),
\]
where
\[
\Tilde{\beta}_{l}=\frac{1}{P_{l}}\sum_{g\in \mathcal{S}_{l}}\hat{\beta}_{g}, \qquad \hat{\bar{\beta}} = \frac{1}{G}\sum_{l=1}^D P_l\Tilde{\beta}_l,
\] and 
\[
    \Tilde{V}_{l}=\frac{1}{P_{l}^{2}}\sum_{g\in \mathcal{S}_{l}}\left(X_{g}^{\prime}X_{g}\right)^{-1}X_{g}^{\prime}e_{g}^{*}e_{g}^{*\prime}X_{g}\left(X_{g}^{\prime}X_{g}\right)^{-1}, 
\]
  with $e_{g}^{*}=Y_{g}-X_{g}\Tilde{\beta}_{l}.$ The corresponding test statistic is
\[
Z = \frac{T_{SB}-kD}{\sqrt{2kD}},
\]
which is evaluated under both the null hypothesis $H_{0}$ and the alternative hypothesis $H_{1}$.

The covariance matrices $\{\Omega_g\}_{g=1}^{G}$ are held fixed across replications. Finite-sample performance is evaluated in terms of empirical size and power at the $5\%$ nominal significance level, using the standard normal critical value $z_{0.975}$.

\begin{table}[H]
    \centering
    \caption{Empirical size and power for tests of parameter constancy under strong dependence}
\begin{tabular}{|c|c|c|c|c|}
\hline 
\multirow{2}{*}{$(P,\,D)$} & \multirow{2}{*}{Size} & \multicolumn{3}{c|}{Power}\tabularnewline
\cline{3-5} \cline{4-5} \cline{5-5} 
 &  & $u_{l}\sim$ Unif $[-0.1,0.1]$ & $u_{l}\sim$ Unif $[-0.2,0.2]$ & $u_{l}\sim N(0,0.01)$\tabularnewline
\hline 
\hline 
(25, 25) & 0.152 & 0.441 & 0.916 & 0.9278\tabularnewline
\hline 
(25, 50) & 0.239 & 0.7193 & 0.9988 & 0.9897\tabularnewline
\hline 
(25, 100) & 0.396 & 0.892 & 1 & 0.9898\tabularnewline
\hline 
(50, 25) & 0.075 & 0.6127 & 0.996 & 0.965\tabularnewline
\hline 
(100, 25) & 0.051 & 0.8376 & 1 & 1\tabularnewline
\hline 
(100, 50) & 0.056 & 0.939 & 1 & 1\tabularnewline
\hline 
\end{tabular}
\begin{minipage}{12.8cm}
\vspace{.1cm}	
The number of replications is 10,000. The empirical size is calculated at $\beta_0=(1,2)$.
\end{minipage}
\end{table}

The simulation results indicate that, as $G$ and $P$ increase, the empirical size of the test approaches the nominal level, while the empirical power converges to one. Moreover, the simulation evidence reflects the importance of the condition $D/P \to 0$, which appears necessary for the desirable finite-sample performance of the proposed procedure.

\subsection{Empirical Illustration}

The following empirical exercise uses household-level data from the Household Consumption Expenditure Survey (HCES) of India for the year 2022--23. The survey’s sampling design is mapped to our superblock framework, with households as individual observations, First Stage Units (FSUs) as clusters, and States as superblocks. For the empirical analysis, Union Territories are excluded due to the small number of clusters they contain. In addition, since each FSU contains at most 18 households, we remove clusters with fewer than five households to ensure meaningful within-cluster variation. After these exclusions, the final sample consists of 239,744 households drawn from 14,491 FSUs across 29 States.

We focus on three key variables in the analysis. The variable \textit{food} denotes household expenditure on food items over the reference period, \textit{total} represents total household consumption expenditure, and \textit{hhsize} measures household size, defined as the total number of members residing in the household.

We consider five alternative Engel curve specifications for modeling household food expenditure behavior. Let $w_i$ denote the ratio of household food expenditure to total household expenditure, $v_i$ denote household food expenditure, and $M_i$ denote total household consumption expenditure. The models, ordered as in the empirical analysis, are given below.

\textit{ Model 1 }(Linear food-share model).
\[
w_i = a_i + b_i M_i + \epsilon_i,
\]

\textit{Model 2} (Linear model).
\[
v_i = a_i + b_i M_i + \epsilon_i,
\]

\textit{Model 3} (Double-log or Cobb--Douglas model).
\[
\ln v_i = \ln a_i + b_i \ln M_i + \epsilon_i.
\]

 \textit{Model 4} (Semi-log model).
\[
\ln v_i = a_i + b_i M_i + \epsilon_i.
\]

\textit{Model 5} (Working--Leser model).
\[
w_i = a_i + b_i \ln M_i + \epsilon_i.
\]

Each of the five models is estimated for the aggregated all-India sample, both with and without household size (\textit{hhsize}) included as an additional control variable.

The following table reports the estimates obtained using the proposed method and the Pooled Ordinary Least Squares (POLS) estimator, along with the value of the test statistic based on the proposed estimator:
\[
Z = \frac{T_{SB} - kD}{\sqrt{2kD}},
\]
where $T_{SB}$ is the superblock-based test statistic defined in the previous subsection.

\begin{table}[H]
\centering
\caption{Table showing the household consumption expenditure for testing state-level
variation}
    \begin{adjustbox}{width=\textwidth}
    \begin{tabular}{|c|c|c|c|c|c|c|}
\hline 
\multirow{2}{*}{} & \multicolumn{3}{c|}{Without hhsize} & \multicolumn{3}{c|}{With hhsize}\tabularnewline
\cline{2-7} \cline{3-7} \cline{4-7} \cline{5-7} \cline{6-7} \cline{7-7} 
 & $\hat{\bar{\beta}}$ & $\hat{\beta}_{POLS}$ & $Z$ & $\hat{\bar{\beta}}$ & $\hat{\beta}_{POLS}$ & $Z$\tabularnewline
\hline 
Model 1 & $\begin{pmatrix}\begin{array}{c}
4.727393e-01\\
-2.008456e-06
\end{array}\end{pmatrix}$ & $\begin{pmatrix}\begin{array}{c}
4.926134e-01\\
-2.499980e-06
\end{array}\end{pmatrix}$ & $463.2962$ & $\begin{pmatrix}\begin{array}{c}
4.542658e-01\\
-5.726445e-06\\
2.017493e-02
\end{array}\end{pmatrix}$ & $\begin{pmatrix}\begin{array}{c}
4.370622e-01\\
-3.009854e-06\\
1.539651e-02
\end{array}\end{pmatrix}$ & $348.6733$\tabularnewline
\hline 
Model 2 & $\begin{pmatrix}\begin{array}{c}
1119.4529217\\
0.3744132
\end{array}\end{pmatrix}$ & $\begin{pmatrix}\begin{array}{c}
3832.1177046\\
0.2303107
\end{array}\end{pmatrix}$ & $453.2751$ & $\begin{pmatrix}\begin{array}{c}
529.7656275\\
0.2946063\\
551.367824
\end{array}\end{pmatrix}$ & $\begin{pmatrix}\begin{array}{c}
1254.8449853\\
0.2066553\\
714.3146155
\end{array}\end{pmatrix}$ & $290.6762$\tabularnewline
\hline 
Model 3 & $\begin{pmatrix}\begin{array}{c}
0.2202282\\
0.8882348
\end{array}\end{pmatrix}$ & $\begin{pmatrix}\begin{array}{c}
1.2618063\\
0.7822229
\end{array}\end{pmatrix}$ & $509.3487$ & $\begin{pmatrix}\begin{array}{c}
1.7128509\\
0.7123225\\
0.0594726
\end{array}\end{pmatrix}$ & $\begin{pmatrix}\begin{array}{c}
1.79900738\\
0.70199741\\
0.05776171
\end{array}\end{pmatrix}$ & $379.0587$\tabularnewline
\hline 
Model 4 & $\begin{pmatrix}\begin{array}{c}
8.066576e+00\\
4.906426e-05
\end{array}\end{pmatrix}$ & $\begin{pmatrix}\begin{array}{c}
8.497149e+00\\
2.022707e-05
\end{array}\end{pmatrix}$ & $821.6923$ & $\begin{pmatrix}\begin{array}{c}
7.968004e+00\\
3.739323e-05\\
7.544114e-02
\end{array}\end{pmatrix}$ & $\begin{pmatrix}\begin{array}{c}
8.13683e+00\\
1.69199e-05\\
9.98657e-02
\end{array}\end{pmatrix}$ & $594.0962$\tabularnewline
\hline 
Model 5 & $\begin{pmatrix}\begin{array}{c}
0.80060181\\
-0.03727856
\end{array}\end{pmatrix}$ & $\begin{pmatrix}\begin{array}{c}
1.23852479\\
-0.08156762
\end{array}\end{pmatrix}$ & $486.5201$ & $\begin{pmatrix}\begin{array}{c}
1.35120594\\
-0.10225263\\
0.02095447
\end{array}\end{pmatrix}$ & $\begin{pmatrix}\begin{array}{c}
1.43733095\\
-0.11125730\\
0.02137633
\end{array}\end{pmatrix}$ & $329.247$\tabularnewline
\hline 
\end{tabular}
    \end{adjustbox}
\begin{minipage}{16.5cm}
    Data are taken from the Household Consumption Expenditure Survey (HCES) of India for the year 2022-23. 
\end{minipage}
\end{table}

\vspace{-.2 cm}

The results reported in the table indicate that, for Model~1, both the proposed estimator and the POLS estimator yield a negative slope coefficient, which is economically plausible. Moreover, across all five model specifications, the null hypothesis $H_{0}:\Delta = 0$ is rejected, providing strong evidence of significant parameter heterogeneity across States. The estimates further suggest that including household size (\textit{hhsize}) has a negligible impact on the results.

\section{Conclusion}\label{section:conclusion}

The present article suggests a novel estimator for regression coefficients in clustered data that accommodates cluster dependence. We demonstrated that the standard pooled ordinary least squares (POLS) method can yield inconsistent results in various practical situations. In contrast, our proposed estimator has been shown to maintain consistency under these conditions.
We analyzed the asymptotic properties of our estimator, applicable to both finite and infinite cluster sizes, and then explored a typical classical random coefficient model. From this, we derived asymptotic results for average (common) parameters and formulated a Wald-type test statistic for addressing general linear hypothesis testing. Additionally, we created a novel test for parameter stability at a higher (superblock) level, assuming that parameters are stable across clusters within the superblock. One real life example has been considered to demonstrate the vast scope of possible applicability of the proposed test for parameter stability.   

Future research can extend the proposed methodology in several ways, such as incorporating two-way clustering, multilevel modeling, and permitting cluster-specific non-random means. These extensions are currently being investigated.

\newpage
\section{Appendix}

For the ease of reading the proofs, we mention briefly the meaning of various types of dependence discussed in Remark \ref{rem:explanation of strong, semi-strong, weak dependence}. Strong dependence refers to a scenario in which $O(G)$ clusters exhibit strong dependency, while other clusters may be semi-strongly or weakly dependent. Semi-strong dependence means that $O(G)$ clusters are semi-strongly dependent, with no strongly dependent clusters present, although weakly dependent clusters may still exist. Weak dependence, on the other hand, describes a situation where all clusters are weakly dependent, with no strong or semi-strongly dependent clusters involved.

\vspace{0.2cm}

{\bf Proof of Lemma \ref{lemma:Exact order summability}}

By Definition \ref{def:exact order and order uniformly}, for the doubly indexed positive sequences $\left\{ z_{N_g}\right\}$ and $\left\{ a_{N_g}\right\},$ $z_{N_g}=O_e\left(a_{N_g}\right)$ uniformly in $g$ means that for some $0<c_1\leq c_2<\infty,$
\[ c_1 \leq \inf_{g,N_g}  \frac{z_{N_g}}{a_{N_g}}  \leq \sup_{g,N_g} \frac{z_{N_g}}{a_{N_g}}  \leq c_2, \]
that is, for all $g$, there exist positive constants $c_1$ and $c_2$ such that $c_1 \cdot a_{N_g} \leq z_{N_g} \leq c_2 \cdot a_{N_g}.$ For any $G\geq 1$, summing these inequalities from $g = 1$ to $G$ gives
\[
c_1\cdot \sum_{g=1}^G a_{N_g} \leq \sum_{g=1}^G z_{N_g} \leq c_2\cdot\sum_{g=1}^G a_{N_g}.
\]
This simplifies to $c_1 \leq \frac{\sum_{g=1}^G Z_{N_g}}{\sum_{g=1}^G a_{N_g}} \leq c_2,$ implying the fact that 
\[
c_1 \leq  \frac{\sum_{g=1}^G z_{N_g}}{\sum_{g=1}^G a_{N_g}}  \leq c_2, \quad \text{ for all } G.
\]
Therefore, the lemma follows directly from Definition \ref{def:exact order and order uniformly}.

\hfill
\qed

For the sake of completeness, we provide the following result.

\begin{result} \label{result:trace inequality}
We provide some of the known results used in the proofs, which can be found in the standard literature. Here, by a positive semi-definite matrix $A\in \mathbb{R}^{n \times n}$, we mean that $A$ is real symmetric and the quadratic form $z^\prime A z \geq 0,$ for all $z \in \mathbb{R}^n.$
    \begin{enumerate}
        \item $\Vert AB \Vert\leq \Vert A \Vert\; \Vert B\Vert,$ for any matrices $A$ and $B.$
        \item $\operatorname{tr}(AB)\leq \operatorname{tr}(A)\;\operatorname{tr}(B),$ for positive semi-definite matrices $A$ and $B.$
        \item $\operatorname{tr}(A)\;\lambda_{\min}(B)\leq \operatorname{tr}(AB)\leq \operatorname{tr}(A)\;\lambda_{\max}(B),$ for positive semi-definite matrices $A$ and $B.$
        \item $z^\prime ABA^\prime z\leq \left(z^\prime AA^\prime z\right) \; \operatorname{tr}(B),$ for any $z \in \mathbbm{R}^k,$ any matrix $A$ and positive semi-definite matrix $B.$
        \item $\lambda_{\max}(ABA^\prime)\leq \lambda_{\max}(AA^\prime)\;\lambda_{\max}(B),$ for any matrix $A$ and positive semi-definite matrix $B.$
        \item $\lambda_{\min}(ABA^\prime)\geq \lambda_{\min}(AA^\prime)\;\lambda_{\min}(B),$ for any matrix $A$ and positive semi-definite matrix $B.$
        \item $\lambda_{\max}(A+B)\leq \lambda_{\max}(A)+\lambda_{\max}(B),$ for positive semi-definite matrices $A$ and $B.$
        \item For two positive semi-definite matrices $A$ and $B,$ if $A\geq B,$ then $\left\Vert A\right\Vert \geq \left\Vert B\right\Vert.$ 
    \end{enumerate}
\end{result}

\textbf{Proof of Theorem \ref{thm:consistency}}

To prove the consistency of $\hat{\bar{\beta}},$ first, note that 
\[
\hat{\bar{\beta}}-\beta=\frac{1}{G}\sum_{g=1}^{G}\left(\frac{X_{g}^{\prime}X_{g}}{N_{g}}\right)^{-1}\frac{X_{g}^{\prime}\epsilon_{g}}{N_{g}}.
\]

Now, using Assumption \ref{assump:exogeneity}, note that
\[
\mathbb{E}\left[\hat{\bar{\beta}}-\beta\right]=\mathbb{E}\left[\frac{1}{G}\sum_{g=1}^{G}\left(\frac{X_{g}^{\prime}X_{g}}{N_{g}}\right)^{-1}\mathbb{E}\left[\frac{X_{g}^{\prime}\epsilon_{g}}{N_{g}}\bigg|X\right]\right]=0.
\]

Using the independence between clusters in the third line, and Result \ref{result:trace inequality} in the sixth line,
\begin{align*}
\mathbb{E}\left\Vert \hat{\bar{\beta}}-\beta\right\Vert ^{2} & =  \mathbb{E}\left\Vert \frac{1}{G}\sum_{g=1}^{G}\left(\frac{X_{g}^{\prime}X_{g}}{N_{g}}\right)^{-1}\frac{X_{g}^{\prime}\epsilon_{g}}{N_{g}}\right\Vert ^{2}\\
 & = \mathbb{E}\mathbb{E}\left[\frac{1}{G^{2}}\sum_{g=1}^{G}\frac{\epsilon_{g}^{\prime}X_{g}}{N_{g}}\left(\frac{X_{g}^{\prime}X_{g}}{N_{g}}\right)^{-1}\sum_{g_{1}=1}^{G}\left(\frac{X_{g_{1}}^{\prime}X_{g_{1}}}{N_{g_{1}}}\right)^{-1}\frac{X_{g_{1}}^{\prime}\epsilon_{g_{1}}}{N_{g_{1}}}\Bigg|X\right]\\
& =\mathbb{E}\mathbb{E}\left[\frac{1}{G^{2}}\sum_{g=1}^{G}\frac{\epsilon_{g}^{\prime}X_{g}}{N_{g}}\left(\frac{X_{g}^{\prime}X_{g}}{N_{g}}\right)^{-2}\frac{X_{g}^{\prime}\epsilon_{g}}{N_{g}}\Bigg|X\right]\\
 & = \mathbb{E}\left[\frac{1}{G^{2}}\sum_{g=1}^{G}\operatorname{tr}\left(\frac{1}{N_{g}^{2}}X_{g}\left(\frac{X_{g}^{\prime}X_{g}}{N_{g}}\right)^{-2}X_{g}^{\prime}\mathbb{E}\left[\epsilon_{g}\epsilon_{g}^{\prime}\big|X\right]\right)\right]\\
 & =\mathbb{E}\left[\frac{1}{G^{2}}\sum_{g=1}^{G}\operatorname{tr}\left(\frac{1}{N_{g}^{2}}X_{g}\left(\frac{X_{g}^{\prime}X_{g}}{N_{g}}\right)^{-2}X_{g}^{\prime}\Omega_{g}\right)\right]\\
 & \leq\mathbb{E}\left[\frac{1}{G^{2}}\sum_{g=1}^{G}\operatorname{tr}\left(\frac{1}{N_{g}}X_{g}\left(\frac{X_{g}^{\prime}X_{g}}{N_{g}}\right)^{-2}X_{g}^{\prime}\right)\frac{\lambda_{max}\left(\Omega_{g}\right)}{N_{g}}\right]\\
 & =\mathbb{E}\left[\frac{1}{G^{2}}\sum_{g=1}^{G}\operatorname{tr}\left(\left(\frac{X_{g}^{\prime}X_{g}}{N_{g}}\right)^{-1}\right)\frac{\lambda_{max}\left(\Omega_{g}\right)}{N_{g}}\right]\\
 & =\frac{1}{G^{2}}\sum_{g=1}^{G}\mathbb{E}\left[\operatorname{tr}\left(\left(\frac{X_{g}^{\prime}X_{g}}{N_{g}}\right)^{-1}\right)\right]\frac{\lambda_{max}\left(\Omega_{g}\right)}{N_{g}}.
\end{align*}
Now, using the facts that $\mathbb{E}\left[\operatorname{tr}\left(\left(\frac{X_{g}^{\prime}X_{g}}{N_{g}}\right)^{-1}\right)\right]=O(1)$
uniformly in $g$ by Assumption \ref{assump:Xg'Xg/Ng inverse is bounded}, $\frac{\lambda_{max}\left(\Omega_{g}\right)}{N_{g}}=O(1)$
uniformly in $g$ for any kind of dependence by Assumption \ref{assump:error independence},
and Lemma \ref{lemma:Exact order summability}, we have 
\[
\mathbb{E}\left\Vert \hat{\bar{\beta}}-\beta\right\Vert ^{2}=O\left(\frac{1}{G}\right)\to0,\text{ as }G\to\infty.
\]

Hence, $\hat{\bar{\beta}}$ is consistent for $\beta,$ for any kind of dependence of $\epsilon_{g},$ and for any cluster sizes.

\hfill 
\qed

\textbf{Proof of Theorem \ref{thm:normal}}

To show the asymptotic normality of $\sqrt{G}V_{G}^{-1/2}\left(\hat{\bar{\beta}}-\beta\right),$
first note that 
\[
\sqrt{G}V_{G}^{-1/2}\left(\hat{\bar{\beta}}-\beta\right)=\frac{1}{\sqrt{G}}V_{G}^{-1/2}\sum_{g=1}^{G}\left(X_{g}^{\prime}X_{g}\right)^{-1}X_{g}^{\prime}\epsilon_{g}.
\]
 For the asymptotic normality, it suffices to prove the Liapounov
condition for CLT. To prove the Liapounov condition, it is required
to show, for any fixed $z\in\mathbb{R}^{k},$ and for some $p>1,$
\begin{align*}
\frac{\sum_{g=1}^{G}\mathbb{E}\left|z^{\prime}\frac{1}{\sqrt{G}}V_{G}^{-1/2}\left(X_{g}^{\prime}X_{g}\right)^{-1}X_{g}^{\prime}\epsilon_{g}\right|^{2p}}{\left(\sum_{g=1}^{G}\mathbb{E}\left[z^{\prime}\frac{1}{\sqrt{G}}V_{G}^{-1/2}\left(X_{g}^{\prime}X_{g}\right)^{-1}X_{g}^{\prime}\epsilon_{g}\right]^{2}\right)^{p}} & \to0,\text{ as }G\to\infty\\
\implies\frac{\sum_{g=1}^{G}\mathbb{E}\left|z^{\prime}V_{G}^{-1/2}\left(X_{g}^{\prime}X_{g}\right)^{-1}X_{g}^{\prime}\epsilon_{g}\right|^{2p}}{\left(\sum_{g=1}^{G}\mathbb{E}\left[z^{\prime}V_{G}^{-1/2}\left(X_{g}^{\prime}X_{g}\right)^{-1}X_{g}^{\prime}\epsilon_{g}\right]^{2}\right)^{p}} & \to0,\text{ as }G\to\infty.
\end{align*}

Using the trace inequality in the second and third line, observe that
\begin{eqnarray*}
\sum_{g=1}^{G}\mathbb{E}\left|z^{\prime}V_{G}^{-1/2}\left(X_{g}^{\prime}X_{g}\right)^{-1}X_{g}^{\prime}\epsilon_{g}\right|^{2p} & = & \sum_{g=1}^{G}\mathbb{E}\mathbb{E}\left[\left(z^{\prime}V_{G}^{-1/2}\left(X_{g}^{\prime}X_{g}\right)^{-1}X_{g}^{\prime}\epsilon_{g}\epsilon_{g}^{\prime}X_{g}\left(X_{g}^{\prime}X_{g}\right)^{-1}V_{G}^{-1/2}z\right)^{p}\bigg|X\right]\\
 & \leq & \sum_{g=1}^{G}\mathbb{E}\mathbb{E}\left[\left(z^{\prime}V_{G}^{-1/2}\left(X_{g}^{\prime}X_{g}\right)^{-1}V_{G}^{-1/2}z\right)^{p}\left(\operatorname{tr}\left(\epsilon_{g}\epsilon_{g}^{\prime}\right)\right)^{p}\bigg|X\right]\\
 & \leq & \sum_{g=1}^{G}\mathbb{E}\left[\left(z^{\prime}V_{G}^{-1}z\right)^{p}\left(\operatorname{tr}\left(\left(X_{g}^{\prime}X_{g}\right)^{-1}\right)\right)^{p}\mathbb{E}\left[\left(\epsilon_{g}^{\prime}\epsilon_{g}\right)^{p}\bigg|X\right]\right].
\end{eqnarray*}

Now, using H\"{o}lder's inequality and Assumption \ref{assump:existence of moment},
the $p^{th}$ absolute moment of $\epsilon_{g}^{\prime}\epsilon_{g}$
conditional on $X,$ is 
\begin{equation} \label{eq: E(epsilon_g'epsilon_g)^p order}
    \mathbb{E}\left[\left(\epsilon_{g}^{\prime}\epsilon_{g}\right)^{p}\bigg|X\right]\leq \mathbb{E}\left[\left(\sum_{i=1}^{N_{g}}1^{\frac{p}{p-1}}\right)^{p-1}\sum_{i=1}^{N_{g}}\left|\epsilon_{gi}\right|^{2p}\Bigg|X\right]=O\left(N_{g}^{p}\right)\;\text{ uniformly in }g. 
\end{equation}

Then, asssuming $\mathbb{E}\left[\left(\epsilon_{g}^{\prime}\epsilon_{g}\right)^{p}\bigg|X\right]\leq c_{0}N_{g}^{p},$
for some $0<c_{0}<\infty,$ we have
\[
\sum_{g=1}^{G}\mathbb{E}\left|z^{\prime}V_{G}^{-1/2}\left(X_{g}^{\prime}X_{g}\right)^{-1}X_{g}^{\prime}\epsilon_{g}\right|^{2p} \leq c_{0}\left(z^{\prime}V_{G}^{-1}z\right)^{p}\sum_{g=1}^{G}\mathbb{E}\left[\left(\operatorname{tr}\left(\left(\frac{X_{g}^{\prime}X_{g}}{N_{g}}\right)^{-1}\right)\right)^{p}\right].
\]
Note that by Assumption \ref{assump:lambda max V_G=lambda min V_G}, we obtain
\[
\frac{z^{\prime}V_{G}^{-1}z}{z^{\prime}z}\leq\lambda_{\max}\left(V_{G}^{-1}\right) =\frac{1}{\lambda_{\min}\left(V_{G}\right)} = 
\begin{cases}
O\left(1\right), & \text{for strong dependence}\\
O\left(\frac{G}{\sum_{g=1}^{G}\frac{h(N_{g})}{N_{g}}}\right), & \text{for semi-strong dependence}\\
O\left(\frac{G}{\sum_{g=1}^{G}\frac{1}{N_{g}}}\right), & \text{for weak dependence}
\end{cases}
\]
Also, for the denominator, we have
\begin{eqnarray*}
 &  & \left(\sum_{g=1}^{G}\mathbb{E}\left[z^{\prime}V_{G}^{-1/2}\left(X_{g}^{\prime}X_{g}\right)^{-1}X_{g}^{\prime}\epsilon_{g}\right]^{2}\right)^{p}\\
 & = & \left(\sum_{g=1}^{G}\mathbb{E}\mathbb{E}\left[\left(z^{\prime}V_{G}^{-1/2}\left(X_{g}^{\prime}X_{g}\right)^{-1}X_{g}^{\prime}\epsilon_{g}\epsilon_{g}^{\prime}X_{g}\left(X_{g}^{\prime}X_{g}\right)^{-1}V_{G}^{-1/2}z\right)\Big|X\right]\right)^{p}\\
 & = & \left(z^{\prime}V_{G}^{-1/2}\mathbb{E}\left[\sum_{g=1}^{G}\left(X_{g}^{\prime}X_{g}\right)^{-1}X_{g}^{\prime}\Omega_{g}X_{g}\left(X_{g}^{\prime}X_{g}\right)^{-1}\right]V_{G}^{-1/2}z\right)^{p}\\
 & = & \left(z^{\prime}V_{G}^{-1/2}G\cdot V_{G}V_{G}^{-1/2}z\right)^{p}\\
& = & G^p\left(z^{\prime}z\right)^{p}.
\end{eqnarray*}

Therefore, using Assumption \ref{assump:Xg'Xg/Ng inverse is bounded} and Lemma \ref{lemma:Exact order summability} in the second line, the Liapounov condition reduces to 
\begin{eqnarray*}
\frac{\sum_{g=1}^{G}\mathbb{E}\left|z^{\prime}V_{G}^{-1/2}\left(X_{g}^{\prime}X_{g}\right)^{-1}X_{g}^{\prime}\epsilon_{g}\right|^{2p}}{\left(\sum_{g=1}^{G}\mathbb{E}\left[z^{\prime}V_{G}^{-1/2}\left(X_{g}^{\prime}X_{g}\right)^{-1}X_{g}^{\prime}\epsilon_{g}\right]^{2}\right)^{p}} & \leq & \frac{c_{0}\left(z^{\prime}V_{G}^{-1}z\right)^{p}\sum_{g=1}^{G}\mathbb{E}\left[\left(\operatorname{tr}\left(\left(\frac{X_{g}^{\prime}X_{g}}{N_{g}}\right)^{-1}\right)\right)^{p}\right]}{G^p\left(z^{\prime}z\right)^{p}}\\
 & \leq & \left(\frac{z^{\prime}V_{G}^{-1}z}{z^{\prime}z}\right)^p
 \frac{c_{0}MG}{G^p} \\
 & = & \begin{cases}
O\left(G^{1-p}\right), & \text{for strong dependence}\\
O\left(\frac{G}{\left[\sum_{g=1}^{G}\frac{h\left(N_{g}\right)}{N_{g}}\right]^{p}}\right), & \text{for semi-strong dependence}\\
O\left(\frac{G}{\left[\sum_{g=1}^{G}\frac{1}{N_{g}}\right]^{p}}\right), & \text{for weak dependence}
\end{cases}
\end{eqnarray*}

Hence, for strong dependence, the asymptotic normality holds trivially. For semi-strong dependence, the asymptotic normality holds if for some $p>1,$ $O\left(\frac{G}{\left[\sum_{g=1}^{G}\frac{h\left(N_{g}\right)}{N_{g}}\right]^{p}}\right)\to0,$
as $G\to\infty.$ For weak dependence, the asymptotic normality holds if for some $p>1,$ $O\left(\frac{G}{\left[\sum_{g=1}^{G}\frac{1}{N_{g}}\right]^{p}}\right)\to0,$
as $G\to\infty.$

\hfill 
\qed

\textbf{Proof of Lemma \ref{lemma: sum h^2/N^2 to 0}} 

Note that
\[\frac{\sum_{g=1}^{G}\frac{h^{2}(N_{g})}{N_{g}^{2}}}{\left(\sum_{g=1}^{G}\frac{h(N_{g})}{N_{g}}\right)^{2}}=\frac{1}{1+\frac{\sum_{g_{1}\neq g_{2}}\frac{h(N_{g_{1}})h(N_{g_{2}})}{N_{g_{1}}N_{g_{2}}}}{\sum_{g=1}^{G}\frac{h^{2}(N_{g})}{N_{g}^{2}}}}=\frac{1}{1+a_{G}}, \text{ say. }
\] 

(i) All the clusters are nearly balanced such that $N_{g}\simeq N,$
for all $g.$ Then, 
\[
a_{G}=\frac{G(G-1)\frac{h^{2}(N)}{N^{2}}}{G\frac{h^{2}(N)}{N^{2}}}=G-1=O(G).
\]

(ii) Suppose that $L$ clusters are extremely large as
compared to the other clusters such that $N_{1},\ldots,N_{L}\to\infty$
and $N_{g}\simeq N\to\infty,$ for $g=L+1,\ldots,G,$ but $\frac{N}{N_{j}}=o(1),$
for $j=1,2,\ldots,L.$ Let $\mathcal{L}$ be the set of extremely
large clusters and without loss of generality, assume that $\frac{h(N_{1})}{N_{1}}<\frac{h(N_{2})}{N_{2}}<\cdots<\frac{h(N_{L})}{N_{L}}.$
Then,
\begin{eqnarray*}
a_{G} & = & \frac{\sum_{g_{1},g_{2}\in\mathcal{L}}\frac{h(N_{g_{1}})h(N_{g_{2}})}{N_{g_{1}}N_{g_{2}}}+\sum_{g_{1}\in\mathcal{L},g_{2}\in\mathcal{L}^{c}}\frac{h(N_{g_{1}})h(N_{g_{2}})}{N_{g_{1}}N_{g_{2}}}+(G-L)(G-L-1)\frac{h^{2}(N)}{N^{2}}}{\sum_{g\in\mathcal{L}}\frac{h^{2}(N_{g})}{N_{g}^{2}}+(G-L)\frac{h^{2}(N)}{N^{2}}}\\
 & \geq & \frac{L(L-1)\frac{h(N_{1})h(N_{2})}{N_{1}N_{2}}+2L(G-L)\frac{h(N_{1})h(N)}{N_{1}N}+(G-L)(G-L-1)\frac{h^{2}(N)}{N^{2}}}{L\frac{h^{2}(N_{L})}{N_{L}^{2}}+(G-L)\frac{h^{2}(N)}{N^{2}}}\\
 & = & \frac{\frac{L(L-1)}{(G-L)}\frac{h(N_{1})h(N_{2})N^{2}}{N_{1}N_{2}h^{2}(N)}+2L\frac{h(N_{1})N}{N_{1}h(N)}+(G-L-1)}{\frac{L}{G-L}\frac{h^{2}(N_{L})N^{2}}{N_{L}^{2}h^{2}(N)}+1},\text{ dividing by }(G-L)\frac{h^{2}(N)}{N^{2}}.
\end{eqnarray*}

Clearly, the numerator is bigger than $G-L-1$ and the denominator
is $o\left(\frac{L}{G-L}\right)+1,$ due to the fact that $\frac{h^{2}(N_{L})N^{2}}{N_{L}^{2}h^{2}(N)}\to0.$
Hence, if $L$ is finite, then $a_{G}=O(G).$ Also, note that $L$
may be $O(G)$ additionally with $G-L=O(G),$ for example, $L=G/2.$
Then also $a_{G}=O(G).$ So, for case (i) and case (ii), $\frac{\sum_{g=1}^{G}\frac{h^{2}(N_{g})}{N_{g}^{2}}}{\left(\sum_{g=1}^{G}\frac{h(N_{g})}{N_{g}}\right)^{2}}\to 0,$ as $G\to\infty.$

Putting $h(N_{g})=1,$ we can have
$\frac{\sum_{g=1}^{G}\frac{1}{N_{g}^{2}}}{\left(\sum_{g=1}^{G}\frac{1}{N_{g}}\right)^{2}}\to0,$ as $G\to\infty$ under the similar cases described above.
This concludes the proof of Lemma \ref{lemma: sum h^2/N^2 to 0}.

\hfill
\qed

\textbf{Proof of Theorem \ref{thm: consistency of V hat}}

To show that $\hat{V}_{G}$ is consistent for $V_{G},$
we have to show that $\frac{1}{\left\Vert V_{G}\right\Vert }\left\Vert \hat{V}_{G}-V_{G}\right\Vert =o_{P}(1).$
To show this, note that $e_{g}=Y_{g}-X_{g}\hat{\bar{\beta}}=\epsilon_{g}-X_{g}\left(\hat{\bar{\beta}}-\beta\right),$
which implies $e_{g}e_{g}^{\prime}=\epsilon_{g}\epsilon_{g}^{\prime}+X_{g}\left(\hat{\bar{\beta}}-\beta\right)\left(\hat{\bar{\beta}}-\beta\right)^{\prime}X_{g}^{\prime}-\epsilon_{g}\left(\hat{\bar{\beta}}-\beta\right)^{\prime}X_{g}^{\prime}-X_{g}\left(\hat{\bar{\beta}}-\beta\right)\epsilon_{g}^{\prime}.$
Now, 
\begin{eqnarray*}
\hat{V}_{G}-V_{G} & = & \frac{1}{G}\sum_{g=1}^{G}\left(\left(X_{g}^{\prime}X_{g}\right)^{-1}X_{g}^{\prime}e_{g}e_{g}^{\prime}X_{g}\left(X_{g}^{\prime}X_{g}\right)^{-1}-\mathbb{E}\left[\left(X_{g}^{\prime}X_{g}\right)^{-1}X_{g}^{\prime}\Omega_{g}X_{g}\left(X_{g}^{\prime}X_{g}\right)^{-1}\right]\right)\\
 & = & T_{1}+T_{2}+T_{3}+T_{4},\text{ say,}
\end{eqnarray*}

where $T_{1}=\frac{1}{G}\sum_{g=1}^{G}\left(\left(X_{g}^{\prime}X_{g}\right)^{-1}X_{g}^{\prime}\epsilon_{g}\epsilon_{g}^{\prime}X_{g}\left(X_{g}^{\prime}X_{g}\right)^{-1}-\mathbb{E}\left[\left(X_{g}^{\prime}X_{g}\right)^{-1}X_{g}^{\prime}\Omega_{g}X_{g}\left(X_{g}^{\prime}X_{g}\right)^{-1}\right]\right),$ 

$T_{2}=\frac{1}{G}\sum_{g=1}^{G}\left(X_{g}^{\prime}X_{g}\right)^{-1}X_{g}^{\prime}X_{g}\left(\hat{\bar{\beta}}-\beta\right)\left(\hat{\bar{\beta}}-\beta\right)^{\prime}X_{g}^{\prime}X_{g}\left(X_{g}^{\prime}X_{g}\right)^{-1},$

$T_{3}=-\frac{1}{G}\sum_{g=1}^{G}\left(X_{g}^{\prime}X_{g}\right)^{-1}X_{g}^{\prime}\epsilon_{g}\left(\hat{\bar{\beta}}-\beta\right)^{\prime}X_{g}^{\prime}X_{g}\left(X_{g}^{\prime}X_{g}\right)^{-1}$
and

$T_{4}=-\frac{1}{G}\sum_{g=1}^{G}\left(X_{g}^{\prime}X_{g}\right)^{-1}X_{g}^{\prime}X_{g}\left(\hat{\bar{\beta}}-\beta\right)\epsilon_{g}^{\prime}X_{g}\left(X_{g}^{\prime}X_{g}\right)^{-1}.$

Using Assumption \ref{assump:error independence}, conditional on $X,$ observe that 
\begin{align*}
\mathbb{E}\left[T_{1}\right] & = \mathbb{E}\mathbb{E}\left[T_{1}\big|X\right]\\
 & = \frac{1}{G}\sum_{g=1}^{G}\mathbb{E}\left[\left(X_{g}^{\prime}X_{g}\right)^{-1}X_{g}^{\prime}\mathbb{E}\left[\epsilon_{g}\epsilon_{g}^{\prime}\big|X\right]X_{g}\left(X_{g}^{\prime}X_{g}\right)^{-1}-\mathbb{E}\left[\left(X_{g}^{\prime}X_{g}\right)^{-1}X_{g}^{\prime}\Omega_{g}X_{g}\left(X_{g}^{\prime}X_{g}\right)^{-1}\right]\right]\\
 & = 0,
\end{align*}
and using the independence between clusters in the second line, and the norm inequality in the third line, we have 
\begin{align*}
\mathbb{E}\left\Vert T_{1}\right\Vert ^{2} & =  \mathbb{E}\left\Vert \frac{1}{G}\sum_{g=1}^{G}\left(X_{g}^{\prime}X_{g}\right)^{-1}\left(X_{g}^{\prime}\epsilon_{g}\epsilon_{g}^{\prime}X_{g}-X_{g}^{\prime}\Omega_{g}X_{g}\right)\left(X_{g}^{\prime}X_{g}\right)^{-1}\right\Vert ^{2}\\
& =  \frac{1}{G^{2}}\sum_{g=1}^{G}\mathbb{E}\mathbb{}\left[\left\Vert \left(X_{g}^{\prime}X_{g}\right)^{-1}X_{g}^{\prime}\left(\epsilon_{g}\epsilon_{g}^{\prime}-\Omega_{g}\right)X_{g}\left(X_{g}^{\prime}X_{g}\right)^{-1}\right\Vert ^{2}\bigg|X\right]\\
& \leq \frac{1}{G^{2}}\sum_{g=1}^{G}\mathbb{E}\left[\left\Vert \left(X_{g}^{\prime}X_{g}\right)^{-1}\right\Vert ^{4}\mathbb{E}\left[\left\Vert X_{g}^{\prime}\left(\epsilon_{g}\epsilon_{g}^{\prime}-\Omega_{g}\right)X_{g}\right\Vert ^{2}\Big|X\right]\right]
\end{align*}

Using the fact that $\mathbb{E}\left\Vert X_{g}^{\prime}\epsilon_{g}\epsilon_{g}^{\prime}X_{g}-X_{g}^{\prime}\Omega_{g}X_{g}\right\Vert ^{2}=\begin{cases}
O\left(N_{g}^{4}\right), & \text{for strong dependence}\\
O\left(N_{g}^{2}h^{2}\left(N_{g}\right)\right), & \text{for semi-strong dependence}\\
O\left(N_{g}^{2}\right), & \text{for weak dependence}
\end{cases}$

uniformly in $g$ by Assumption \ref{assump:Order of second moment of Xg'Omega_g Xg}, and $\mathbb{E}\left[\left\Vert \left(X_{g}^{\prime}X_{g}\right)^{-1}\right\Vert ^{4}\right]=O\left(N_{g}^{-4}\right)$
uniformly in $g$ by Assumption \ref{assump:Xg'Xg/Ng inverse is bounded}, we have from Lemma \ref{lemma:Exact order summability}
\[
\mathbb{E}\left\Vert T_{1}\right\Vert ^{2}=\begin{cases}
\frac{1}{G^{2}}\sum_{g=1}^{G}O\left(N_{g}^{-4}\right)O\left(N_{g}^{4}\right)=O\left(\frac{1}{G}\right), & \text{for strong dependence}\\
\frac{1}{G^{2}}\sum_{g=1}^{G}O\left(N_{g}^{-4}\right)O\left(N_{g}^{2}h^{2}\left(N_{g}\right)\right)=O\left(\frac{1}{G^{2}}\sum_{g=1}^{G}\frac{h^{2}\left(N_{g}\right)}{N_{g}^{2}}\right), & \text{for semi-strong dependence}\\
\frac{1}{G^{2}}\sum_{g=1}^{G}O\left(N_{g}^{-4}\right)O\left(N_{g}^{2}\right)=O\left(\frac{1}{G^{2}}\sum_{g=1}^{G}\frac{1}{N_{g}^{2}}\right), & \text{for weak dependence}
\end{cases}
\]
which deduces that 
\[
\left\Vert T_{1}\right\Vert =\begin{cases}
O_{P}\left(\frac{1}{\sqrt{G}}\right), & \text{for strong dependence}\\
O_{P}\left(\frac{1}{G}\sqrt{\sum_{g=1}^{G}\frac{h^{2}\left(N_{g}\right)}{N_{g}^{2}}}\right), & \text{for semi-strong dependence}\\
O_{P}\left(\frac{1}{G}\sqrt{\sum_{g=1}^{G}\frac{1}{N_{g}^{2}}}\right), & \text{for weak dependence}
\end{cases}
\]

Now, for the second term, we have 
\[
\mathbb{E}\left\Vert T_{2}\right\Vert =\mathbb{E}\left\Vert \left(\hat{\bar{\beta}}-\beta\right)\left(\hat{\bar{\beta}}-\beta\right)^{\prime}\right\Vert =\frac{1}{G}\mathbb{E}\left[\operatorname{tr}\left(V_{G}\right)\right].
\]
 By Assumption \ref{assump:lambda max V_G=lambda min V_G}, 
\[
O_{e}\left(\lambda_{min}\left(V_{G}\right)\right)=O_{e}\left(\lambda_{max}\left(V_{G}\right)\right)=\begin{cases}
O_{e}\left(1\right), & \text{for strong dependence}\\
O_{e}\left(\frac{1}{G}\sum_{g=1}^{G}\frac{h\left(N_{g}\right)}{N_{g}}\right), & \text{for semi-strong dependence}\\
O_{e}\left(\frac{1}{G}\sum_{g=1}^{G}\frac{1}{N_{g}}\right), & \text{for weak dependence}
\end{cases}
\]
This yields 
\[
\left\Vert T_{2}\right\Vert =\begin{cases}
O_{P}\left(\frac{1}{G}\right), & \text{for strong dependence}\\
O_{P}\left(\frac{1}{G^{2}}\sum_{g=1}^{G}\frac{h\left(N_{g}\right)}{N_{g}}\right), & \text{for semi-strong dependence}\\
O_{P}\left(\frac{1}{G^{2}}\sum_{g=1}^{G}\frac{1}{N_{g}}\right), & \text{for weak dependence}
\end{cases}
\]

Note that 
\[
T_{3}=-\frac{1}{G}\sum_{g=1}^{G}\left(X_{g}^{\prime}X_{g}\right)^{-1}X_{g}^{\prime}\epsilon_{g}\left(\hat{\bar{\beta}}-\beta\right)^{\prime}=-\left(\hat{\bar{\beta}}-\beta\right)\left(\hat{\bar{\beta}}-\beta\right)^{\prime}=-T_{2},
\]
and $T_{4}=T_{3}.$ Hence, $\left\Vert T_{3}\right\Vert =\left\Vert T_{2}\right\Vert =\left\Vert T_{4}\right\Vert .$
By the triangular inequality for norms
\begin{eqnarray*}
\frac{\left\Vert \hat{V}_{G}-V_{G}\right\Vert }{\left\Vert V_{G}\right\Vert } & \leq & \frac{1}{\left\Vert V_{G}\right\Vert }\left[\left\Vert T_{1}\right\Vert +\left\Vert T_{2}\right\Vert +\left\Vert T_{3}\right\Vert +\left\Vert T_{4}\right\Vert \right].
\end{eqnarray*}

From Lemma \ref{lemma: sum h^2/N^2 to 0}, we have $\frac{\sum_{g=1}^{G}\frac{h^{2}(N_{g})}{N_{g}^{2}}}{\left(\sum_{g=1}^{G}\frac{h(N_{g})}{N_{g}}\right)^{2}}\to 0$ and  $\frac{\sum_{g=1}^{G}\frac{1}{N_{g}^{2}}}{\left(\sum_{g=1}^{G}\frac{1}{N_{g}}\right)^{2}}\to 0,$
as $G\to\infty,$ for the nearly balanced case generally, and the unbalanced case with
$O(G)$ clusters being nearly balanced. Therefore,
\begin{align*}
\frac{\left\Vert T_{1}\right\Vert }{\left\Vert V_{G}\right\Vert } & =\begin{cases}
O_{P}\left(\frac{1}{\sqrt{G}}\right), & \text{for strong dependence}\\
O_{P}\left(\frac{\sqrt{\sum_{g=1}^{G}\frac{h^{2}\left(N_{g}\right)}{N_{g}^{2}}}}{\sum_{g=1}^{G}\frac{h\left(N_{g}\right)}{N_{g}}}\right), & \text{for semi-strong dependence}\\
O_{P}\left(\frac{\sqrt{\sum_{g=1}^{G}\frac{1}{N_{g}^{2}}}}{\sum_{g=1}^{G}\frac{1}{N_{g}}}\right), & \text{for weak dependence}
\end{cases}\\
 & =o_{P}(1),\quad\text{ for any kind of dependence.}
\end{align*}

Also, using Lemma \ref{lemma: sum h^2/N^2 to 0},
\[
\frac{\left\Vert T_{2}\right\Vert }{\left\Vert V_{G}\right\Vert }=O_{P}\left(\frac{1}{G}\right)=o_{P}(1),\text{\; for any kind of dependence.}
\]

Therefore, we have shown that 
\[
\frac{\left\Vert \hat{V}_{G}-V_{G}\right\Vert }{\left\Vert V_{G}\right\Vert }=o_{P}(1),
\]

which deduces that $\hat{V}_{G}$ is consistent for $V_{G},$ under any kind of dependence, for the nearly balanced case. For the unbalanced case, the above holds for (i) strong dependence generally, (ii) weak and semi-strong dependence, if $O(G)$ clusters are nearly balanced.

\hfill
\qed

For the sake of completeness, we provide the following lemma.
\begin{lemma}\label{lemma:lambda_min lambda_max inequality}
Suppose that $A$ is a $m\times n$ matrix with
full row rank $m(\leq n),$ where $m$ is finite. Then, for any $n\times n$
positive definite matrix $B,$
(i) $\lambda_{\min}\left(ABA^{\prime}\right)\ge\lambda_{\min}\left(AA^{\prime}\right)\lambda_{\min}\left(B\right)$
and
(ii) $\lambda_{\max}\left(ABA^{\prime}\right)\leq\lambda_{\max}\left(AA^{\prime}\right)\lambda_{\max}\left(B\right).$
\end{lemma}

\textbf{Proof of Lemma \ref{lemma:lambda_min lambda_max inequality}}

Note that
\[
\lambda_{\min}\left(ABA^{\prime}\right)=\underset{z\in\mathbb{R}^{k}}{min}\frac{z^{\prime}ABA^{\prime}z}{z^{\prime}z}\geq\underset{z\in\mathbb{R}^{k}}{min}\frac{z^{\prime}AA^{\prime}z}{z^{\prime}z}\lambda_{\min}\left(B\right)\ge\lambda_{\min}\left(AA^{\prime}\right)\lambda_{\min}\left(B\right),
\]

and
\[
\lambda_{\max}\left(ABA^{\prime}\right)=\underset{z\in\mathbb{R}^{k}}{max}\frac{z^{\prime}ABA^{\prime}z}{z^{\prime}z}\leq\underset{z\in\mathbb{R}^{k}}{max}\frac{z^{\prime}AA^{\prime}z}{z^{\prime}z}\lambda_{\max}\left(B\right)\leq\lambda_{\max}\left(AA^{\prime}\right)\lambda_{\max}\left(B\right).
\]

\hfill
\qed

\textbf{Proof of Theorem \ref{thm:testing Rbeta=r}}

We show the asymptotic distribution of the Wald-type test statistic based on $\hat{\bar{\beta}}$ in two steps.

\textit{Step 1}: In this step, we show that under $H_{0},$
\[
    \Pi_{G}^{-1/2}\sqrt{G}\left(R\hat{\bar{\beta}}-r\right)\xrightarrow{d}N_{q}(0,I_{q}), \;\text{ as } G\rightarrow\infty
\]
(i) for the nearly balanced case, (ii) for the unbalanced case and strong dependence, (iii) for the unbalanced case, and weak or semi-strong dependence, if $O(G)$ clusters are nearly balanced. 

First note that under $H_{0},$
\[
    \Pi_{G}^{-1/2}\sqrt{G}\left(R\hat{\bar{\beta}}-r\right)=\Pi_{G}^{-1/2}\sqrt{G}R\left(\hat{\bar{\beta}}-\beta\right)=\Pi_{G}^{-1/2}RV_{G}^{1/2}V_{G}^{-1/2}\sqrt{G}\left(\hat{\bar{\beta}}-\beta\right).
\]
Note that Theorem \ref{thm:normal} yields $V_{G}^{-1/2}\sqrt{G}\left(\hat{\bar{\beta}}-\beta\right)\xrightarrow{d}N_{k}(0,I_{k}),$ (i) for nearly balanced case, (ii) for the unbalanced case and strong dependence, (iii) for the unbalanced case, and weak or semi-strong dependence, if $O(G)$ clusters are nearly balanced.
Since $R$ is a finite matrix of dimension $q\times k$ with full
row rank $q,$ and $V_{G}$ is finite and positive definite, it follows
from Corollary 4.24 of \cite{white1984asymptotic} that 
\[
    \Pi_{G}^{-1/2}RV_{G}^{1/2}V_{G}^{-1/2}\sqrt{G}\left(\hat{\bar{\beta}}-\beta\right)\xrightarrow{d}N_{q}(0,I_{q}),\;\text{ as } G\rightarrow\infty
\]
for the cases mentioned already.

\textit{Step 2}: In this step, we show that under $H_{0},$
\[
    \hat{\Pi}_{G}^{-1/2}\sqrt{G}\left(R\hat{\bar{\beta}}-r\right)\xrightarrow{d}N_{q}(0,I_{q}),\;\text{ as } G\rightarrow\infty.
\] 

To show this, first note that using the norm inequality,
\begin{eqnarray*}
\left\Vert\hat{\Pi}_{G}^{-1/2}\sqrt{G}\left(R\hat{\bar{\beta}}-r\right)-\Pi_{G}^{-1/2}\sqrt{G}\left(R\hat{\bar{\beta}}-r\right)\right\Vert & \leq & \left\Vert\hat{\Pi}_{G}^{-1/2}\Pi_{G}^{1/2}-I_{l}\right\Vert\left\Vert\Pi_{G}^{-1/2}\sqrt{G}\left(R\hat{\bar{\beta}}-r\right)\right\Vert.
\end{eqnarray*}
 From the Step 1, we have shown that under $H_0,$ $\Pi_{G}^{-1/2}\sqrt{G}\left(R\hat{\bar{\beta}}-r\right)\xrightarrow{d}N_{q}(0,I_{q}),$ which deduces that
 \[
    \left\Vert\Pi_{G}^{-1/2}\sqrt{G}\left(R\hat{\bar{\beta}}-r\right)\right\Vert=O_P(1),\quad\text{ under } H_{0}.
 \]
 
Therefore, it is enough to show that $\left\Vert\hat{\Pi}_{G}^{-1/2}\Pi_{G}^{1/2}-I_{l}\right\Vert=o_{P}(1).$ 

Since $\Pi_{G}=RV_{G}R^{\prime},$ it follows from Lemma \ref{lemma:lambda_min lambda_max inequality} that 
\begin{align*}
\lambda_{\min}\left(\Pi_{G}\right) & \ge\lambda_{\min}\left(RR^{\prime}\right)\lambda_{\min} \left(V_{G}\right)=O_{e}(1)\lambda_{\min}\left(V_{G}\right) \\
\lambda_{\max}\left(\Pi_{G}\right) & \leq\lambda_{\max}\left(RR^{\prime}\right)\lambda_{\max}\left(V_{G}\right)=O_{e}(1)\lambda_{\max}\left(V_{G}\right).
\end{align*}

Since $O_{e}\left(\lambda_{\min}\left(V_{G}\right)\right)=O_{e}\left(\lambda_{\max}\left(V_{G}\right)\right)$ by Assumption \ref{assump:lambda max V_G=lambda min V_G}, we have  $\left\Vert\Pi_{G}\right\Vert=O_{e}\left(\left\Vert V_{G}\right\Vert\right).$
In the proof of Theorem \ref{thm: consistency of V hat}, we have shown that
$\frac{1}{\left\Vert V_{G}\right\Vert}\left\Vert\hat{V}_{G}-V_{G}\right\Vert=o_{P}(1),$
which implies $\frac{1}{O_{e}\left(\left\Vert V_{G}\right\Vert\right)}\left\Vert R\hat{V}_{G}R^{\prime}-RV_{G}R^{\prime}\right\Vert=o_{P}(1),$
which further implies that $\frac{1}{\left\Vert\Pi_{G}\right\Vert}\left\Vert\hat{\Pi}_{G}-\Pi_{G}\right\Vert=o_{P}(1).$
By Assumption \ref{assump:lambda max V_G=lambda min V_G}, we deduce that $\frac{V_{G}}{\left\Vert V_{G}\right\Vert}$
is uniformly positive definite, and it
subsequently follows that $\frac{\Pi_{G}}{\left\Vert\Pi_{G}\right\Vert}$
is also uniformly positive definite. So, by the continuity theorem for
probability convergence, taking $g(Z)=Z^{-1/2}\left(\frac{\Pi_G}{\left\Vert\Pi_G\right\Vert}\right)^{1/2}$ as a continuous function of $Z$ and putting $Z=\frac{\hat{\Pi}_G}{\left\Vert\Pi_G\right\Vert}$, we have $\left\Vert\hat{\Pi}_{G}^{-1/2}\Pi_{G}^{1/2}-I_{l}\right\Vert=o_{P}(1).$

Hence, we have shown that under $H_{0},$
\[
    \hat{\Pi}_{G}^{-1/2}\sqrt{G}\left(R\hat{\bar{\beta}}-r\right)\xrightarrow{d}N_{q}(0,I_{q}),\quad\text{ as } G\to\infty.
\]
 Consequently, by the continuous mapping theorem, it follows that
\[
G\left(R\hat{\bar{\beta}}-r\right)^{\prime}\hat{\Pi}_{G}^{-1}\left(R\hat{\bar{\beta}}-r\right)\xrightarrow{d}\chi_{q}^{2},\quad \text{ as } G\to\infty
\]
(i) for the nearly balanced case, (ii) for the unbalanced case and strong dependence, (iii) for the unbalanced case, and weak and semi-strong dependence, if $O(G)$ clusters are nearly balanced. This proves the theorem.

\hfill
\qed

{\bf Proof of Theorem \ref{thm:consistency under varying parameters}}

To prove the consistency of $\hat{\bar{\beta}}$ for model \eqref{eq: model2 for u_g}, first, note that from \eqref{eq: beta_hat for model2}
\[
\hat{\bar{\beta}} - \beta = \frac{1}{G}\sum_{g=1}^G u_g + \frac{1}{G}\sum_{g=1}^G \left(X_g^\prime X_g\right)^{-1}X_g^\prime \epsilon_g.
\]
Thus, $\hat{\bar{\beta}}$ is unbiased by Assumption \ref{assump:exogeneity} and \ref{assump: distribution of u_g}. Hence, to prove the consistency of $\hat{\bar{\beta}},$ it
suffices to show $\mathbb{E}\left\Vert\hat{\bar{\beta}}-\beta\right\Vert^{2}\rightarrow0,$
as $G\rightarrow\infty.$  Using the independence between $\epsilon_{g}$
and $u_{g}$ by Assumption \ref{assump: epsilon_g, X_g and u_g independent},
\[
\mathbb{E}\left\Vert \hat{\bar{\beta}}-\beta\right\Vert ^{2}=\operatorname{tr}\left(\mathbb{E}\left[\left(\hat{\bar{\beta}}-\beta\right)\left(\hat{\bar{\beta}}-\beta\right)^{\prime}\right]\right)=\operatorname{tr}\left(V^{*}\right),
\]

where $V^{*}=V+\frac{1}{G}\Delta$ as defined in \eqref{eq: V^*}. From the proof of Theorem \ref{thm:consistency}, we have already shown that for any kind of dependence, under Assumptions \ref{assump:exogeneity}, \ref{assump:error independence}, and \ref{assump:Xg'Xg/Ng inverse is bounded},
\[
\operatorname{tr}(V)=O\left(\frac{1}{G}\right),
\]

irrespective of the cluster sizes. Also, under Assumption \ref{assump: distribution of u_g},
\[
\operatorname{tr}(\Delta)=O\left(1\right).
\]

Therefore, 
\[
\mathbb{E}\left\Vert \hat{\bar{\beta}}-\beta\right\Vert ^{2}=O\left(\frac{1}{G}\right),
\]

which yields that $\hat{\bar{\beta}}$ is consistent for $\beta$ for any kind of dependence of $\epsilon_{g}$ irrespective of the cluster sizes.

\hfill 
\qed

{\bf Proof of Theorem \ref{thm:normal under varying parameters}}

To prove the asymptotic normality of $V^{*-1/2}\left(\hat{\bar{\beta}}-\beta\right),$
first note that $\hat{\bar{\beta}}-\beta=\frac{1}{G}\sum_{g=1}^{G}\left(X_{g}^{\prime}X_{g}\right)^{-1}X_{g}^{\prime}\epsilon_{g}+\frac{1}{G}\sum_{g=1}^{G}u_{g}$
as in \eqref{eq: beta_hat for model2}. Assuming $w_{g}=\left(X_{g}^{\prime}X_{g}\right)^{-1}X_{g}^{\prime}\epsilon_{g}$,
we obtain 
\[
V^{*-1/2}\left(\hat{\bar{\beta}}-\beta\right)=V^{*-1/2}\left[\frac{1}{G}\sum_{g=1}^{G}w_{g}+\frac{1}{G}\sum_{g=1}^{G}u_{g}\right].
\]

For the asymptotic normality, it suffices to prove the Liapounov condition
for CLT. To prove the Liapounov condition, it is required to show
that for any fixed $z\in\mathbb{R}^{k},$ and for some $p>1,$
\begin{align*}
\frac{\sum_{g=1}^{G}\mathbb{E}\left|z^{\prime}V^{*-1/2}\frac{1}{G}\left(w_{g}+u_{g}\right)\right|^{2p}}{\left(\sum_{g=1}^{G}\mathbb{E}\left[z^{\prime}V^{*-1/2}\frac{1}{G}\left(w_{g}+u_{g}\right)\right]^{2}\right)^{p}} & \to0,\;\text{ as }G\to\infty\\
\implies\frac{\sum_{g=1}^{G}\mathbb{E}\left|z^{\prime}V^{*-1/2}\left(w_{g}+u_{g}\right)\right|^{2p}}{\left(\sum_{g=1}^{G}\mathbb{E}\left[z^{\prime}V^{*-1/2}\left(w_{g}+u_{g}\right)\right]^{2}\right)^{p}} & \to0,\;\text{ as }G\to\infty.
\end{align*}

Using the trace inequality in the second line, the inequality $\left|a+b+c\right|^{p}\leq3^{p-1}\left(\left|a\right|^{p}+\left|b\right|^{p}+\left|c\right|^{p}\right)$
in the fourth line, observe that 
\begin{align*}
\sum_{g=1}^{G}\mathbb{E}\left|z^{\prime}V^{*-1/2}\left(w_{g}+u_{g}\right)\right|^{2p} & =\sum_{g=1}^{G}\mathbb{E}\left[\left|z^{\prime}V^{*-1/2}\left(w_{g}+u_{g}\right)\left(w_{g}+u_{g}\right)^{\prime}V^{*-1/2}z\right|^{p}\right]\\
 & \leq\sum_{g=1}^{G}\mathbb{E}\left[\left|z^{\prime}V^{*-1}z\right|^{p}\left|\operatorname{tr}\left(\left(w_{g}+u_{g}\right)\left(w_{g}+u_{g}\right)^{\prime}\right)\right|^{p}\right]\\
 & =\left|z^{\prime}V^{*-1}z\right|^{p}\sum_{g=1}^{G}\mathbb{E}\left[\left|w_{g}^{\prime}w_{g}+u_{g}^{\prime}u_{g}+2w_{g}^{\prime}u_{g}\right|^{p}\right]\\
 & \leq\left|z^{\prime}V^{*-1}z\right|^{p}\sum_{g=1}^{G}3^{p-1}\left(\mathbb{E}\left|w_{g}^{\prime}w_{g}\right|^{p}+\mathbb{E}\left|u_{g}^{\prime}u_{g}\right|^{p}+\mathbb{E}\left|2w_{g}^{\prime}u_{g}\right|^{p}\right).
\end{align*}

By the trace inequality, 
\begin{align*}
\mathbb{E}\left|w_{g}^{\prime}w_{g}\right|^{p} & \leq\mathbb{E}\left[\left|\operatorname{tr}\left(X_{g}\left(X_{g}^{\prime}X_{g}\right)^{-2}X_{g}^{\prime}\right)\right|^{p}\left|\epsilon_{g}^{\prime}\epsilon_{g}\right|^{p}\right]\\
 & =\mathbb{E}\left[\left|\operatorname{tr}\left(\left(X_{g}^{\prime}X_{g}\right)^{-1}\right)\right|^{p}\mathbb{E}\left[\left|\epsilon_{g}^{\prime}\epsilon_{g}\right|^{p}\big|X\right]\right].
\end{align*}

Since $\mathbb{E}\left[\left|\epsilon_{g}^{\prime}\epsilon_{g}\right|^{p}\big|X\right]=O\left(N_{g}^{p}\right)$
uniformly in $g$ follows by Assumption \ref{assump:existence of moment} from \eqref{eq: E(epsilon_g'epsilon_g)^p order} in the proof of Theorem \ref{thm:normal}, we obtain from Assumption \ref{assump:Xg'Xg/Ng inverse is bounded}
\[
\mathbb{E}\left|w_{g}^{\prime}w_{g}\right|^{p}=O(1),\quad\text{ uniformly in }g.
\]

Using the H\"{o}lder's inequality in the second line, and Assumption \ref{assump: existence of moment for ug} in the fourth line 
\begin{align*}
\mathbb{E}\left|u_{g}^{\prime}u_{g}\right|^{p} & =\mathbb{E}\left|\sum_{j=1}^{k}u_{gj}^{2}\right|^{p}\\
 & \leq\left(\sum_{j=1}^{k}1^{p^*}\right)^{p/p^*}\mathbb{E}\left[\sum_{j=1}^{k}\left|u_{gj}\right|^{2p}\right],\;\text{ where }\frac{1}{p}+\frac{1}{p^*}=1\\
 & =k^{p/p^*}\sum_{j=1}^{k}\mathbb{E}\left[\left|u_{gj}\right|^{2p}\right]\\
 & =O\left(k^{p}\right)=O(1),\quad\text{ uniformly in }g.
\end{align*}

Applying the C-S inequality in the first line, using the independence
between $w_{g}$ and $u_{g}$ by Assumption \ref{assump: epsilon_g, X_g and u_g independent} in the second line, and putting the orders of the previous terms, we obtain
\begin{align*}
\mathbb{E}\left|w_{g}^{\prime}u_{g}\right|^{p} & \leq\mathbb{E}\left[\left|w_{g}^{\prime}w_{g}\right|^{1/2}\left|u_{g}^{\prime}u_{g}\right|^{1/2}\right]^{p}\\
 & =\mathbb{E}\left[\left|w_{g}^{\prime}w_{g}\right|^{p/2}\right]\mathbb{E}\left[\left|u_{g}^{\prime}u_{g}\right|^{p/2}\right]\\
 & =O(1),\quad\text{ uniformly in }g.
\end{align*}

Therefore, for some $0<c_{1}<\infty,$ by Lemma \ref{lemma:Exact order summability}
\[
\sum_{g=1}^{G}\mathbb{E}\left|z^{\prime}V^{*-1/2}\left(w_{g}+u_{g}\right)\right|^{2p}\leq c_{1}G\left|z^{\prime}V^{*-1}z\right|^{p}.
\]

Note that for any $z\in\mathbb{R}^{k},$
\[
\frac{z^{\prime}V^{*-1}z}{z^{\prime}z}\leq\lambda_{max}\left(V^{*-1}\right)=\frac{1}{\lambda_{min}\left(V^{*}\right)}.
\]

Since $\lambda_{min}\left(V^{*}\right)\geq\lambda_{min}\left(V\right)+\lambda_{min}\left(\frac{1}{G}\Delta\right)\geq\lambda_{min}\left(\frac{1}{G}\Delta\right)\geq\frac{c_{2}}{G},$
for some $0<c_{2}<\infty$ by Assumption \ref{assump: distribution of u_g}, it follows
that \textbf{
\[
\frac{z^{\prime}V^{*-1}z}{z^{\prime}z}\leq\frac{G}{c_{2}}.
\]
}

For the denominator, putting the expression of $V^*$ from \eqref{eq: V^*} in the second
line, we have 
\begin{align*}
\left(\sum_{g=1}^{G}\mathbb{E}\left[z^{\prime}V^{*-1/2}\left(w_{g}+u_{g}\right)\right]^{2}\right)^{p} & =\left(z^{\prime}V^{*-1/2}\sum_{g=1}^{G}\mathbb{E}\left[\left(w_{g}+u_{g}\right)\left(w_{g}+u_{g}\right)^{\prime}\right]V^{*-1/2}z\right)^{p}\\
 & =\left(z^{\prime}V^{*-1/2}G^{2}V^{*}V^{*-1/2}z\right)^{p}\\
 & =G^{2p}\left(z^{\prime}z\right)^{p}.
\end{align*}

Therefore, the Liapounov condition reduces to 
\[
\frac{\sum_{g=1}^{G}\mathbb{E}\left|z^{\prime}V^{*-1/2}\frac{1}{G}\left(w_{g}+u_{g}\right)\right|^{2p}}{\left(\sum_{g=1}^{G}\mathbb{E}\left[z^{\prime}V^{*-1/2}\frac{1}{G}\left(w_{g}+u_{g}\right)\right]^{2}\right)^{p}}\leq c_{1}G^{1-2p}\left|\frac{z^{\prime}V^{*-1}z}{z^{\prime}z}\right|^{p}\leq\frac{c_{1}}{c_{2}}G^{2(1-p)},
\]

which goes to $0,$ as $G\to\infty,$ for some $p>1.$ Hence, the CLT holds for $V^{*-1/2}\left(\hat{\bar{\beta}}-\beta\right),$
for strong, semi-strong, and weak dependence and for any cluster sizes.

\hfill 
\qed

\textbf{Proof of Theorem \ref{thm: V hat consistency for model 2}}

To show that $\hat{V}^{*}_G$ is consistent for $V_G^*,$ it suffices to show that 
\[
    \mathbb{E}\left[ \frac{\| \hat{V}^{*}_G - V^*_G \|}{\| V^*_G \|} \right] \to 0, \;\text{ as } G\to\infty.
\]
To show this, first, note that $e_g = Y_g - X_g \hat{\bar{\beta}} = \epsilon_g + X_g u_g - X_g (\hat{\bar{\beta}} - \beta),$ which implies that
\begin{align*}
    e_g e_g' = & (\epsilon_g + X_g u_g)(\epsilon_g + X_g u_g)' + X_g (\hat{\bar{\beta}} - \beta)(\hat{\bar{\beta}} - \beta)' X_g' \\
    & - X_g (\hat{\bar{\beta}} - \beta)(\epsilon_g + X_g u_g)' - (\epsilon_g + X_g u_g)(\hat{\bar{\beta}} - \beta)' X_g'. 
\end{align*}

Then, $\hat{V}^{*}_G - V^*_G$ can be expressed as:
\begin{align*}
    \hat{V}^{*}_G - V^*_G & =  \frac{1}{G}\sum_{g=1}^{G}\left(X_{g}^{\prime}X_{g}\right)^{-1}X_{g}^{\prime}e_{g}e_{g}^{\prime}X_{g}\left(X_{g}^{\prime}X_{g}\right)^{-1} - V_G^* \\
    & =  T_1 + T_2 + T_3 + T_4, \quad \text{say,}
\end{align*}
where \[
T_1 =   \frac{1}{G} \sum_{g=1}^{G} (X_g' X_g)^{-1} X_g' \left[(\epsilon_g + X_g u_g)(\epsilon_g + X_g u_g)'\right] X_g (X_g' X_g)^{-1} - V_G^*, 
\]
\[
T_2 = \frac{1}{G} \sum_{g=1}^{G} (\hat{\bar{\beta}} - \beta) (\hat{\bar{\beta}} - \beta)',
\]
\[
T_3 = -\frac{1}{G} \sum_{g=1}^{G} (\hat{\bar{\beta}} - \beta)(\epsilon_g + X_g u_g)' X_g (X_g' X_g)^{-1}, \quad \text{and}
\]
\[
T_4 = -\frac{1}{G} \sum_{g=1}^{G} (X_g' X_g)^{-1} X_g' (\epsilon_g + X_g u_g) (\hat{\bar{\beta}} - \beta)' .
\]

First, note that using Assumption \ref{assump:exogeneity}, Assumption \ref{assump: epsilon_g, X_g and u_g independent} and Assumption \ref{assump: distribution of u_g}
\[
    \mathbb{E}\mathbb{E}\left[T_1\right] = \mathbb{E}\left[\frac{1}{G} \sum_{g=1}^{G} (X_g' X_g)^{-1} X_g' \left[(\epsilon_g + X_g u_g)(\epsilon_g + X_g u_g)'\right] X_g (X_g' X_g)^{-1} \bigg| X \right] -V_G^* = 0. 
\]

Now, conditional on $X,$
\[ 
T_1 = \frac{1}{G} \sum_{g=1}^{G} (X_g' X_g)^{-1} X_g' \left[(\epsilon_g + X_g u_g)(\epsilon_g + X_g u_g)' - (\Omega_g + X_g \Delta X_g')\right] X_g (X_g' X_g)^{-1} \]

 Assuming $A_g = (\epsilon_g + X_g u_g)(\epsilon_g + X_g u_g)' - (\Omega_g + X_g \Delta X_g'),$ and using the independence between clusters by Assumption \ref{assump:error independence}, $\mathbb{E} \operatorname{tr}[T_1 T_1']$ can be written as
\begin{align*}
    \mathbb{E} \operatorname{tr}[T_1 T_1'] & = \frac{1}{G^2} \sum_{g=1}^{G} \mathbb{E}\left[ (X_g' X_g)^{-1} X_g' A_g X_g (X_g' X_g)^{-2} X_g' A_g' X_g (X_g' X_g)^{-1}\right] \\
    & \leq \frac{1}{G^2} \sum_{g=1}^{G} \mathbb{E} \left[\operatorname{tr}\left( X_g (X_g' X_g)^{-2} X_g'\right) \cdot \mathbb{E}\operatorname{tr}[ (X_g' X_g)^{-1} X_g' A_g A_g' X_g (X_g' X_g)^{-1} \big| X] \right],\; \text{ by trace inequality} \\
    & \leq \frac{1}{G^2} \sum_{g=1}^{G} \mathbb{E} \left[\operatorname{tr}^2\left((X_g' X_g)^{-1}\right) \cdot \mathbb{E}\operatorname{tr}[ A_g A_g' \big| X] \right],\; \text{ by trace inequality.} 
\end{align*}

Now, 
\[ \mathbb{E}\operatorname{tr}[A_g A_g'] = \mathbb{E}\operatorname{tr} \left[ ( (\epsilon_g + X_g u_g)(\epsilon_g + X_g u_g)' - (\Omega_g + X_g \Delta X_g') )^2 \right] \]
After expanding the square and taking expectations involving fourth moments,
\begin{align*}
    & \; \mathbb{E}\operatorname{tr}[A_g A_g'] \\
   = &  \; \mathbb{E} \left[ \{ (\epsilon_g + X_g u_g)' (\epsilon_g + X_g u_g) \}^2 \right] - \operatorname{tr}[(\Omega_g + X_g \Delta X_g')^2] \\
    = &  \; \mathbb{E} \left[ (\epsilon_g' \epsilon_g)^2 + (u_g' X_g' X_g u_g)^2 + 4(\epsilon_g' X_g u_g)^2 + 2 \epsilon_g'\epsilon_g u_g'X_g'X_g u_g + 2 u_g' X_g'X_g u_g \epsilon_g'X_g u_g + 2\epsilon_g'\epsilon_g \epsilon_g'X_gu_g\right] \\
      &  \; -\operatorname{tr}[(\Omega_g + X_g \Delta X_g')^2].
\end{align*}
Note that by Assumption \ref{assump:existence of moment}, and Assumption \ref{assump: existence of moment for ug}, all the above terms are $O(N_g^2)$. 
Applying the inequality $\mathbb{E}\|T_1\| \leq (\mathbb{E} \operatorname{tr}[T_1 T_1'])^{1/2},$ using Assumption \ref{assump:Xg'Xg/Ng inverse is bounded} and Lemma \ref{lemma:Exact order summability}, we obtain
\begin{align*}
    \mathbb{E}\|T_1\| & \leq \frac{1}{G} \left( \sum_{g=1}^{G} O(N_g^{-2}) \cdot O(N_g^2) \right)^{1/2} \\
     &  = O\left(\frac{1}{\sqrt{G}}\right), \quad \text{for any kind of dependence.} 
\end{align*}

Note that $T_2 = (\hat{\bar{\beta}} - \beta)(\hat{\bar{\beta}} - \beta)'$. Since $V_G^* = V_G + \Delta$, so by Assumption \ref{assump:lambda max V_G=lambda min V_G} and Assumption \ref{assump: distribution of u_g},
\[
\mathbb{E}\|T_2\| =  \mathbb{E}\left| (\hat{\bar{\beta}} - \beta)' (\hat{\bar{\beta}} - \beta) \right| = \operatorname{tr}(V^*) = \frac{1}{G}\operatorname{tr}(V_G^*) = O\left(\frac{1}{G}\right).
\] 

Note that 
\[
T_{3} = -\frac{1}{G} \sum_{g=1}^{G} (\hat{\bar{\beta}} - \beta)(\epsilon_g + X_g u_g)' X_g (X_g' X_g)^{-1} = -\left(\hat{\bar{\beta}}-\beta\right)\left(\hat{\bar{\beta}}-\beta\right)^{\prime}=-T_{2},
\]
and similarly, $T_{4}=T_{3}.$ Hence, $\left\Vert T_{3}\right\Vert =\left\Vert T_{2}\right\Vert =\left\Vert T_{4}\right\Vert .$
By the triangular inequality for norms
\[
\frac{\left\Vert \hat{V}^{*}_{G}-V^*_{G}\right\Vert }{\left\Vert V^*_{G}\right\Vert }  \leq  \frac{1}{\left\Vert V^*_{G}\right\Vert }\left[\left\Vert T_{1}\right\Vert +\left\Vert T_{2}\right\Vert +\left\Vert T_{3}\right\Vert +\left\Vert T_{4}\right\Vert \right].
\]

Also, $\|V_G^*\| = \| V_G + \Delta \| \geq \| \Delta\| = O_e(1)$, by Assumption \ref{assump: distribution of u_g}. Therefore,
\[ \mathbb{E}\left[ \frac{\| \hat{V}^{*}_G - V^*_G \|}{\| V^*_G \|} \right] = O\left(\frac{1}{\sqrt{G}}\right) \to 0, \quad \text{as } G \to \infty. \]

Hence, $\hat{V}^{*}_G$ is consistent for $V_G^*,$ under model \eqref{eq: model2 for u_g}.

\hfill 
\qed

\textbf{Proof of Theorem \ref{thm: Rb = r testing for model2}}

We show the asymptotic distribution of the Wald-type test statistic based on $\hat{\bar{\beta}}$ for model \eqref{eq: model2 for u_g} in two steps.

\textit{Step 1}: In this step, we show that under $H_{0},$
\[
    \Gamma_{G}^{-1/2}\sqrt{G}\left(R\hat{\bar{\beta}}-r\right)\xrightarrow{d}N_{q}(0,I_{q}), \;\text{ as } G\rightarrow\infty.
\]

First note that under $H_{0},$
\[
    \Gamma_{G}^{-1/2}\sqrt{G}\left(R\hat{\bar{\beta}}-r\right)=\Gamma_{G}^{-1/2}\sqrt{G}R\left(\hat{\bar{\beta}}-\beta\right)=\Gamma_{G}^{-1/2}RV_{G}^{*1/2}V_{G}^{*-1/2}\sqrt{G}\left(\hat{\bar{\beta}}-\beta\right).
\]
Note that Theorem \ref{thm:normal under varying parameters} yields $V_{G}^{*-1/2}\sqrt{G}\left(\hat{\bar{\beta}}-\beta\right)\xrightarrow{d}N_{k}(0, I_{k}).$ Also, $V_G^* = V_G + \Delta$ is finite and positive definite due to the fact that $V_G$ and $\Delta$ are both finite and positive definite by Assumption \ref{assump:lambda max V_G=lambda min V_G} and Assumption \ref{assump: distribution of u_g}, respectively. 
Since $R$ is a finite matrix of dimension $q\times k$ with full
row rank $q,$ and $V^{*}_{G}$ is finite and positive definite, it follows
from Corollary 4.24 of \cite{white1984asymptotic} that 
\[
    \Gamma_{G}^{-1/2}RV_{G}^{*1/2}V_{G}^{*-1/2}\sqrt{G}\left(\hat{\bar{\beta}}-\beta\right)\xrightarrow{d}N_{q}(0,I_{q}),\;\text{ as } G\rightarrow\infty.
\]

\textit{Step 2}: In this step, we show that under $H_{0},$
\[
    \hat{\Gamma}_{G}^{-1/2}\sqrt{G}\left(R\hat{\bar{\beta}}-r\right)\xrightarrow{d}N_{q}(0,I_{q}),\;\text{ as } G\rightarrow\infty.
\] 

To show this, first note that using the norm inequality,
\[
\left\Vert\hat{\Gamma}_{G}^{-1/2}\sqrt{G}\left(R\hat{\bar{\beta}}-r\right)-\Gamma_{G}^{-1/2}\sqrt{G}\left(R\hat{\bar{\beta}}-r\right)\right\Vert \leq \left\Vert\hat{\Gamma}_{G}^{-1/2}\Gamma_{G}^{1/2}-I_{l}\right\Vert\left\Vert\Gamma_{G}^{-1/2}\sqrt{G}\left(R\hat{\bar{\beta}}-r\right)\right\Vert.
\]
 From the Step 1, we have shown that under $H_0,$ $\Gamma_{G}^{-1/2}\sqrt{G}\left(R\hat{\bar{\beta}}-r\right)\xrightarrow{d}N_{q}(0,I_{q}),$ which deduces that
 \[
    \left\Vert\Gamma_{G}^{-1/2}\sqrt{G}\left(R\hat{\bar{\beta}}-r\right)\right\Vert=O_P(1),\quad\text{ under } H_{0}.
 \]
 
Therefore, it is enough to show that $\left\Vert\hat{\Gamma}_{G}^{-1/2}\Gamma_{G}^{1/2}-I_{l}\right\Vert=o_{P}(1).$ 

Since $\Gamma_{G}=RV^{*}_{G}R^{\prime},$ it follows from Lemma \ref{lemma:lambda_min lambda_max inequality} that 
\begin{align*}
\lambda_{\min}\left(\Gamma_{G}\right) & \ge\lambda_{\min}\left(RR^{\prime}\right)\lambda_{\min} \left(V^{*}_{G}\right)=O_{e}(1)\lambda_{\min}\left(V^{*}_{G}\right) \\
\lambda_{\max}\left(\Gamma_{G}\right) & \leq\lambda_{\max}\left(RR^{\prime}\right)\lambda_{\max}\left(V^{*}_{G}\right)=O_{e}(1)\lambda_{\max}\left(V^{*}_{G}\right).
\end{align*}

The following holds for any kind of dependence by Assumptions \ref{assump:lambda max V_G=lambda min V_G} and \ref{assump: distribution of u_g}:
\begin{align*}
    \lambda_{\max}(V_G^*) & \leq \lambda_{\max}(V_G) + \lambda_{\max}(\Delta) = O_e(1), \quad \text{ and } \\
    \lambda_{\min}(V_G^*) & \geq \lambda_{\min}(V_G) + \lambda_{\min}(\Delta) \geq \lambda_{\min}(\Delta) = O_e(1).
\end{align*}
Therefore, 
$O_{e}\left(\lambda_{\min}\left(V^{*}_{G}\right)\right)=O_{e}\left(\lambda_{\max}\left(V^{*}_{G}\right)\right)$ for any kind of dependence, and consequently,   $\left\Vert\Gamma_{G}\right\Vert=O_{e}\left(\left\Vert V^{*}_{G}\right\Vert\right).$
In the proof of Theorem \ref{thm: V hat consistency for model 2}, we have shown that
$\frac{1}{\left\Vert V^{*}_{G}\right\Vert}\left\Vert\hat{V}^{*}_{G}-V^{*}_{G}\right\Vert=o_{P}(1),$
which implies 
\[
\frac{1}{O_{e}\left(\left\Vert V^{*}_{G}\right\Vert\right)}\left\Vert R\hat{V}^{*}_{G}R^{\prime}-RV^{*}_{G}R^{\prime}\right\Vert=o_{P}(1),
\]
which further implies that $\frac{1}{\left\Vert\Gamma_{G}\right\Vert}\left\Vert\hat{\Gamma}_{G}-\Gamma_{G}\right\Vert=o_{P}(1).$
By Assumptions \ref{assump:lambda max V_G=lambda min V_G} and \ref{assump: distribution of u_g}, we deduce that $\frac{V^{*}_{G}}{\left\Vert V^{*}_{G}\right\Vert}$
is uniformly positive definite, and it
subsequently follows that $\frac{\Gamma_{G}}{\left\Vert\Gamma_{G}\right\Vert}$
is also uniformly positive definite. So, by the continuity theorem for
probability convergence, taking $g(Z)=Z^{-1/2}\left(\frac{\Gamma_G}{\left\Vert\Gamma_G\right\Vert}\right)^{1/2}$ as a continuous function of $Z$ and putting $Z=\frac{\hat{\Gamma}_G}{\left\Vert\Gamma_G\right\Vert}$, we have $\left\Vert\hat{\Gamma}_{G}^{-1/2}\Gamma_{G}^{1/2}-I_{l}\right\Vert=o_{P}(1).$

Hence, we have shown that under $H_{0},$
\[
    \hat{\Gamma}_{G}^{-1/2}\sqrt{G}\left(R\hat{\bar{\beta}}-r\right)\xrightarrow{d}N_{q}(0,I_{q}),\quad\text{ as } G\to\infty.
\]
 Consequently, by the continuous mapping theorem, it follows that
\[
G\left(R\hat{\bar{\beta}}-r\right)^{\prime}\hat{\Gamma}_{G}^{-1}\left(R\hat{\bar{\beta}}-r\right)\xrightarrow{d}\chi_{q}^{2},\quad \text{ as } G\to\infty.
\]

\hfill 
\qed

\textbf{Proof of Lemma \ref{lemma: V_l hat consistency}}

To show that $\Tilde{V}_{l}$ is consistent for $V_{G,l},$ it suffices to show that 
\[
    \mathbb{E}\left[\frac{\left\Vert \Tilde{V}_{l}-V_{G,l}\right\Vert }{\left\Vert V_{G,l}\right\Vert}\right] \to 0, \quad \text{ as } G\to\infty.
\]

To show this, first, note that $e_{g}^{*}=Y_{g}-X_{g}\Tilde{\beta}_{l}=\epsilon_{g}-X_{g}\left(\Tilde{\beta}_{l}-\beta\right),$
which implies 
\[    e_{g}^{*}e_{g}^{*\prime}=\epsilon_{g}\epsilon_{g}^{\prime}+X_{g}\left(\Tilde{\beta}_{l}-\beta\right)\left(\Tilde{\beta}_{l}-\beta\right)^{\prime}X_{g}^{\prime}-\epsilon_{g}\left(\Tilde{\beta}_{l}-\beta\right)^{\prime}X_{g}^{\prime}-X_{g}\left(\Tilde{\beta}_{l}-\beta\right)\epsilon_{g}^{\prime}.
\]

Then, $\Tilde{V}_{l}-V_{G,l}$ can be expressed as:
\[
\Tilde{V}_{l}-V_{G,l}=T_{1}+T_{2}+T_{3}+T_{4},\quad \text{ say,}
\]
 where
 \[
    T_{1}=\frac{1}{P_{l}^{2}}\sum_{g\in \mathcal{S}_{l}}\left(\left(X_{g}^{\prime}X_{g}\right)^{-1}X_{g}^{\prime}\epsilon_{g}\epsilon_{g}^{\prime}X_{g}\left(X_{g}^{\prime}X_{g}\right)^{-1}-\mathbb{E}\left[\left(X_{g}^{\prime}X_{g}\right)^{-1}X_{g}^{\prime}\Omega_{g}X_{g}\left(X_{g}^{\prime}X_{g}\right)^{-1}\right]\right),
 \] 
\[
    T_{2}=\frac{1}{P_{l}^{2}}\sum_{g\in \mathcal{S}_{l}}\left(\Tilde{\beta}_{l}-\beta\right)\left(\Tilde{\beta}_{l}-\beta\right)^{\prime},
\]
\[
    T_{3}=-\frac{1}{P_{l}^{2}}\sum_{g\in \mathcal{S}_{l}}\left(X_{g}^{\prime}X_{g}\right)^{-1}X_{g}^{\prime}\epsilon_{g}\left(\Tilde{\beta}_{l}-\beta\right)^{\prime}, \quad \text{ and }
\]
\[
T_{4}=-\frac{1}{P_{l}^{2}}\sum_{g\in \mathcal{S}_{l}}\left(\Tilde{\beta}_{l}-\beta\right)\epsilon_{g}^{\prime}X_{g}\left(X_{g}^{\prime}X_{g}\right)^{-1}.
\]

Now, conditional on $X,$ observe that 
\begin{align*}
\mathbb{E}\left[T_{1}\right] & =  \mathbb{E}\mathbb{}\left[T_{1}\big|X\right]\\
& = \frac{1}{P_{l}^{2}}\sum_{g\in \mathcal{S}_{l}}\mathbb{E}\left[\left(X_{g}^{\prime}X_{g}\right)^{-1}X_{g}^{\prime}\mathbb{E}\left[\epsilon_{g}\epsilon_{g}^{\prime}\big|X\right]X_{g}\left(X_{g}^{\prime}X_{g}\right)^{-1}-\mathbb{E}\left[\left(X_{g}^{\prime}X_{g}\right)^{-1}X_{g}^{\prime}\Omega_{g}X_{g}\left(X_{g}^{\prime}X_{g}\right)^{-1}\right]\right]\\
& = 0,
\end{align*}
and using the independence between clusters by Assumption \ref{assump:error independence} in the second line, and the norm inequality in the third line, we have 
\begin{align*}
\mathbb{E}\left\Vert T_{1}\right\Vert ^{2} & =  \mathbb{E}\left\Vert \frac{1}{P_{l}^{2}}\sum_{g\in \mathcal{S}_{l}}\left(X_{g}^{\prime}X_{g}\right)^{-1}\left(X_{g}^{\prime}\epsilon_{g}\epsilon_{g}^{\prime}X_{g}-X_{g}^{\prime}\Omega_{g}X_{g}\right)\left(X_{g}^{\prime}X_{g}\right)^{-1}\right\Vert ^{2}\\
& = \frac{1}{P_{l}^{4}}\sum_{g\in \mathcal{S}_{l}}\mathbb{E}\mathbb{}\left[\left\Vert \left(X_{g}^{\prime}X_{g}\right)^{-1}X_{g}^{\prime}\left(\epsilon_{g}\epsilon_{g}^{\prime}-\Omega_{g}\right)X_{g}\left(X_{g}^{\prime}X_{g}\right)^{-1}\right\Vert ^{2}\bigg|X\right]\\
& \leq \frac{1}{P_{l}^{4}}\sum_{g\in \mathcal{S}_{l}}\mathbb{E}\left[\left\Vert \left(X_{g}^{\prime}X_{g}\right)^{-1}\right\Vert ^{4}\mathbb{E}\left[\left\Vert X_{g}^{\prime}\left(\epsilon_{g}\epsilon_{g}^{\prime}-\Omega_{g}\right)X_{g}\right\Vert ^{2}\Big|X\right]\right]
\end{align*}

Now, using  $\mathbb{E}\left[\left\Vert \left(X_{g}^{\prime}X_{g}\right)^{-1}\right\Vert ^{4}\right]=O\left(N_{g}^{-4}\right)$
uniformly in $g$ by Assumption \ref{assump:Xg'Xg/Ng inverse is bounded}, and applying Assumption \ref{assump:Order of second moment of Xg'Omega_g Xg} and Lemma \ref{lemma:Exact order summability}, we obtain
\[
\mathbb{E}\left\Vert T_{1}\right\Vert ^{2}=\begin{cases}
\frac{1}{P_{l}^{4}}\sum_{g\in \mathcal{S}_{l}}O\left(N_{g}^{-4}\right)O\left(N_{g}^{4}\right)=O\left(\frac{1}{P_{l}^{3}}\right), & \text{for strong dependence}\\
\frac{1}{P_{l}^{4}}\sum_{g\in \mathcal{S}_{l}}O\left(N_{g}^{-4}\right)O\left(N_{g}^{2}h^{2}\left(N_{g}\right)\right)=O\left(\frac{1}{P_{l}^{4}}\sum_{g\in \mathcal{S}_{l}}\frac{h^{2}\left(N_{g}\right)}{N_{g}^{2}}\right), & \text{for semi-strong dependence}\\
\frac{1}{P_{l}^{4}}\sum_{g\in \mathcal{S}_{l}}O\left(N_{g}^{-4}\right)O\left(N_{g}^{2}\right)=O\left(\frac{1}{P_{l}^{4}}\sum_{g\in \mathcal{S}_{l}}\frac{1}{N_{g}^{2}}\right), & \text{for weak dependence}
\end{cases}
\]
Therefore,
\[
\left\Vert T_{1}\right\Vert =\begin{cases}
O_{P}\left(\frac{1}{P_{l}^{3/2}}\right), & \text{for strong dependence}\\
O_{P}\left(\frac{1}{P_{l}^{2}}\sqrt{\sum_{g\in \mathcal{S}_{l}}\frac{h^{2}\left(N_{g}\right)}{N_{g}^{2}}}\right), & \text{for semi-strong dependence}\\
O_{P}\left(\frac{1}{P_{l}^{2}}\sqrt{\sum_{g\in \mathcal{S}_{l}}\frac{1}{N_{g}^{2}}}\right), & \text{for weak dependence}
\end{cases}
\]
 
uniformly in $l.$ Now, for the second term, we have 
\[
\mathbb{E}\left\Vert T_{2}\right\Vert =\mathbb{E}\left\Vert \frac{1}{P_{l}}\left(\Tilde{\beta}_{l}-\beta\right)\left(\Tilde{\beta}_{l}-\beta\right)^{\prime}\right\Vert =\frac{1}{P_{l}}\mathbb{E}\left[\operatorname{tr}\left(V_{G,l}\right)\right].
\]

By Assumption  Assumption \ref{assump:error independence}, \ref{assump:Xg'Xg/Ng inverse is bounded}, and Lemma \ref{lemma:Exact order summability},  
\[
\mathbb{E}\left[\operatorname{tr}\left(V_{G,l}\right)\right] = 
\begin{cases}
O\left(\frac{1}{P_{l}}\right), & \text{for strong dependence}\\
O\left(\frac{1}{P_{l}^{2}}\sum_{g\in \mathcal{S}_{l}}\frac{h\left(N_{g}\right)}{N_{g}}\right), & \text{for semi-strong dependence}\\
O\left(\frac{1}{P_{l}^{2}}\sum_{g\in \mathcal{S}_{l}}\frac{1}{N_{g}}\right), & \text{for weak dependence}
\end{cases}
\]

uniformly in $l.$ This yields 
\[
\left\Vert T_{2}\right\Vert =\begin{cases}
O_{P}\left(\frac{1}{P_{l}^{2}}\right), & \text{for strong dependence}\\
O_{P}\left(\frac{1}{P_{l}^{3}}\sum_{g\in \mathcal{S}_{l}}\frac{h\left(N_{g}\right)}{N_{g}}\right), & \text{for semi-strong dependence}\\
O_{P}\left(\frac{1}{P_{l}^{3}}\sum_{g\in \mathcal{S}_{l}}\frac{1}{N_{g}}\right), & \text{for weak dependence}
\end{cases}
\]

Note that 
\[
T_{3}=-\frac{1}{P_{l}^{2}}\sum_{g\in \mathcal{S}_{l}}\left(X_{g}^{\prime}X_{g}\right)^{-1}X_{g}^{\prime}\epsilon_{g}\left(\hat{\bar{\beta}}-\beta\right)^{\prime}=-T_{2}.
\]
and $T_{4}=T_{3}.$ Hence, $\left\Vert T_{3}\right\Vert =\left\Vert T_{2}\right\Vert =\left\Vert T_{4}\right\Vert .$
By the triangular inequality for norms, 
\[
\frac{\left\Vert \Tilde{V}_{l}-V_{G,l}\right\Vert }{\left\Vert V_{G,l}\right\Vert }  \leq  \frac{1}{\left\Vert V_{G,l}\right\Vert }\left[\left\Vert T_{1}\right\Vert +\left\Vert T_{2}\right\Vert +\left\Vert T_{3}\right\Vert +\left\Vert T_{4}\right\Vert \right].
\]

By Assumption \ref{assump: lambda_max Vl = lambda_min Vl}, 
\[
\left\Vert V_{G,l}\right\Vert = \begin{cases}
O_e\left(\frac{1}{P_{l}}\right), & \text{for strong dependence}\\
O_e\left(\frac{1}{P_{l}^{2}}\sum_{g\in \mathcal{S}_{l}}\frac{h\left(N_{g}\right)}{N_{g}}\right), & \text{for semi-strong dependence}\\
O_e\left(\frac{1}{P_{l}^{2}}\sum_{g\in \mathcal{S}_{l}}\frac{1}{N_{g}}\right), & \text{for weak dependence}
\end{cases}
\]
Therefore, using Lemma \ref{lemma: sum h^2/N^2 to 0}
\begin{align*}
\frac{\left\Vert T_{1}\right\Vert }{\left\Vert V_{G,l}\right\Vert } & =\begin{cases}
O_{P}\left(\frac{1}{\sqrt{P_l}}\right), & \text{for strong dependence}\\
O_{P}\left(\frac{\sqrt{\sum_{g\in \mathcal{S}_{l}}\frac{h^{2}\left(N_{g}\right)}{N_{g}^{2}}}}{\sum_{g\in \mathcal{S}_{l}}\frac{h\left(N_{g}\right)}{N_{g}}}\right), & \text{for semi-strong dependence}\\
O_{P}\left(\frac{\sqrt{\sum_{g\in \mathcal{S}_{l}}\frac{1}{N_{g}^{2}}}}{\sum_{g\in \mathcal{S}_{l}}\frac{1}{N_{g}}}\right), & \text{for weak dependence}
\end{cases}\\
 & =o_{P}(1),\quad \text{ for any kind of dependence, since } P_l\to\infty.
\end{align*}

Also, for any kind of dependence,
\[
\frac{\left\Vert T_{2}\right\Vert }{\left\Vert V_{G,l}\right\Vert }=O_{P}\left(\frac{1}{P_l}\right)=o_{P}(1),\quad \text{ as } P_l\to\infty.
\]

Hence, for any kind of error dependence,
\[
\frac{\left\Vert \Tilde{V}_{l}-V_{G,l}\right\Vert }{\left\Vert V_{G,l}\right\Vert }=o_{P}(1),\quad \text{ as } P_l\to\infty.
\]
which implies $\Tilde{V}_{l}$ is consistent for $V_{G,l},$ for every $l=1,2,\ldots,D.$

\hfill 
\qed

\textbf{Proof of Theorem \ref{thm: testing Delta=0}} 

We show the asymptotic normality of the test statistic in two steps.

\textit{Step 1}: Let $T_{SB}^* = \sum_{l=1}^D (\Tilde{\beta}_{l} - \hat{\bar{\beta}})' V_{G,l}^{-1} (\Tilde{\beta}_{l} - \hat{\bar{\beta}})$.
We aim to show that
$$ Z^* = \frac{T_{SB}^* - kD}{\sqrt{2kD}} \xrightarrow[H_0]{d} N(0,1), \;\text{ as } \min_l P_l \to \infty, \, D \to \infty. $$

Expanding the quadratic form:
\begin{align*}
T_{SB}^* & = \sum_{l=1}^D (\Tilde{\beta}_{l} - \beta)' V_{G,l}^{-1} (\Tilde{\beta}_{l} - \beta) + \sum_{l=1}^D (\hat{\bar{\beta}} - \beta)' V_{G,l}^{-1} (\hat{\bar{\beta}} - \beta) - 2 \sum_{l=1}^D (\Tilde{\beta}_{l} - \beta)' V_{G,l}^{-1} (\hat{\bar{\beta}} - \beta) \\
&= A_D + B_D + C_D, \; \text{ say.}
\end{align*}

We begin by analyzing the term $A_D.$
Following the proof of Theorem \ref{thm:normal}, it is easy to observe that under Assumptions \ref{assump:exogeneity}, \ref{assump:error independence}, \ref{assump:Xg'Xg/Ng inverse is bounded}, \ref{assump:existence of moment} and \ref{assump: lambda_max Vl = lambda_min Vl},
\begin{equation} \label{eq: normality for Vl^-1/2(beta_l - beta)}
    V_{G,l}^{-1/2}(\Tilde{\beta}_{l} - \beta) \xrightarrow{d} N(0, I_k), \quad \text{ as } P_l \to \infty,
\end{equation}
for: (i) strong dependence generally, (ii) semi-strong dependence, if for some $p>1,$ $\frac{P_l}{\left[\sum_{g\in \mathcal{S}_l}\frac{h\left(N_{g}\right)}{N_{g}}\right]^{p}}\to0,$
as $P_l\to\infty.$ (iii) weak dependence, if for some $p>1,$ $\frac{P_l}{\left[\sum_{g\in \mathcal{S}_l}\frac{1}{N_{g}}\right]^{p}}\to0,$ as $P_l\to\infty.$

So, by the continuous mapping theorem, for fixed $D$, $A_D \xrightarrow{d} \chi^2_{kD}$. Therefore, using the CLT for the sum of independent chi-squares, we obtain
\[
    \frac{A_D - kD}{\sqrt{2kD}} \xrightarrow{d} N(0,1), \quad \text{ as } \min_l P_l \to \infty \text{ and } D \to \infty.
\]

For the term $B_D,$ note that by trace inequality
\begin{equation*}
    \mathbb{E}\left|B_D\right|  = \mathbb{E}\left[\sum_{l=1}^D (\hat{\bar{\beta}} - \beta)' V_{G,l}^{-1} (\hat{\bar{\beta}} - \beta)\right] = \sum_{l=1}^D \operatorname{tr}\left(V_{G,l}^{-1} \, \mathbb{E}\left[(\hat{\bar{\beta}} - \beta) (\hat{\bar{\beta}} - \beta)'\right]\right) \leq \sum_{l=1}^D \operatorname{tr}(V_{G,l}^{-1}) \, \operatorname{tr} (V).
\end{equation*}
By Assumption \ref{assump:lambda max V_G=lambda min V_G}, Assumption \ref{assump: lambda_max Vl = lambda_min Vl}, and Lemma \ref{lemma:Exact order summability},
\begin{align*}
     \mathbb{E}\left|B_D\right| & \leq \sum_{l=1}^D \operatorname{tr}(V_{G,l}^{-1}) \, \operatorname{tr} (V) \\
     & = \begin{cases}
\sum_{l=1}^D O(P_l) \cdot O(\frac{1}{G})  & , \text{for strong dependence} \\
\sum_{l=1}^D O\left( \frac{P^2_l}{\sum_{g \in \mathcal{S}_l} \frac{h(N_g)}{N_g}}\right) \cdot O\left(\frac{1}{G^2}\sum_{g=1}^G \frac{h(N_g)}{N_g}\right)  & , \text{for semi-strong dependence} \\
\sum_{l=1}^D O\left( \frac{P^2_l}{\sum_{g \in \mathcal{S}_l} \frac{1}{N_g}}\right) \cdot O\left(\frac{1}{G^2}\sum_{g=1}^G \frac{1}{N_g}\right)   & , \text{for weak dependence}
\end{cases} \\
& = \begin{cases}
O\left(1\right) &, \text{for strong dependence} \\
O\left( \sum_{l=1}^D \frac{P^2_l}{\sum_{g \in \mathcal{S}_l} \frac{h(N_g)}{N_g}} \cdot \frac{1}{G^2} \sum_{l=1}^D \sum_{g \in \mathcal{S}_l} \frac{h(N_g)}{N_g} \right) &, \text{for semi-strong dependence} \\
O\left(\sum_{l=1}^D \frac{P^2_l}{\sum_{g \in \mathcal{S}_l} \frac{1}{N_g}} \cdot \frac{1}{G^2} \sum_{l=1}^D \sum_{g \in \mathcal{S}_l} \frac{1}{N_g} \right) &, \text{for weak dependence}
\end{cases} 
\end{align*}

 So, 
\[
\frac{B_D}{\sqrt{2kD}} = \begin{cases}
O_P\left(\frac{1}{\sqrt{D}}\right) &, \text{for strong dependence} \\
O_P\left( \frac{1}{\sqrt{D}}\sum_{l=1}^D \frac{P^2_l}{\sum_{g \in \mathcal{S}_{l}} \frac{h(N_g)}{N_g}} \cdot \frac{1}{G^2} \sum_{l=1}^D \sum_{g \in \mathcal{S}_{l}} \frac{h(N_g)}{N_g} \right) &, \text{for semi-strong dependence} \\
O_P\left(\frac{1}{\sqrt{D}}\sum_{l=1}^D \frac{P^2_l}{\sum_{g \in \mathcal{S}_{l}} \frac{1}{N_g}} \cdot \frac{1}{G^2} \sum_{l=1}^D \sum_{g \in \mathcal{S}_{l}} \frac{1}{N_g} \right) &, \text{for weak dependence}
\end{cases} 
\]

By Assumption \ref{assump: condition for T_SB normality}, the following holds:
\[
\frac{1}{\sqrt{D}}\sum_{l=1}^D \frac{P^2_l}{\sum_{g \in \mathcal{S}_{l}} \frac{h(N_g)}{N_g}} \cdot \frac{1}{G^2} \sum_{l=1}^D \sum_{g \in \mathcal{S}_{l}} \frac{h(N_g)}{N_g}  \to 0, \;\text{ as } D\to\infty.
\]
The expression in the weak dependence case holds similarly by putting $h(N_g)=1.$ Hence, $\frac{B_D}{\sqrt{2kD}} = o_p(1)$, for any kind of dependence.

For the term $C_D,$ by the triangle inequality in the first line and the C-S inequality in the second line
\begin{align*}
    \mathbb{E}\left|C_D \right| & \leq 2\mathbb{E}\left[ \sum_{l=1}^D \left|(\Tilde{\beta}_{l} - \beta)' V_{G,l}^{-1} (\hat{\bar{\beta}} - \beta)\right|\right] \\
    & \leq 2 \sqrt{A_D^2} \sqrt{B_D^2}. 
\end{align*}
Since $ \frac{A_D - kD}{\sqrt{2kD}} \xrightarrow{d} N(0,1), $ and $\frac{B_D}{\sqrt{2kD}} =o_P(1),$ we obtain $\frac{C_D}{\sqrt{2kD}} =o_P(1).$
Therefore, 
\[
    \frac{T_{SB}^* - kD}{\sqrt{2kD}} = \underbrace{\frac{A_D - kD}{\sqrt{2kD}}}_{\xrightarrow{d} N(0,1)} + \underbrace{\frac{B_D}{\sqrt{2kD}}}_{\xrightarrow{p} 0} +\underbrace{\frac{C_D}{\sqrt{2kD}}}_{\xrightarrow{p} 0} 
\]
which implies $Z^* \xrightarrow{d} N(0,1)$, as  $P_l \to \infty, D \to \infty.$

\textit{Step 2:} In this step, we show that $Z - Z^* \xrightarrow{P} 0$ as $D\to \infty$ along with $\frac{D}{\min_l P_l}\to 0.$
Note that
\begin{align*}
\left|Z - Z^*\right| & = \left|\frac{T_{SB} - T_{SB}^*}{\sqrt{2kD}} \right|\\
&\leq \frac{1}{\sqrt{2kD}} \sum_{l=1}^D \left|(\Tilde{\beta}_{l} - \hat{\bar{\beta}})' (\Tilde{V}_{l}^{-1} - V_{G,l}^{-1}) (\Tilde{\beta}_{l} - \hat{\bar{\beta}}) \right|\\
&= \frac{1}{\sqrt{2kD}} \sum_{l=1}^D \left|(\Tilde{\beta}_{l} - \hat{\bar{\beta}})' V_{G,l}^{-1/2} \left[ V_{G,l}^{1/2} \Tilde{V}_{l}^{-1} V_{G,l}^{1/2} - I_k \right] V_{G,l}^{-1/2} (\Tilde{\beta}_{l} - \hat{\bar{\beta}})\right| \\
&\leq \frac{1}{\sqrt{2kD}} \sum_{l=1}^D (\Tilde{\beta}_{l} - \hat{\bar{\beta}})' V_{G,l}^{-1} (\Tilde{\beta}_{l} - \hat{\bar{\beta}}) \cdot \left\| V_{G,l}^{1/2} \Tilde{V}_{l}^{-1} V_{G,l}^{1/2} - I_k \right\|
\end{align*}

From the proof of Lemma \ref{lemma: V_l hat consistency}, we have 
\[
\frac{\left\Vert \Tilde{V}_{l}-V_{G,l}\right\Vert }{\left\Vert V_{G,l}\right\Vert }=O_{P}\left(\frac{1}{\sqrt{P_l}}\right), \quad \text{ uniformly in  } l.
\]  
Also, since $\frac{V_{G,l}}{\|V_{G,l}\|}$ is uniformly positive definite, by the continuity theorem for probability convergence, taking $g(W) = \left( \frac{V_{G,l}}{\|V_{G,l}\|} \right)^{1/2} W^{-1} \left( \frac{V_{G,l}}{\|V_{G,l}\|} \right)^{1/2}$ and setting $W = \frac{\Tilde{V}_{l}}{\|V_{G,l}\|}$, we have
    $$ \| V_{G,l}^{1/2} \Tilde{V}_{l}^{-1} V_{G,l}^{1/2} - I_k \| = O_{P}\left(\frac{1}{\sqrt{P_l}}\right), \quad \text{ uniformly in  } l. $$

Hence, using the norm inequality in the second line
\begin{align*}
\left|Z - Z^*\right| &\leq \frac{1}{\sqrt{2kD}} \sum_{l=1}^D (\Tilde{\beta}_{l} - \hat{\bar{\beta}})' V_{G,l}^{-1} (\Tilde{\beta}_{l} - \hat{\bar{\beta}}) \cdot  \left\Vert V_{G,l}^{1/2} \Tilde{V}_{l}^{-1} V_{G,l}^{1/2} - I_k \right\Vert \\
&\leq \max_l \left\Vert V_{G,l}^{1/2} \Tilde{V}_{l}^{-1} V_{G,l}^{1/2} - I_k \right\Vert\cdot \frac{T_{SB}^*}{\sqrt{2kD}} \\
&= O_P\left(\frac{1}{\min_l P_l}\right) \cdot O_P\left(\sqrt{D}\right) \\
&= O_P\left(\frac{D}{\min_l P_l}\right).
\end{align*}

Thus, if $\frac{D}{\min_l P_l}\to 0,$ $Z$ and $Z^*$ have the same asymptotic distribution and 
\[
Z \xrightarrow[H_0]{d} N(0,1).
\]

\hfill 
\qed

\textbf{Proof of Result \ref{result: V<VOLS, for weak dependence}}

Here, conditional on $X,$
\[
V=\frac{1}{G^{2}}\sum_{g=1}^G (X_{g}^{\prime}X_{g})^{-1}X_{g}^{\prime}\Omega_{g}X_{g}(X_{g}^{\prime}X_{g})^{-1}=\frac{1}{G^2}\sum_{g=1}^G  \frac{a_g X_g^\prime X_g}{a_g^2} = \frac{1}{G},
\]
and 
\[
\Sigma = \left(\sum_{g=1}^G X_{g}^{\prime}X_{g}\right)^{-1} \sum_{g=1}^G  X_{g}^{\prime}\Omega_{g}X_{g} \left(\sum_{g=1}^G X_{g}^{\prime}X_{g}\right)^{-1} = \frac{\sum_{g=1}^G  a_g^2}{\left( \sum_{g=1}^G  a_g \right)^2}.
\]
Since $\left( \sum_{g=1}^G  a_g \right)^2 \leq  G\sum_{g=1}^G a_g^2$ by C-S inequality, we have 
\[
\frac{1}{G} \leq \frac{\sum_{g=1}^G  a_g^2}{\left( \sum_{g=1}^G  a_g \right)^2}.
\]
Hence the result.

\hfill 
\qed

{\bf Proof of Result  \ref{result:efficiency} } 

For the given $\Omega_g,$ the variance of the POLS estimator, conditional on $X,$ is
\begin{align*}
    \Sigma & = \left(\sum_{g=1}^G X_g^\prime X_g\right)^{-1}\left( (a-b)\sum_{g=1}^G {X_g^\prime  X_g} + b\sum_{g=1}^G X^\prime_g \mathbbm{1} \mathbbm{1} ^\prime X_g \right) \left(\sum_{g=1}^G X_g^\prime X_g\right)^{-1} \\
    & = \frac{(a-b) \sum_{g=1}^G N_g + b  \sum_{g=1}^G N^2_g }{\left( \sum_{g=1}^G N_g \right)^2} \\
    & = \frac{a-b}{\sum_{g=1}^G N_g } + \frac{b\sum_{g=1}^G N^2_g}{\left( \sum_{g=1}^G N_g \right)^2}.
\end{align*}

On the other hand, the variance of $\hat{\bar{\beta}}$ is
\[
V=\frac{1}{G^{2}}\sum_{g=1}^G (X_{g}^{\prime}X_{g})^{-1}X_{g}^{\prime}\Omega_{g}X_{g}(X_{g}^{\prime}X_{g})^{-1} = \frac{1}{G^2}\sum_{g=1}^G  \frac{aN_g + b N_g (N_g-1)}{N_g^2} = \frac{(a-b)}{G^2} \sum_{g=1}^G \frac{1}{N_g} + \frac{b}{G}.
\]

So, we have $\operatorname{Var}(\hat{\beta}_{POLS}) = \Sigma = O\left(\frac{N_1^2 + G}{(N_1+G)^2}\right),$ and  $\operatorname{Var}(\hat{\bar{\beta}}) = V = O(\frac{1}{G}).$ 
Then, the efficiency of $\hat{\bar{\beta}}$ compared to $\hat{\beta}_{POLS}$, for sufficiently large $N_1$ and $G$, is 
\begin{eqnarray*}
\frac{\Sigma}{V} & = & O\left(\frac{GN_{1}^{2}+G^{2}}{(N_{1}+G)^{2}}\right)\\
 & = & \begin{cases}
O\left(\frac{G+\frac{G^{2}}{N_{1}^{2}}}{\left(1+\frac{G}{N_{1}}\right)^{2}}\right) & =O(G),\;\;\text{ if }\frac{N_{1}}{G}\rightarrow\infty \; \text{ or } N_1=O(G)\\
O\left(\frac{\frac{N_{1}^{2}}{G}+1}{\left(\frac{N_{1}}{G}+1\right)^{2}}\right) & =\begin{cases}
O\left(N_{1}^{1-\delta}\right), & \text{if }\frac{N_{1}}{G}=o(1)\text{ and }G=O\left(N_{1}^{1+\delta}\right),\text{ for \ensuremath{0<\delta<1}}\\
\rightarrow c, & \text{if }\frac{N_{1}}{G}=o(1)\text{ and }G=O_e\left(N_{1}^{2}\right), \text{ for } 1<c<\infty\\
\rightarrow1, & \text{if }\frac{N_{1}}{G}=o(1)\text{ and }G=O\left(N_{1}^{2+\delta_{1}}\right),\text{ for \ensuremath{\delta_{1}>0}}
\end{cases}
\end{cases}
\end{eqnarray*}

Hence, $\hat{\bar{\beta}}$ is asymptotically more efficient than $\hat{\beta}_{POLS}$ except when $\frac{N^2}{G}\to 0$. 

\hfill
\qed

\begin{lemma} \label{lemma:A_n=B_n+o_P(1) implies A_n = O_P(1)}
    Let $\left\{A_n\right\}$ be a sequence of finite-dimensional matrices such that $A_n=B_n+o_P(1),$ where $B_n$ is finite and uniformly positive definite over $n.$ Then, $A_n=O_P(1).$
\end{lemma}

\textbf{Proof of Lemma \ref{lemma:A_n=B_n+o_P(1) implies A_n = O_P(1)}}

Since $B_n$ is finite, there exists a constant \( M > 0 \) such that $ \|B_n\| \leq M, $ for all $n.$ $B_n$ is uniformly positive definite means that there exists a constant \( c > 0 \) such that for all nonzero vectors \( v \), $v^\top B_n v \geq c \|v\|^2,$ uniformly in \( n \). By the addition property of stochastic orders, $O_P(1)+o_P(1)=O_P(1)$, from Section 2.2 of \cite{van2000asymptotic}. To prove that \( A_n = O_P(1) \), it is enough to show that \( B_n = O_P(1) \), i.e., we show that every element \( (B_n)_{ij} \) of \( B_n \) is \( O_P(1) \). 

By Definition \ref{def:stochastically bounded uniformly}, a sequence of random variables \( Z_n \) is \( O_P(1), \) if for every $\varepsilon > 0$, there exist $M_\varepsilon > 0$ and ${L_\varepsilon} \in \mathbb{N}$ such that 
     \[ \mathbb{P}\left(\left|Z_n\right| > M_\varepsilon\right) < \varepsilon, \; \text{ for all } n \geq {L_\varepsilon}. \]

In this case, \( B_n \) is deterministic, so the probability statements become deterministic bounds. Since \( \|B_n\| \leq M \) for all \( n \), it follows that for any entry \( (B_n)_{ij} \),
\[
|(B_n)_{ij}| \leq \|B_n\| \leq M.
\]
This is a consequence of the fact that the absolute value of any entry of a matrix is bounded by its norm. Now, let \( \varepsilon > 0 \) be given, and choose \( M_\varepsilon = M,\; L_\varepsilon = 1. \) Then for all \( n \geq {L_\varepsilon}, \) and all \( i, j \),
\[
\mathbb{P}\left(|(B_n)_{ij}| > M_\varepsilon\right) = \mathbb{P}\left(|(B_n)_{ij}| > M\right) = 0 < \varepsilon.
\]
This satisfies the definition of \( O_P(1) \) for each element of \( B_n \). Since every entry of \( B_n \) is bounded by \( M \) (independent of \( n \)), we conclude that $ B_n = O_P(1). $ Consequently, we have 
\[
A_n=O_P(1).
\]
\hfill
\qed

{\bf Proof of Theorem \ref{thm: ols consistency under varying parameters}}

To show the consistency of $\hat{\beta}_{POLS}$, first note that under model \eqref{eq: model 2 for eta_g}
\[
    \hat{\beta}_{POLS}-\beta=\left(X^{\prime}X\right)^{-1}\sum_{g=1}^{G}X_{g}^{\prime}\eta_{g}= \left(X^{\prime}X\right)^{-1} \left[X'\epsilon + \sum_{g=1}^{G}X_{g}^{\prime}X_{g}u_{g}\right].
\]

Note that $ \left(\frac{X^{\prime}X}{\sum_{g=1}^{G}N_{g}}\right)^{-1} - C_{G}^{-1}=o_P(1),$ which follows from Assumption \ref{assump: Xg'Xg=O(Ng)}, where $C_{G}^{-1}$ is finite and uniformly positive definite over $G.$ Then, by Lemma \ref{lemma:A_n=B_n+o_P(1) implies A_n = O_P(1)},
\[
\left(\frac{X^{\prime}X}{\sum_{g=1}^{G}N_{g}}\right)^{-1}=O_P(1).
\]

Now, using Assumption \ref{assump:exogeneity}, note that 
\[
\mathbb{E}\left[\frac{X^{\prime}\epsilon}{\sum_{g=1}^{G}N_{g}}\right]=\mathbb{E}\mathbb{E}\left[\frac{X^{\prime}\epsilon}{\sum_{g=1}^{G}N_{g}} \bigg|  X\right]=\frac{1}{\sum_{g=1}^{G}N_{g}}\mathbb{E}\left[X^{\prime}\mathbb{E}\left[\epsilon \big|  X\right]\right]=0.
\]
 Also, using the independence across clusters by Assumption \ref{assump:error independence}, the variance-covariance matrix of the vector  $\frac{X^{\prime}\epsilon}{\sum_{g=1}^{G}N_{g}}$ is 
\[
\mathbb{E}\left[\operatorname{Var}\left(\frac{X^{\prime}\epsilon}{\sum_{g=1}^{G}N_{g}} \bigg|  X\right)\right] = \frac{\sum_{g=1}^{G}\mathbb{E}\left[X_{g}^{\prime}\Omega_{g}X_{g}\right]}{\left(\sum_{g=1}^{G}N_{g}\right)^{2}}.
\]

Since this is a positive semi-definite matrix,
\[
\left\Vert \frac{\sum_{g=1}^{G}\mathbb{E}\left[X_{g}^{\prime}\Omega_{g}X_{g}\right]}{\left(\sum_{g=1}^{G}N_{g}\right)^{2}} \right\Vert_e = \lambda_{\max}\left(\frac{\sum_{g=1}^{G} \mathbb{E}\left[X_{g}^{\prime}\Omega_{g}X_{g}\right]}{\left(\sum_{g=1}^{G}N_{g}\right)^{2}}\right).
\]
Then, by Result \ref{result:trace inequality} and Jensen's inequality, we have 
\begin{align*}
    \lambda_{\max}\left(\frac{\sum_{g=1}^{G} \mathbb{E}\left[X_{g}^{\prime}\Omega_{g}X_{g}\right]}{\left(\sum_{g=1}^{G}N_{g}\right)^{2}}\right) & \leq   \frac{\sum_{g=1}^{G} \lambda_{\max}\left(\mathbb{E}\left[X_{g}^{\prime}\Omega_{g}X_{g}\right]\right)}{\left(\sum_{g=1}^{G}N_{g}\right)^{2}}\\
    &  \leq   \frac{\sum_{g=1}^{G} \mathbb{E}\left[\lambda_{\max}\left(X_{g}^{\prime}\Omega_{g}X_{g}\right)\right]}{\left(\sum_{g=1}^{G}N_{g}\right)^{2}}\\
    &  \leq   \frac{\sum_{g=1}^{G} \mathbb{E}\left[\lambda_{\max}\left(X_{g}^{\prime}X_{g}\right) \lambda_{\max}\left(\Omega_{g}\right)\right]}{\left(\sum_{g=1}^{G}N_{g}\right)^{2}}.
\end{align*}

Using the fact that $\mathbb{E}\left[\lambda_{\max}\left(X_{g}^{\prime}X_{g}\right)\right] = O\left(N_g\right)$ uniformly in $g$ by Assumption \ref{assump: Xg'Xg=O(Ng)}, it follows from Assumption \ref{assump:error independence} and Lemma \ref{lemma:Exact order summability} that 
\begin{align*}
\left\Vert \operatorname{Var}\left(\frac{X^{\prime}\epsilon}{\sum_{g=1}^{G}N_{g}}\right) \right\Vert= & \begin{cases}
O\left(\frac{\sum_{g=1}^{G}N_{g}^{2}}{\left(\sum_{g=1}^{G}N_{g}\right)^{2}}\right), & \text{for strong dependence}\\
O\left(\frac{\sum_{g=1}^{G}N_{g}h(N_{g})}{\left(\sum_{g=1}^{G}N_{g}\right)^{2}}\right), & \text{for semi-strong dependence}\\
O\left(\frac{1}{\sum_{g=1}^{G}N_{g}}\right), & \text{for weak dependence}
\end{cases}
\end{align*}

Now, for the terms involving $u_g$, note that by Assumptions \ref{assump: epsilon_g, X_g and u_g independent} and \ref{assump: distribution of u_g} 
\[
\mathbb{E}\left[\frac{1}{\sum_{g=1}^G N_g} \sum_{g=1}^G X_g'X_g u_g\right] = 0,
\]

and the variance-covariance matrix of $\frac{1}{\sum_{g=1}^G N_g} \sum_{g=1}^G X_g'X_g u_g$ is
\[
\frac{1}{\left(\sum_{g=1}^G N_g\right)^2}\sum_{g=1}^{G}\mathbb{E}\left[X_{g}^{\prime}X_{g}\Delta X_{g}^{\prime}X_{g}\right].
\]

Since this is a positive semi-definite matrix,
\[
\left\Vert \frac{1}{\left(\sum_{g=1}^G N_g\right)^2}\sum_{g=1}^{G}\mathbb{E}\left[X_{g}^{\prime}X_{g}\Delta X_{g}^{\prime}X_{g}\right] \right\Vert_e = \lambda_{\max}\left(\frac{1}{\left(\sum_{g=1}^{G}N_{g}\right)^{2}} \sum_{g=1}^{G} \mathbb{E}\left[X_{g}^{\prime}X_{g}\Delta X_{g}^{\prime}X_{g}\right]\right).
\]
Then, by Result \ref{result:trace inequality} and Jensen's inequality, we have 
\begin{align*}
    \lambda_{\max}\left(\frac{1}{\left(\sum_{g=1}^{G}N_{g}\right)^{2}} \sum_{g=1}^{G} \mathbb{E}\left[X_{g}^{\prime}X_{g}\Delta X_{g}^{\prime}X_{g}\right]\right) & \leq   \frac{\sum_{g=1}^{G} \lambda_{\max}\left(\mathbb{E}\left[X_{g}^{\prime}X_{g}\Delta X_{g}^{\prime}X_{g}\right]\right)}{\left(\sum_{g=1}^{G}N_{g}\right)^{2}}\\
    &  \leq   \frac{\sum_{g=1}^{G} \mathbb{E}\left[\lambda_{\max}\left(X_{g}^{\prime}X_{g}\Delta X_{g}^{\prime}X_{g}\right)\right]}{\left(\sum_{g=1}^{G}N_{g}\right)^{2}}\\
    &  \leq   \frac{\sum_{g=1}^{G} \mathbb{E}\left[\lambda_{\max}\left((X_{g}^{\prime}X_{g})^2\right) \lambda_{\max}\left(\Delta\right)\right]}{\left(\sum_{g=1}^{G}N_{g}\right)^{2}}.
\end{align*}

Using the fact that $\mathbb{E}\left[\lambda_{\max}\left((X_{g}^{\prime}X_{g})^2\right)\right] = O\left(N_g^2\right)$ uniformly in $g$ by Assumption \ref{assump: Xg'Xg=O(Ng)}, it follows from Assumption \ref{assump: distribution of u_g} and Lemma \ref{lemma:Exact order summability} that 
\begin{align*}
\left\Vert \operatorname{Var}\left(\frac{1}{\sum_{g=1}^G N_g} \sum_{g=1}^G X_g'X_g u_g\right) \right\Vert= O\left(\frac{\sum_{g=1}^{G}N_{g}^{2}}{\left(\sum_{g=1}^{G}N_{g}\right)^{2}}\right).
\end{align*}

Since $u_g$ and $\epsilon_g$ are independently distributed by Assumption \ref{assump: epsilon_g, X_g and u_g independent}, it follows that 
\[
   \left\Vert \operatorname{Var}\left(\frac{X'\epsilon + \sum_{g=1}^{G}X_{g}^{\prime}X_{g}u_{g}}{\sum_{g=1}^G N_g}\right)\right\Vert = O\left(\frac{\sum_{g=1}^{G}N_{g}^{2}}{\left(\sum_{g=1}^{G}N_{g}\right)^{2}}\right),
\]
for any kind of dependence of $\epsilon_g.$

Observe that
$\frac{\sum_{g=1}^{G}N_{g}^{2}}{\left(\sum_{g=1}^{G}N_{g}\right)^{2}}\to0,$ if $ \frac{\underset{g:g=1(1)G}{max}N_g}{\sum_{g=1}^{G}N_{g}}\to0.$ In particular, for nearly balanced case, assuming $N_g\simeq N,$  $\frac{\sum_{g=1}^{G}N_{g}^{2}}{\left(\sum_{g=1}^{G}N_{g}\right)^{2}}\simeq \frac{GN^2}{G^2 N^2}=\frac{1}{G}\to 0,$ as $G\to \infty.$ Therefore, 
\[
    \frac{X'\epsilon + \sum_{g=1}^{G}X_{g}^{\prime}X_{g}u_{g}}{\sum_{g=1}^G N_g} = o_P(1),
\]
if (i) the clusters are nearly balanced, or (ii) the clusters are unbalanced with $\frac{\max_g N_g}{\sum_{g=1}^G N_g}\to 0,$ as $G\to\infty.$ Since $\left(\frac{X^{\prime}X}{\sum_{g=1}^{G}N_{g}}\right)^{-1}=O_P(1),$ we obtain 
\[
\hat{\beta}_{POLS}-\beta = \left(\frac{X^{\prime}X}{\sum_{g=1}^{G}N_{g}}\right)^{-1}  \left[\frac{X'\epsilon + \sum_{g=1}^{G}X_{g}^{\prime}X_{g}u_{g}}{\sum_{g=1}^G N_g}\right] = O_P(1) \cdot o_P(1) = o_P(1),
\] 
for (i) the nearly balanced clusters or (ii) the unbalanced clusters with $\frac{\max_g N_g}{\sum_{g=1}^G N_g}\to 0,$ as $G\to\infty$. This proves the theorem.

\hfill 
\qed

{\bf Proof of Theorem \ref{thm: ols inconsistency under varying clusters}}

First note that under model \eqref{eq: model 2 for eta_g}
\[
    \hat{\beta}_{POLS}-\beta=\left(X^{\prime}X\right)^{-1}\sum_{g=1}^{G}X_{g}^{\prime}\eta_{g}= \left(X^{\prime}X\right)^{-1} \left[X'\epsilon + \sum_{g=1}^{G}X_{g}^{\prime}X_{g}u_{g}\right].
\]

Since $ \left(\frac{X^{\prime}X}{\sum_{g=1}^{G}N_{g}}\right)^{-1} - C_{G}^{-1}=o_P(1)$  follows from Assumption \ref{assump: Xg'Xg=O(Ng)}, where $C_{G}^{-1}$ is finite and uniformly positive definite over $G,$ we have by Lemma \ref{lemma:A_n=B_n+o_P(1) implies A_n = O_P(1)},
\[
\left(\frac{X^{\prime}X}{\sum_{g=1}^{G}N_{g}}\right)^{-1}=O_P(1).
\]

Let us assume that $Z_{G} = \sum_{g=1}^G X_g'\eta_g,$ and $Z_G^* = \frac{Z_G}{\sum_{g=1}^G N_g}.$
Since $\epsilon_g$ and $u_g$ are independent by Assumption \ref{assump: epsilon_g, X_g and u_g independent}, 
\[
    \operatorname{Var}\left(X'\epsilon + \sum_{g=1}^{G}X_{g}^{\prime}X_{g}u_{g}\right) = \operatorname{Var}\left(X'\epsilon \right) + \operatorname{Var}\left( \sum_{g=1}^{G}X_{g}^{\prime}X_{g}u_{g}\right) \geq \operatorname{Var}\left( \sum_{g=1}^{G}X_{g}^{\prime}X_{g}u_{g}\right).
\]

Now, by Assumption \ref{assump: distribution of u_g}
\[
\operatorname{Var}\left( \sum_{g=1}^{G}X_{g}^{\prime}X_{g}u_{g}\right) =  \sum_{g=1}^{G}X_{g}^{\prime}X_{g}\Delta X_{g}^{\prime}X_{g}\,.
\]
Note that 
\begin{align*}
    \lambda_{\max}\left(\sum_{g=1}^{G}X_{g}^{\prime}X_{g}\Delta X_{g}^{\prime}X_{g}\right) & \leq \lambda_{\max}\left(\sum_{g=1}^{G}(X_{g}^{\prime}X_{g})^2 \right) \lambda_{\max}\left(\Delta \right) \quad \text{ and } \\
    \lambda_{\min}\left(\sum_{g=1}^{G}X_{g}^{\prime}X_{g}\Delta X_{g}^{\prime}X_{g}\right) & \geq \lambda_{\min}\left(\sum_{g=1}^{G}(X_{g}^{\prime}X_{g})^2 \right) \lambda_{\min}\left(\Delta \right).
\end{align*}
Using Assumptions \ref{assump: distribution of u_g}  and \ref{assump: Xg'Xg=O(Ng)}, we obtain
\[
\left\Vert \operatorname{Var}\left[Z^*_{G}\right] \right\Vert=O_{e}\left(\frac{\sum_{g=1}^{G}N_{g}^{2}}{\left(\sum_{g=1}^{G}N_{g}\right)^{2}}\right).
\]
Hence, the second moment of  the norm of $Z^*_{G}$ is bounded away from zero, unless $\frac{\sum_{g=1}^{G}N_{g}^{2}}{\left(\sum_{g=1}^{G}N_{g}\right)^{2}}\to 0$.

The $4^{th}$ moment of the norm of $Z_{G}$ is 
\begin{align*}
    \mathbb{E}\left[\left\Vert Z_{G}\right\Vert^{4}\right] & = \mathbb{E}\left[\left(\sum_{g=1}^{G}\eta_{g}^{\prime}X_{g}\sum_{g_{1}=1}^G X_{g_{1}}^{\prime}\eta_{g_{1}}\right)^{2}\right] \\
    & = \mathbb{E}\left[\left(\sum_{g,g_{1}=1}^G\left(X_{g}u_{g}+\epsilon_{g}\right)^{\prime}X_{g}X_{g_{1}}^{\prime}\left(X_{g_{1}}u_{g_{1}}+\epsilon_{g_{1}}\right)\right)^{2}\right] \\
    & = \mathbb{E}\left[\left(T_{1}+T_{2}+2 \, T_{3}\right)^{2}\right]\\
    & \leq 3\;\mathbb{E}\left[T_{1}^{2}+T_{2}^{2}+4\,T_{3}^{2}\right], \;\text{ by the C-S inequality, }
\end{align*}

where $T_{1} = \sum_{g,g_{1}=1}^G \epsilon_{g}^{\prime}X_{g}X_{g_{1}}^{\prime}\epsilon_{g_{1}},$ \; 
$T_{2}=\sum_{g,g_{1}=1}^G u_{g}^{\prime}X_{g}^{\prime}X_{g}X_{g_{1}}^{\prime}X_{g_{1}}u_{g_{1}},$\; and  $T_{3}=\sum_{g,g_{1}=1}^G\epsilon_{g}^{\prime}X_{g}X_{g_{1}}^{\prime}X_{g_{1}}u_{g_{1}}.$

Considering the first term, note that
\[
    \mathbb{E}\left[T_{1}^{2}\right] = \mathbb{E}\left[\left(\sum_{g,g_{1}=1}^G \epsilon_{g}^{\prime}X_{g}X_{g_{1}}^{\prime}\epsilon_{g_{1}}\right)^{2}\right] = \sum_{g,g_1,h,h_1=1}^G 
\mathbb{E}\!\left[
\epsilon_g' X_g X_{g_1}' \epsilon_{g_1}
\;\epsilon_h' X_h X_{h_1}' \epsilon_{h_1}
\right].
\]

Since the errors are independent across clusters by Assumption \ref{assump:error independence}, the expectation of the
product vanishes unless $(g,h)=(g_1,h_1)$ or $(g,h)=(h_1,g_1)$. Thus, for some $c>0,$
\begin{align*}
    \mathbb{E}\left[T_1^2\right] & \leq 
c\sum_{g=1}^G \sum_{g_1=1}^G
\mathbb{E}\left[ \left(\epsilon_g' X_g X_{g_1}' \epsilon_{g_1}\right)^2 \right] \\ 
& \leq c\left(\sum_{g=1}^G
\mathbb{E}\left[ \epsilon_g' X_g X_{g}' \epsilon_{g} \right]\right)^2 ,\; \text{ by the C-S inequality} \\
& = c\left(\sum_{g=1}^G
\mathbb{E} \operatorname{tr}\left[X_g' \Omega_g X_g \right]\right)^2 
\end{align*}

Since $\mathbb{E} \operatorname{tr}\left[X_g' \Omega_g X_g \right] \leq \mathbb{E} \operatorname{tr}\left[X_g' X_g\right]\lambda_{\max}(\Omega_g)$ holds by Result \ref{result:trace inequality}, we obtain from Assumptions  \ref{assump:error independence}, \ref{assump: Xg'Xg=O(Ng)} and Lemma \ref{lemma:Exact order summability},
\[
  \mathbb{E}\left[T_1^2\right]  = \begin{cases}
    O\left(\left(\sum_{g=1}^G N_g^2\right)^2\right), & \text{ for strong dependence } \\
     O\left(\left(\sum_{g=1}^G N_g h(N_g)\right)^2\right), & \text{ for semi-strong dependence } \\
     O\left(\left(\sum_{g=1}^G N_g\right)^2\right), & \text{ for weak dependence }
\end{cases}
\]

For the term $T_2,$ note that
\[
\mathbb{E}\left[T_{2}^{2}\right]
=
\sum_{g,g_{1},h,h_{1}=1}^{G}
\mathbb{E}\!\left[
u_{g}^{\prime}X_{g}^{\prime}X_{g}
X_{g_{1}}^{\prime}X_{g_{1}}u_{g_{1}}
\;
u_{h}^{\prime}X_{h}^{\prime}X_{h}
X_{h_{1}}^{\prime}X_{h_{1}}u_{h_{1}}
\right].
\]

Since $u_g$'s are independent across clusters by Assumption \ref{assump: epsilon_g, X_g and u_g independent}, the expectation of the
product vanishes unless $(g,h)=(g_1,h_1)$ or $(g,h)=(h_1,g_1)$. Thus, for some $c_1>0,$
\begin{align*}
    \mathbb{E}\left[T_{2}^{2}\right] & \leq c_1\sum_{g=1}^{G}\sum_{g_{1}=1}^{G}
\mathbb{E}\left[\left(u_{g}^{\prime}X_{g}^{\prime}X_{g}X_{g_{1}}^{\prime}X_{g_{1}}u_{g_{1}}\right)^{2}\right] \\
& \leq c_1\left(\sum_{g=1}^G
\mathbb{E}\left[u_{g}^{\prime}X_{g}^{\prime}X_{g}X_{g}^{\prime}X_{g}u_{g}\right]\right)^2 ,\; \text{ by the C-S inequality} \\
& = c_1\left(\sum_{g=1}^G
\mathbb{E} \operatorname{tr}\left[X_g' X_g \Delta X_g' X_g \right]\right)^2 \\
& \leq c_1\left(\sum_{g=1}^G
\mathbb{E} \operatorname{tr}\left[(X_g' X_g)^2 \right] \lambda_{\max}(\Delta)\right)^2 ,\; \text{ by Result \ref{result:trace inequality}}
\end{align*}

Using Assumptions \ref{assump: distribution of u_g}, \ref{assump: Xg'Xg=O(Ng)} and Lemma \ref{lemma:Exact order summability}, we obtain 
\[
\mathbb{E}\left[T_2^2\right] = O\left(\left(\sum_{g=1}^{G}N_g^{2}\right)^{2}
\right).
\]

For $T_3,$ note that
\[
\mathbb{E}\!\left[T_3^2\right]
=
\sum_{g,g_1,h,h_1=1}^G
\mathbb{E}\!\left[
\epsilon_g' X_g X_{g_1}' X_{g_1} u_{g_1}
\; \epsilon_h' X_h X_{h_1}' X_{h_1} u_{h_1}
\right].
\]

Since $\epsilon_g$ and $u_g$'s are both independent across clusters by Assumptions \ref{assump:error independence} and \ref{assump: epsilon_g, X_g and u_g independent}, the expectation of the product vanishes unless $(g,h)=(g_1,h_1)$ or $(g,h)=(h_1,g_1)$. Thus, for some $c_2>0,$
\begin{align*}
\mathbb{E}\left[T_{3}^{2}\right] & \leq c_2\sum_{g=1}^G \sum_{g_1=1}^G
\mathbb{E}\left[\left( \epsilon_g' X_g X_{g_1}' X_{g_1} u_{g_1}\right)^2\right] \\
& \leq c_2\sum_{g=1}^G \sum_{g_1=1}^G
\mathbb{E}\left[ \epsilon_g' X_g X_g' \epsilon_g u'_{g_1} X_{g_1}' X_{g_1} X_{g_1}' X_{g_1} u_{g_1}\right],\; \text{ by the C-S inequality} \\
& = c_2\sum_{g=1}^G \sum_{g_1=1}^G
\mathbb{E}\left[ \epsilon_g' X_g X_g' \epsilon_g\right] \mathbb{E}\left[ u'_{g_1} X_{g_1}' X_{g_1} X_{g_1}' X_{g_1} u_{g_1}\right] , \; \text{ by Assumption \ref{assump: epsilon_g, X_g and u_g independent}}\\
& = c_2\sum_{g=1}^G 
\mathbb{E}\operatorname{tr}\left[ X_g'\Omega_g X_g\right] \sum_{g_1=1}^G \mathbb{E}\operatorname{tr}\left[ X_{g_1}' X_{g_1} \Delta X_{g_1}' X_{g_1} \right].
\end{align*}

Since $\mathbb{E} \operatorname{tr}\left[X_g' \Omega_g X_g \right] \leq \mathbb{E} \operatorname{tr}\left[X_g' X_g\right]\lambda_{\max}(\Omega_g)$ holds by Result \ref{result:trace inequality}, we obtain from Assumptions  \ref{assump:error independence}, \ref{assump: distribution of u_g}, \ref{assump: Xg'Xg=O(Ng)} and Lemma \ref{lemma:Exact order summability},
\[
\mathbb{E}\left[T_3^2\right] = O\left(\left(\sum_{g=1}^{G}N_g^{2}\right)^{2}
\right).
\]

Combining all these, we have 
\[
\mathbb{E}\left[\left\Vert Z_{G}\right\Vert^{4}\right] = O\left(\left(\sum_{g=1}^{G}N_g^{2}\right)^{2}
\right).
\]

If we assume that $Z_G^* = \frac{Z_G}{\sum_{g=1}^G N_g},$ then 
\[
  \mathbb{E}\left[\left\Vert Z^*_{G}\right\Vert^{4}\right]\; = \;O\left( \frac{\left(\sum_{g=1}^{G}N_g^{2}\right)^{2}}{\left(\sum_{g=1}^{G}N_g\right)^{4}}\right) \; = \; O(1).
\]
Therefore, the fourth moment of the norm of $Z^*_{G}$ is bounded. We have also shown that the second moment of the norm of $Z^*_{G}$ is bounded away from zero unless $\frac{\sum_{g=1}^{G}N_{g}^{2}}{\left(\sum_{g=1}^{G}N_{g}\right)^{2}}\to 0$, that is, $Z^*_{G}$ does not converge to $0$ in the mean square. Since the boundedness of the $(2+\delta)^{th}$ moment implies uniform integrability of $\left\{ \left\Vert Z^*_{G}\right\Vert^{2}:G\in\mathbb{N}\right\},$ the family is uniformly integrable. So, we can claim that $Z^*_{G}$
does not converge to $0$ in probability, since probability
convergence and uniform integrability in $\mathscr{L}^{p}$ implies
convergence at the $p^{th}$ moment. Hence, 
\[
    Z^*_{G}=\frac{\sum_{g=1}^G X_g'\eta_g}{\sum_{g=1}^{G}N_{g}}\neq o_P(1),
\]
which deduces the fact that $\hat{\beta}_{POLS}$ is not consistent for $\beta,$ for strong dependence with unbalanced clusters unless $\frac{\sum_{g=1}^{G}N_{g}^{2}}{\left(\sum_{g=1}^{G}N_{g}\right)^{2}}\to 0.$

\hfill 
\qed

\bibliography{references}

\end{document}